\documentclass{article}

\usepackage[shortlabels]{enumitem}
\usepackage{amssymb}
\usepackage{amsmath}
\usepackage{amsthm}
\usepackage{amsopn}
\usepackage{graphicx}
\usepackage{bbm}
\usepackage{mathdots}
\usepackage[export]{adjustbox}
\usepackage{hyperref}
\usepackage{float}
\usepackage{nicefrac}
\usepackage{stmaryrd}

\hypersetup{
    colorlinks = true,
    citecolor=black,
    filecolor=black,
    linkcolor=blue,
    urlcolor=black,
    linktoc=all,     
}

\graphicspath{{images/}}

\newtheorem{theorem}{Theorem}[section]
\newtheorem{definition}[theorem]{Definition}
\newtheorem{lemma}[theorem]{Lemma}
\newtheorem{remark}[theorem]{Remark}
\newtheorem{example}[theorem]{Example}
\newtheorem{corollary}[theorem]{Corollary}
\newtheorem{proposition}[theorem]{Proposition}

\newcommand{\CC}{\mathbb C}
\newcommand{\RR}{\mathbb R}
\newcommand{\NN}{\mathbb N}
\newcommand{\ZZ}{\mathbb Z}
\newcommand{\mca}{\mathcal A}
\newcommand{\mcb}{\mathcal B}
\newcommand{\mcc}{\mathcal C}

\newcommand{\mch}{\mathcal H}
\newcommand{\mci}{\mathcal I}
\newcommand{\mcj}{\mathcal J}
\newcommand{\mck}{\mathcal K}

\newcommand{\mfg}{\mathfrak g}
\newcommand{\mco}{\mathcal O}
\newcommand{\mcu}{\mathcal U}
\newcommand{\mcr}{\mathcal R}
\newcommand{\mcm}{\mathcal M}
\newcommand{\mcn}{\mathcal N}
\newcommand{\mcx}{\mathcal X}
\newcommand{\TT}{\mathbb T}

\newcommand{\mcl}{\mathcal L}

\newcommand{\mft}{\mathfrak t}
\newcommand{\mfm}{\mathfrak m}
\newcommand{\sto}{\xrightarrow{\ \text{s}\ }}

\DeclareMathOperator{\dist}{dist}

\DeclareMathOperator{\supp}{supp}
\DeclareMathOperator{\std}{std}

\DeclareMathOperator{\Tr}{Tr}
\DeclareMathOperator{\tr}{Tr}

\DeclareMathOperator{\MS}{MS}

\DeclareMathOperator{\End}{End}

\DeclareMathOperator{\Sp}{Span}
\DeclareMathOperator{\op}{op}
\DeclareMathOperator{\Id}{Id}

\DeclareMathOperator{\essran}{ess\,ran}
\DeclareMathOperator{\essinf}{ess\,inf}
\DeclareMathOperator{\esssup}{ess\,sup}
\DeclareMathOperator{\esssupp}{ess\,supp}
\DeclareMathOperator{\inv}{inv}
\DeclareMathOperator{\ess}{ess}
\author{Ood Shabtai}
\title{Localization of quantum systems at Liouville tori}
\date{}
\begin{document}
\maketitle
\begin{abstract} 
We consider a collection of pairwise commuting quantum observables in the setting of Berezin--Toeplitz quantization of a closed K\"{a}hler manifold and assume that the Arnold--Liouville theorem applies to their principal symbols. We use joint eigensections of these observables to define isometric embeddings of the quantum spaces into $L^2(\Lambda_{a_0})$, where $\Lambda_{a_0}$ is a fixed Liouville torus. These embeddings allow a broad class of quantum observables, including some defined by discontinuous functions, to be realized as sequences of operators on $L^2(\Lambda_{a_0})$ that converge strongly to multiplication operators. We discuss the spectral implications of this convergence and give applications to contractions of Lie algebra representations and to pairs of spectral projections of quantum observables.
\end{abstract}
\tableofcontents
\section{Introduction}
Let $(M, \omega)$ be a closed connected K\"{a}hler manifold of real dimension $2n$, equipped with a \textit{prequantum line bundle}\footnote{A prequantum line bundle on $(M, \omega)$ is a holomorphic Hermitian line bundle such that the curvature of its Chern connection equals $-i\omega$.} $L \to M$. Let $\mch_k$ be the space of holomorphic sections of $L^{\otimes k} \to M$. Note that $\mch_k$ is a finite-dimensional complex vector space, equipped with a natural inner product\footnote{The inner product on $\mch_k$ is obtained by integration of the fiberwise Hermitian product of $L^{\otimes k}$ with respect to the Liouville measure $\mu = \frac{|\omega^{\wedge n}|}{n!}$.}. 
The \textit{Berezin-Toeplitz operator} $T_k(f) \in \End(\mch_k)$ associated with a smooth function\footnote{We denote $C^\infty(M) = C^\infty(M, \RR)$ and $C(M) = C(M, \RR)$ (though at times we will also consider complex-valued functions).} $f \in C^\infty(M)$ is given by
\begin{equation}\label{toeplitz_def} T_k(f) = \Pi_k \mcm_{f},\end{equation}
where $\mcm_{f}$ is the operator of multiplication by $f$ and $\Pi_k : L^2(M, L^{\otimes k}) \to \mch_k$ is the orthogonal ("Bergman") projection\footnote{The Hilbert space $L^2(M, L^{\otimes k})$ is the completion of $C^\infty(M, L^{\otimes k})$ with respect to the inner product obtained via integration of the fiberwise Hermitian product of $L^{\otimes k}$.}. The operator $T_k(f)$ is characterized by the equality $$\langle T_k(f) s, s' \rangle = \langle f s, s' \rangle\ \text{for all } s, s' \in \mch_k.$$More broadly, it is useful to define Berezin-Toeplitz operators as operators of the form\footnote{Here $\mco(k^{-\infty})$ refers to a sequence $R_k \in \End(\mch_k)$ which satisfies that for any $N > 0$ there exists $C_N > 0$ such that for all sufficiently large $k$, $\Vert R_k \Vert_{\op} < C_N k^{-N}$.} $$T_k(f(\cdot, k)) + \mco(k^{-\infty}) : \mch_k \to \mch_k,$$ where $f(\cdot, k)\in C^\infty(M)$ is a sequence which admits an asymptotic expansion $$f(\cdot, k) \sim \sum_{j \ge 0} f_j k^{-j}$$
in the $C^\infty$ topology. The function $f_0 \in C^\infty(M)$ is called the \textit{principal symbol} of the operator, and it holds that \begin{equation}\label{norm_corr}\Vert T_k(f(\cdot, k))\Vert_{\op} = \max_M |f_0| + \mco(k^{-1}).\end{equation}

The mapping $T_k : C^\infty(M) \to \End(\mch_k)$ is the \textit{Berezin-Toeplitz quantization} (see for example \cite{bms, bpu, charles, mm, mm1, lefloch}), in the context of which $\mch_k$ is viewed as the space of pure states of a quantum mechanical system, the operator $T_k(f)$ is viewed as the quantum counterpart of the classical observable $f \in C^\infty(M)$, and $k \to \infty$ is the semiclassical limit. A fundamental property of Berezin-Toeplitz quantization is the \textit{correspondence principle} \begin{equation}\label{corr_princ}\big{\Vert} [ T_k(f), T_k(g)] + i k^{-1}T_k(\{f, g\})\big{\Vert}_{\op} = \mco(k^{-2}),\ f,g \in C^\infty(M),\end{equation}
where $[A,B] = AB - BA$ is the commutator of linear operators, $\Vert A \Vert_{\op}$ is the operator norm, and $\{f, g\}$ is the \textit{Poisson bracket} of $f,g$. The correspondence principle provides a relation between quantum and classical mechanics which is essential to the notion of quantization.

In the present paper, we consider an $n$-tuple of self-adjoint Berezin-Toeplitz operators $$T_k^1, ..., T_k^n \in \End(\mch_k),$$ and we assume that for every sufficiently large $k$, these operators are pairwise commuting. In particular, this requires\footnote{In light of the correspondence principle (\ref{corr_princ}), combined with the \textit{norm correspondence} (\ref{norm_corr}).} their principal symbols $f^1_{0}, ..., f^n_{0}$ to be pairwise Poisson commuting. Let $$f_0 = (f^1_{0}, ..., f^n_{0}) : M \to \RR^n$$
be the joint principal symbol of $T^1_{k}, ..., T^n_{k}$, and assume that $$a_0 = (a^1_{0}, ..., a^n_{0}) \in  f_0(M)$$ is a regular value of $f_0$ such that the level set $$\Lambda_{a_0} =  f_0^{-1}(a_0)$$ is connected. Then the Arnold-Liouville theorem (\cite{zung}) applies, and $\Lambda_{a_0}$ is termed a \textit{Liouville torus} (as it is diffeomorphic to the torus $\TT^n = \RR^n/\ZZ^n$).
The implications for the joint spectrum
$\sigma_k \subset \mathbb R^n$ of $T^1_{k},..., T^n_{k}$ are significant. Namely, the elements of $\sigma_k$ near $a_0$ are described by so-called \textit{Bohr-Sommerfeld conditions} (\cite{charles_bohr}, Theorem 3.1, \cite{charles_symb}, Theorem 2.1)\footnote{See also \cite{ac, vn, cdv} for earlier results in the setting of pseudodifferential quantization.}: they belong, up to corrections of order $\mco(k^{-\infty})$, to a \textit{deformed lattice} which admits a precise definition using geometric data. Consequently, inside neighborhoods of $a_0$ whose diameters tend to $0$, the joint spectrum is well approximated via a fixed \textit{genuine lattice}, rescaled by a factor of $\frac 1 k$. Moreover, the corresponding joint eigensections concentrate about $\Lambda_{a_0}$, and may be identified (through a linear isometry) with standard exponential functions\footnote{Here, $m \cdot \theta = m^1 \theta^1 + ... + m^n \theta^n \mod \ZZ.$} \begin{equation}\label{exp_fns}e_m(\theta) = e^{2\pi i m \cdot \theta},\ \theta \in \TT^n,\  m \in \ZZ^n.\end{equation} %
This identification provides a non-trivial link between classical function theory on $\TT^n$ and the Berezin-Toeplitz quantization of $M$. The purpose of the present work is to elaborate on this link and some of its potential applications.%

We conclude the introduction by noting that the results of the present paper substantially rely on the foundational works \cite{charles, charles_bohr} of L. Charles, in which the setting that we consider here was studied in detail. We refer the reader to \cite{charles, charles_bohr} for further information concerning Bohr-Sommerfeld conditions in Berezin-Toeplitz quantization (and the tools used in their study).
\section{Main results}\label{main_res_sect}
In what follows, we consider isometric embeddings of $\mch_k$ into $L^2(\TT^n)$, such that joint eigensections of $T^1_{k},...,T^n_{k}$ corresponding to joint eigenvalues lying in certain shrinking neighborhoods of $a_0$ are mapped to exponential functions (\ref{exp_fns}).
To this end, we first approximate the joint spectrum $\sigma_k$ in these neighborhoods via a fixed lattice which is rescaled by $\frac 1 k$.

Recall that by the Arnold-Liouville theorem, there exists a neighborhood $U$ of $\Lambda_{a_0}\subset M$, a neighborhood $D$ of $0 \in \RR^n$, and a symplectomorphism
\begin{equation}\label{al_sympl} \Psi : \left(U, \omega|_U\right) \to (D\times \TT^n,\ \omega_{\std}),\ \Psi(\Lambda_{a_0}) = \{0\} \times 
\TT^n, \end{equation}
where $\omega_{\std}$ is the standard symplectic form\footnote{That is, $\omega_{\std} = \sum_{j=1}^n dy^j \wedge d\theta^j$, where $(y^1, ..., y^n) \in D$, $(\theta^1, ..., \theta^n) \in \TT^n$.} on $D \times \TT^n$. Moreover, in the coordinates induced by $\Psi$, the Hamiltonian vector field of $f^j_0|_U$ is given by
\begin{equation*} \nu_j \cdot \partial_{\theta} = \nu_j^1 \partial_{\theta^1} + ... + \nu_j^n \partial_{\theta^n},\   \nu_j \in C^\infty(D, \RR^n),\ j =1,...,n.\end{equation*}
The approximation of $\sigma_k$ makes use of the lattice $\nu(0) \ZZ^n$, where \begin{equation}\label{freq_matr} \nu(0) = \left(\begin{array}{ccc} \text{---} & \nu_1(0) & \text{---} \\ & \vdots & \\ \text{---} & \nu_n(0) & \text{---} \end{array}\right).\end{equation}
\begin{proposition}\label{spec_properties} The joint spectrum $\sigma_k$ of $T^1_{k}$, ..., $T^n_{k}$ satisfies the following.
\begin{enumerate}
\item{Let $a_{k,0} \in \sigma_k$ satisfy that $\dist(a_0, \sigma_k) = |a_0 - a_{k,0}|$. Then $$a_{k,0} =  a_0 +  \mco(k^{-1}).$$}
\item{Let $p_k \ge 0$ be a non-decreasing sequence, $p_k = o(k^{\frac 1 2})$ and $\lim_{k \to \infty} p_k = \infty$. Denote
\begin{equation*} \begin{aligned} &\mca_{k} = a_{k, 0} +\frac \pi k \nu(0) \mcc_{k},\\ &\mcc_{k} = \big{\{}(y^1, ..., y^n) \in \RR^n \ | \ |y^j| \le 2\lfloor p_k \rfloor + 1,\ j=1,...,n\big{\}}.\end{aligned}\end{equation*}
Then for every sufficiently large $k$, $$\begin{aligned} &\sigma_{k,0} = \sigma_k \cap \mca_{k} = \{ a_{k,  m} \ | \ m \in \mci_{k}\},\\ &\mci_{k} = \big{\{}(m^1,...,m^n) \in \ZZ^n \ | \ |m^j| \le p_k,\ j=1,...,n\big{\}},\end{aligned}$$ where $a_{k,m} \in \sigma_{k,0}$ is a joint eigenvalue of multiplicity $1$ and satisfies $$\Big{|}a_{k, m} -  a_{k,0} - \frac{2\pi} k \nu(0)  m\Big{|} \le C\frac{|m|^2} {k^2}$$ for some $C >0$ (which is independent of $m\in \mci_k$).}
\item{There exists $c > 0$ such that for all large enough $k$, $\dist(\sigma_k \setminus \sigma_{k, 0}, \sigma_{k,0}) \ge \frac c k$.}
\end{enumerate}\end{proposition}
We note that Proposition \ref{spec_properties} essentially amounts to a first order Taylor approximation of the description of $\sigma_k$ near $a_0$ provided in \cite{charles_bohr}, Theorem 3.1\footnote{cf. Theorem \ref{charles_bohr_thm}, which is sufficient for our needs in the present work.}.
From now on, we fix $p_k$ as in Proposition \ref{spec_properties}.

\subsection{Strong convergence of quantum operators}
Recall that a sequence of operators $A_k\in \End(\mch)$ on a Hilbert space $\mch$ is said to converge strongly to $A \in \End(\mch)$ if it holds that for all $v \in \mch$, $$\lim_{k \to \infty} A_k v = A v.$$
This will be denoted by $A_k \sto A$. The central result of the present paper (Theorem \ref{str_conv_thm}, combined with Corollary \ref{str_conv_cor}) is that certain types of quantum operators, suitably identified with finite rank operators on $L^2(\TT^n)$, converge strongly to multiplication operators. The theorem is formulated so as to apply to quantum operators defined via discontinuous functions\footnote{Note that formula (\ref{toeplitz_def}) makes sense for any $f \in L^\infty(M)$, and likewise, if $Q_k \in \End(\mch_k)$ is self-adjoint then $\kappa(Q_k) \in \End(\mch_k)$ is well-defined for any function $\kappa : \RR \to \CC$.}. This formulation was chosen partly to reflect the robustness of the phenomenon (which pertains to $L^2$ convergence, where smoothness is not as important as in other contexts), and partly due to the specific application described in Sect. \ref{pairs_sect}. We note that strong convergence has some spectral implications (see Corollary \ref{str_conv_spec_imp}). Throughout, $L^2(\TT^n)$ is defined using the normalized Lebesgue measure.
\begin{theorem}\label{str_conv_thm} There exist orthonormal sets $\mcb_k = \{s_{k,m} \ | \ m \in \mci_k\}$, where $s_{k,m}$ is a joint eigensection of $T_k^1, ..., T_k^n$ corresponding to $a_{k,m} \in \sigma_{k,0}$, such that the following holds. Let $$U_k : \mch_k \to \mch_k^{\TT^n} \subset L^2(\TT^n)$$ be any unitary transformation such that $$U_k s_{k,m} = e_m,\ m \in \mci_k.$$
Given $Q \in \End(\mch_k)$, let $[Q]^{\TT^n} \in \End(L^2(\TT^n))$ denote the operator specified by $$[Q]^{\TT^n}\big{|}_{\mch_k^{\TT^n}} = U_k Q U_k^{-1},\ [Q]^{\TT^n}\big{|}_{(\mch_k^{\TT^n})^\perp} = 0.$$
Below, we identify $\Lambda_{a_0}\subset M$ with $\TT^n$.
\begin{enumerate}
\item{Let $A \subset M$ be an open set such that $\partial A \cap \Lambda_{a_0}$ has measure zero in $\Lambda_{a_0}$. Let $h \in C(M)$ be a continuous function, and write $h_A = h \mathbbm 1_A$. Then $$[T_k(h_A)]^{\TT^n} \sto \mcm_{h_{A, 0}},\ h_{A,0} = h_A\big{|}_{\Lambda_{a_0}}.$$}
\item{Let $g(\cdot, k) \in C^\infty(M)$ be a sequence of functions admitting the expansion $$g(\cdot, k) \sim \sum_{j =0}^\infty g_j k^{-j}$$
in the $C^\infty$ topology. Let $b_0 \in g_0(M)$ be a regular value of $g_0$ such that $$g_0^{-1}(b_0) \pitchfork \Lambda_{a_0} \ne \emptyset.$$
Denote $\Pi_{k,g,b_0} = \mathbbm 1_{(b_0, \infty)}\big{(}T_k(g(\cdot,k))\big{)}$ and $A_0 = \{ g_0 > b_0\} \cap \Lambda_{a_0}$. Then $$[\Pi_{k,g,b_0}]^{\TT^n} \sto \mcm_{\mathbbm 1_{A_0}}.$$}
\end{enumerate}\end{theorem}
Together with some basic facts about strong convergence, Berezin-Toeplitz operators, and the mapping $[\ \cdot\ ]^{\TT^n} : \End(\mch_k) \to \End(L^2(\TT^n))$, Theorem \ref{str_conv_thm} is in fact applicable to a relatively large class of operators. More specifically,
\begin{corollary}\label{str_conv_cor} Recall that if $A_k \sto A$, $B_k \sto B$ then $A_k B_k \sto AB$. Also note that the assignment $$\End(\mch_k)\ni Q \mapsto [Q]^{\TT^n}\in \End(L^2(\TT^n))$$ is linear and multiplicative (i.e., $[Q_1 Q_2]^{\TT^n} = [Q_1]^{\TT^n} [Q_2]^{\TT^n}$). Finally, note that if $g(\cdot,k) \in C^\infty(M)$ is real-valued and admits an asymptotic expansion as above, and $\kappa \in C(\RR, \CC)$ is continuous then\footnote{The map $T_k : C^\infty(M)\to \End(\mch_k)$ extends to a $\CC$-linear map $T_k : L^\infty(M) \to \End(\mch_k)$.} $$\Vert \kappa(T_k(g(\cdot, k))) - T_k(\kappa \circ g_0)\Vert_{\op}= o(1).$$ Hence, for any sequence $Q_k \in \End(\mch_k)$ constructed via these operations, we also conclude that $[Q_k]^{\TT^n}$ converges strongly to a multiplication operator. \end{corollary}
Next, we recall that strong convergence implies weak convergence, that is, if $Q_k \in \End(\mch_k)$ is a sequence such that $[Q_k]^{\TT^n} \sto \mcm_{h_0}$ for some $h_0 \in L^\infty(\TT^n)$ then for any $u,v \in L^2(\TT^n)$, $$\lim_{k \to \infty} \langle [Q_k]^{\TT^n} u, v \rangle_{L^2(\TT^n)} = \langle h_0 u, v \rangle_{L^2(\TT^n)}.$$
Notably, if $u,v$ are trigonometric polynomials, then for any sufficiently large $k$ there exist $u_k, v_k \in \mch_k$ such that $$u = U_k u_k,\ v = U_k v_k,$$
and hence $$\lim_{k \to \infty} \langle Q_k u_k, v_k \rangle_{\mch_k} = \langle h_0 u, v \rangle_{L^2(\TT^n)}.$$
Thus, matrix coefficients of $Q_k$ converge to Fourier coefficients of $h_0$, as follows.
\begin{corollary}\label{matr_coef_cor} Assume $Q_k \in \End(\mch_k)$, $h_0 \in L^\infty(\TT^n)$ and $[Q_k]^{\TT^n} \sto \mcm_{h_0}$. For $m \in \ZZ^n$, the $m$-th Fourier coefficient of $h_0$ is given by $$\widehat{h_0}(m) = \langle h_0, e_m\rangle_{L^2(\TT^n)}.$$ If $l,m \in \ZZ^n$ are fixed, then $$\lim_{k \to \infty} \langle Q_k s_{k,l}, s_{k,m} \rangle_{\mch_k} = \widehat{h_0}(m-l).$$ \end{corollary}
For operators defined via smooth functions, the convergence rate is $\mco(k^{-1})$.
\begin{remark}\label{matr_coef_rk} If $Q_k = T_k(h)$ for $h \in C^\infty(M)$, then (using the above notations) $$\langle Q_k s_{k,l}, s_{k,m} \rangle = \widehat{h_0}(m-l) + \mco(k^{-1}),\ h_0 = h|_{\Lambda_{a_0}}.$$ Similarly for $\kappa(T_k(h))$, $\kappa \in C^\infty(\RR, \CC)$, since $\kappa(T_k(h)) = T_k(\kappa \circ h) + \mco(k^{-1})$. \end{remark}
We note that in fact, as far as the proof of Theorem \ref{str_conv_thm} is concerned, the convergence of matrix coefficients to Fourier coefficients is established first, and Theorem \ref{str_conv_thm} is then deduced as a consequence (see Lemma \ref{str_conv_explanation1}).

The results above only address the leading order behaviour of the operators involved. However, in the case of the operators $$[T^j_{k}]^{\TT^n},\ j = 1,...,n,$$ the sub-leading behaviour is essentially described in Proposition \ref{spec_properties}. Namely, these operators behave approximately as first order differential operators with constant coefficients. More specifically,
\begin{remark}\label{str_conv_rk} Writing $a_{k,0} = (a^1_{k,0}, ..., a^n_{k,0})$, we note that the sequence $$k\left([T^j_{k}]^{\TT^n} - a^j_{k,0}\Id_{L^2(\TT^n)}\right)$$
converges to the operator $-i \nu_j(0) \cdot \partial_\theta,\ \partial_\theta = (\partial_{\theta^1}, ..., \partial_{\theta^n})$ in the sense of strong resolvent convergence\footnote{The notion of strong convergence is not applicable in the context of unbounded operators. If $A_k$, $A$ are (possibly unbounded) self adjoint operators on a Hilbert space $\mch$ such that $R_i(A_k) \sto R_i(A)$, where $R_i(A) = (A - i \Id_{\mch})^{-1}$ is the resolvent and similarly for $A_k$, then we say that $A_k$ converges to $A$ in the sense of strong resolvent convergence.} (see Remark \ref{str_res_conv_exp}).\end{remark}
We note that strong convergence implies strong resolvent convergence, which has the following spectral implications (\cite{oliveira}). We denote the spectrum of a linear operator $A$ by $\sigma(A)$.
\begin{corollary}\label{str_conv_spec_imp} If $A_k$ is a sequence of (possibly unbounded) self-adjoint operators on a complex, separable Hilbert space $\mch$ which converges, in the sense of strong resolvent convergence, to a (possibly unbounded) self-adjoint operator $A$, then for every $\lambda \in \sigma(A)$ there exists a sequence $\lambda_k \in \sigma(A_k)$ such that $$\lim_{k \to \infty} \lambda_k = \lambda.$$
In the present context, if $Q_k =Q_k^* \in \End(\mch_k)$ is a sequence as in Theorem \ref{str_conv_thm} (or Corollary \ref{str_conv_cor}) such that $[Q_k]^{\TT^n} \sto \mcm_{h_0}$ for some $h_0 \in L^\infty(\TT^n)$, then\footnote{Here, $\essran(h_0)$ denotes the essential range of $h_0$.}
\begin{equation}\label{spec_imp}\lim_{k \to \infty}\max_{\lambda \in \essran(h_0)} \dist(\lambda, \sigma(Q_k))= 0.\end{equation}
\end{corollary}
\begin{remark} In the previous corollary, note that $\sigma([Q_k]^{\TT^n}) = \sigma(Q_k) \cup\{0\}$. Thus, the case $\lambda = 0 \in \essran(h_0)$ needs special attention. For instance, given $\varepsilon >0$, fix a normalized trigonometric polynomial $u\in L^2(\TT^n)$ with $\Vert h_0 u \Vert < \frac \varepsilon 2$. Then for all sufficiently large $k$, we obtain unit sections $s_k\in \mch_k$ such that $\Vert Q_k s_k \Vert < \varepsilon$. It follows that $\dist(0, \sigma(Q_k)) < \varepsilon$, therefore $\dist(0, \sigma(Q_k))  \to 0$. The limit (\ref{spec_imp}) readily follows using the compactness of $\essran(h_0)$. \end{remark}
We conclude the subsection by noting that the formulation of Theorem \ref{str_conv_thm} may be considered as somewhat misleading, in the following sense. While the embedding $U_k : \mch_k \to \mch_k^{\TT^n}$ is defined on the entire quantum space $\mch_k$, the strong convergence described in the theorem is determined by the action of the relevant operators on joint spectral subspaces associated with joint eigenvalues $a_{k,m}\in \sigma_k$ for $m \in \ZZ^n$ arbitrarily large, but fixed\footnote{Accordingly, the theorem does not depend on the choice of sequence $p_k$ used to define $\mci_k$ in Proposition \ref{spec_properties}.}. If we define $$\mch_{k,a_0}^{\TT^n} = U_k \mch_{k,a_0},\ \mch_{k,a_0} = \Sp \mcb_k,$$
where $\mcb_k$ is as specified in Theorem \ref{str_conv_thm}, then $\dim \mch_{k,a_0} = o(\dim \mch_k),$ while the subspaces $$(\mch_{k,a_0}^\perp)^{\TT^n} = U_k \mch_{k,a_0}^\perp \subset L^2(\TT^n),$$
"escape to infinity" as $k \to \infty$. The strong convergence described in the theorem is determined by the action of the operators on $\mch_{k,a_0}^{\TT^n}$ (informally speaking, $\mch_{k,a_0}$ is a sort of "local quantum space" about $\Lambda_{a_0}$).

We justify the formulation of Theorem \ref{str_conv_thm} mostly on technical grounds, since it seems to be somewhat simpler than the alternatives, and it readily allows to exploit the standard asymptotic properties of Berezin-Toeplitz quantization. Additionally, informally speaking, the formulation seems roughly consistent with a valid perspective, namely, the way the quantization of $M$ appears when viewed "from $\Lambda_{a_0}$". If instead we compress the operators to $\mch_{k,a_0}$, then in some respects, subsets of $M$ which are separated from $\Lambda_{a_0}$ are "negligible" already for finite $k$ (rather than in the limit $k \to \infty$). In this context, we note the following simple examples (which are proven in Sect. \ref{lagr_sect}).
\begin{example}\label{loc_example1} Let $h : M \to \RR$ be a bounded measurable function. For $\varepsilon >0$, denote $$\Lambda_{a_0, \varepsilon} = \{x \in M \ | \ \dist(x, \Lambda_{a_0}) \le \varepsilon\}.$$
Let $\lambda \in \RR$. Assume that there exists $\varepsilon > 0$ such that $\lambda \not \in [\inf_{\Lambda_{a_0, \varepsilon}} h, \sup_{\Lambda_{a_0, \varepsilon}} h]$. Let $$\Pi_{k, a_0} : \mch_k \to \mch_{k,a_0}$$ be the orthogonal projection on $\mch_{k,a_0}$. Then for all sufficiently large $k$, $$\lambda \not \in \sigma(T_{k,a_0}(h)),$$ where $T_{k,a_0}(h) = \Pi_{k,a_0} T_k(h) \Pi_{k,a_0}$ is viewed as an operator on $\mch_{k,a_0}$. \end{example}
\begin{example}\label{loc_example2} Let $h\in C(M)$. Consider the compressed operators $$T_{k,a_0}(h) = \Pi_{k,a_0} T_k(h) \Pi_{k,a_0} \in \End(\mch_{k,a_0}).$$ Then the Hausdorff distance $$\delta_{k,a_0} = \dist_{\text{H}}\big{(}\sigma\big{(}T_{k,a_0}(h)\big{)}, h(\Lambda_{a_0})\big{)}$$
between $\sigma\big{(}T_{k,a_0}(h)\big{)}$ and $h(\Lambda_{a_0})$ satisfies\footnote{This is not the case for $T_k(h)$, assuming that $h\in C(M)$ and $h(M) \setminus h(\Lambda_{a_0}) \ne \emptyset$.} $\lim_{k \to \infty} \delta_{k,a_0} = 0$. \end{example}
\subsection{Pairs of spectral projections of quantum observables}\label{pairs_sect}
While strong convergence is much weaker than convergence with respect to the operator norm (which cannot be expected in the context of Theorem \ref{str_conv_thm}), it can still have some notable implications. In the present subsection, we describe such implications for pairs of spectral projections associated with $T_k^j$, $j=1,...,n$, and $T_k(g(\cdot, k))$, $g(\cdot, k) \in C^\infty(M)$. Spectral projections of Berezin-Toeplitz operators have attracted some attention in recent years (\cite{zz1, zz2, shabtai3, ioos}). The results we present here (Theorem \ref{main_proj_thm}, Corollary \ref{max_cor}) extend earlier findings (\cite{shabtai1, shabtai2}) on pairs of spectral projections of spin operators.
\begin{theorem}\label{main_proj_thm}
Let $g$, $b_0$, and $\Pi_{k, g, b_0}$ be as in Theorem \ref{str_conv_thm}. For $c \in \RR^n \setminus \{0\}$, let $X_{c \cdot f_0}$ denote the Hamiltonian vector field of $c \cdot f_0$. Assume that the flow of $X_{c \cdot f_0}\big{|}_{\Lambda_{a_0}}$ is periodic, and that $\{c \cdot f_0, g_0\}(x) \ne 0$
for some $x \in \{g_0 = b_0\} \cap \Lambda_{a_0}$. Write $c = (c^1, ..., c^n)$, and define $$\Pi_k^c = \mathbbm 1_{(c \cdot a_0, \infty)}(S_k^c),\ S_k^c = c^1 T_k^1 + ... + c^n T_k^n.$$  Let $N^c_{k}(I)$ denote the number of eigenvalues of the operator $$\Pi^c_{k} \Pi_{k,g,b_0} \Pi^c_{k}\in \End(\Pi_k^c \mch_k)$$ in the interval $I\subset [0,1]$, counted with multiplicity. Then for any $0 \le t < t' \le 1$, $$\lim_{k \to \infty} N^c_{k}([t,t']) = \infty.$$ \end{theorem}
\begin{example}\label{proj_ex} In the setting of Theorem \ref{main_proj_thm}, by transversality, there exists $j \in \{1,...,n\}$ such that $\lim_{k \to \infty} N_k^{c_j}([t, t']) = \infty$ for any $0 \le t < t' \le 1$, where $c_j$ is the $j$-th row of $\nu(0)^{-1}$. \end{example}
\begin{remark} In the setting of Theorem \ref{main_proj_thm}, if $\mu(\{c \cdot f_0 = c \cdot a_0\}) = 0$ then $N_k^c([t, 1-t]) = o(\dim \mch_k)$ for all $0 < t < \frac 1 2$ (see Proposition \ref{two_proj_prop}). A much sharper estimate of $\mco(\log k)$ was obtained in \cite{shabtai2} for certain pairs of spectral projections of spin operators.\end{remark}

Generally, if $P_1$, $P_2$ are orthogonal projections on a complex separable Hilbert space (possibly infinite dimensional), then the spectrum of $P_1 P_2 P_1$ largely determines the complex unital algebra generated by $P_1$, $P_2$. Accordingly, Theorem \ref{main_proj_thm} is equivalent\footnote{\cite{shabtai_thesis}, Claim 3.4.} to a certain "maximality property", as follows.

Let $\mca$ be the complex unital algebra generated by noncommuting variables $p_1,p_2$ which satisfy $p_1^2 = p_1$, $p_2^2 = p_2$. Let $0 \ne \phi \in \mca$. Then it is known (\cite{shabtai2}) that there exists a universal, tight upper bound $B_\phi > 0$ (which depends only on $\phi$) for $\Vert \phi(P_1, P_2)\Vert_{\op}$, where $P_1$, $P_2$ are arbitrary orthogonal projections (on some separable complex Hilbert space). 
\begin{corollary}[Asymptotic maximality for polynomials in two projections]\label{max_cor} Let $c \in \RR^n\setminus \{0\}$ be as in Theorem \ref{main_proj_thm}, so that for any $0 \le t < t' \le 1 $, it holds that $\lim_{k \to \infty} N^c_{k}([t, t']) = \infty.$ Then for any $0 \ne \phi \in \mca$, $$\lim_{k \to \infty} \Vert \phi(\Pi^c_{k}, \Pi_{k,g,b_0})\Vert_{\op} = B_\phi.$$  \end{corollary}
\begin{remark} The universal upper bound $B_\phi$ admits a concise formula in terms of $\phi$ (\cite{shabtai2}, Corollary 6). For example, if $\phi(p_1, p_2) = p_1 p_2 - p_2 p_1$ is the commutator polynomial, then $B_\phi = \frac 1 2$. Thus, in the setting of Corollary \ref{max_cor}, $$\lim_{k \to \infty} \Vert [\Pi_k^c, \Pi_{k,g,b_0} ]\Vert_{\op} = \frac 1 2.$$ The computation of this limit was one of the main motivations for the present paper; the fact that it is nonzero demonstrates that the correspondence principle (\ref{corr_princ}) does not extend to "discontinuous observables". \end{remark}
The proof of Theorem \ref{main_proj_thm} is reduced, partly via Theorem \ref{str_conv_thm}, to the study of pairs of spectral projections on $L^2(\TT^n)$. Indeed, by Theorem \ref{str_conv_thm}, $$[\Pi_{k,g,b_0}]^{\TT^n} \sto \mcm_{\mathbbm 1_{A_0}},\ A_0 = \{g_0 > b_0\} \cap \Lambda_{a_0}.$$
By contrast, strictly speaking, the sequences $[\Pi^c_{k}]^{\TT^n}$ do not converge in general. However, the potential lack of convergence may only occur due to "bounded fluctuations". This may be illustrated as follows. Consider the orthogonal projections \begin{equation}\label{j_szego}\Pi^j : L^2(\TT^n) \to H^2_j(\TT^n)\end{equation}
on the "partial Hardy spaces"
$$\begin{aligned} &H_j^2(\TT^n) = \{g \in L^2(\TT^n) \ | \ \hat g(m) = 0 \text{ for all } m \in \ZZ^n_{j, -}\},\\ &\ZZ^n_{j, -} = \{(m^1, ..., m^n) \in \ZZ^n \ | \ m^j < 0\}.\end{aligned}$$
\begin{corollary} For $j \in \{1,...,n\}$ and $c_j$ as in Example \ref{proj_ex}, define $$\Pi_{k,\varepsilon}^j = \mathbbm 1_{(-\varepsilon, \infty)}\Big{(}k\big{(} S^{c_j}_{k}- c_j \cdot a_{k,0} \Id_{\mch_k}\big{)}\Big{)}.$$ Then in light of Proposition \ref{spec_properties}, for any $\varepsilon > 0$ small enough,  $[\Pi_{k,\varepsilon}^j]^{\TT^n} \sto \Pi^j.$
Thus, if $g(\cdot, k) \in C^\infty(M)$ and $A_0 \subset \TT^n$ are as in Theorem \ref{str_conv_thm} then $$[\Pi_{k,\varepsilon}^j]^{\TT^n} [\Pi_{k,g, b_0}]^{\TT^n} [\Pi_{k, \varepsilon}^j]^{\TT^n} \sto \Pi^j \mcm_{\mathbbm 1_{A_0}} \Pi^j,$$
and if $\sigma(\Pi^j \mcm_{\mathbbm 1_{A_0}} \Pi^j) = [0,1]$, then $\sigma(\Pi_{k,\varepsilon}^j \Pi_{k, g, b_0} \Pi_{k,\varepsilon}^j)$ becomes dense in $[0,1]$.
Similarly, if $h_A$ and $h_{A,0}$ are as in Theorem \ref{str_conv_thm} then $$[\Pi_{k,\varepsilon}^j]^{\TT^n} [T_k(h_A)]^{\TT^n} [\Pi_{k, \varepsilon}^j]^{\TT^n} \sto \Pi^j \mcm_{h_{A, 0}} \Pi^j.$$
The operators $\Pi_{k,\varepsilon}^j T_k(h_A) \Pi_{k,\varepsilon}^j$ are examples of \textit{quantum Zeno Hamiltonians}\footnote{i.e., quantum observables which are compressed to a subspace of the full quantum space.} (\cite{fp}). The strong convergence relates their spectra with that of $\Pi^j \mcm_{h_{A,0}} \Pi^j$. The spectrum of operators of the form $\Pi^j \mcm_{h_0} \Pi^j$, with $h_0 \in L^\infty(\TT^n)$ real valued, is described in (\ref{hartman_wintner}) for $n=1$, and in Lemma \ref{direct_int_spec} for $n>1$. \end{corollary}
We conclude the discussion on pairs of spectral projections by noting the following general spectral behaviour (proven in Sect. \ref{proj_proof_sect}).
\begin{proposition}\label{two_proj_prop} Let $h_1(\cdot, k), h_2(\cdot, k) \in C^\infty(M)$ be sequences with expansions $h_j \sim \sum_{l \ge 0} h_{j,l} k^{-l},\ j = 1,2,$ in the $C^\infty$ topology. Assume that for $\alpha_j \in h_{j,0}(M)$ it holds that $\mu(\{h_{j,0} = \alpha_j\}) = 0$, and define $$\Pi_{k,j} = \mathbbm 1_{(\alpha_j, \infty)}(T_k(h_j(\cdot, k))),\ j = 1,2.$$
Let $N_k(I)$ denote the number of eigenvalues (counting multiplicities) of the operator $\Pi_{k,1} \Pi_{k,2} \Pi_{k,1} \in \End(\Pi_{k,1} \mch_k)$ in the interval $I\subset [0,1]$. Then for every $0 < t < \frac 1 2$,  $$\begin{aligned} &N_k([0,t]) = \Big{(}\frac k {2\pi} \Big{)}^n \mu(\{h_{1,0}> \alpha_1\} \cap \{h_{2,0} < \alpha_2\}) + o(k^n),\\ &N_k((t, 1-t)) = o(k^n),\\ &N_k([1-t, 1]) = \Big{(}\frac k {2\pi} \Big{)}^n \mu(\{h_{1,0} > \alpha_1\} \cap \{h_{2,0} > \alpha_2\}) + o(k^n).\end{aligned}$$   \end{proposition}
\begin{remark}[cf. Corollary \ref{ms_cor}, and \cite{shabtai3}, Corollary 1.2] In the above notation, note that if $\{h_{1,0} \ge \alpha_1\} \cap \{h_{2,0} \ge \alpha_2\} = \emptyset$, then $\Vert \Pi_{k,1} \Pi_{k,2} \Vert_{\op} = \mco(k^{-\infty})$. \end{remark}
\subsection{Contraction of projective quasi-representations of Lie algebras}
Contraction of Lie algebras, introduced in the works of I.E. Segal (\cite{segal}) and E. İnönü and E. Wigner (\cite{iw}), is a way of relating two (non-isomorphic) Lie algebras through a certain type of limiting process. When a Lie algebra $\mfg_\infty$ is obtained from a Lie algebra $\mfg$ via contraction, it could be expected that a corresponding relation might hold at the level of representation theory. In the context of Berezin-Toeplitz quantization, which typically does not produce genuine representations, it is possible to consider contractions of sequences of "representations up to manageable errors".
The results of Sect. \ref{main_res_sect} imply that the Berezin-Toeplitz quantization of $M$ gives rise to certain such contractions.
We refer the reader to \cite{aa, cahen} for earlier works on contractions of representations by means of Berezin-Toeplitz and geometric quantization.
We begin by recalling the definition of sequential contraction of Lie algebras.
\begin{definition}[\cite{bbms}]\label{contr_def} Let $\mfg = (V, [\cdot, \cdot])$ be a Lie algebra with underlying vector space $V$ and Lie bracket $[\cdot,\cdot]$. Assume that $t_k : V \to V$ is a sequence of invertible linear maps such that for every $X,Y \in V$, the limit $$[X,Y]_\infty := \lim_{k \to \infty} t_k^{-1}([t_k(X), t_k(Y)])$$
exists and defines a Lie bracket on $V$. Then the Lie algebra $\mfg_\infty = (V, [\cdot, \cdot]_\infty)$ is called the \textit{(sequential) contraction} of $\mfg$ by $t_k$. \end{definition}

Next, inspired by the definitions specified in \cite{bbms, bs}, we formulate a notion of contraction of a sequence of linear maps $\xi_k : \mathfrak g \to \End(W)$ to a Lie algebra representation $\xi_\infty : \mathfrak g_\infty \to \End(W)$, where $W$ is an inner product space. We do not require the maps $\xi_k$ to be Lie algebra representations.
\begin{definition}\label{contr_maps_def} Assume that $\mfg_\infty$ is the sequential contraction of $\mfg$ via the invertible maps $t_k$. Let $\xi_\infty : \mfg_\infty \to \End(W)$ be a representation of $\mfg_\infty$ on an inner product space $W$, and let $\xi_k : \mfg \to \End(W)$ be a sequence of linear maps. We say that $\xi_k$ contracts to $\xi_\infty$ if the following holds.
\begin{enumerate}
\item{For every $X \in V$ and $w \in W$, $$\lim_{k \to \infty} \big{\Vert} \big{(}(\xi_k \circ t_k)(X)  - \xi_\infty(X)\big{)}w \big{ \Vert} = 0.$$}
\item{For every $X, Y \in V$ and $w \in W$, $$\lim_{k \to \infty} \big{\Vert}\big{(} \big{[}(\xi_k \circ t_k)(X), (\xi_k \circ t_k)(Y)\big{]} - \xi_\infty([X,Y]_\infty)\big{)}w \big{\Vert} = 0.$$}
\end{enumerate}\end{definition}
Although the definition above allows for general sequences of linear maps $\xi_k : \mfg \to \End(W)$, in practice we will only encounter sequences consisting of linear maps which are Lie algebra representations up to "manageable" errors.
\begin{definition}\label{spqr} Let $\mfg$ be a real Lie algebra, and $W$ a complex inner product space. A sequence of linear maps $$\xi_k : \mfg \to \End(W)$$
is called a semiclassical projective quasi-representation of $\mfg$ on $W$ if there exists a Lie algebra $2$-cocycle\footnote{i.e., $c$ is bilinear, antisymmetric, and $c([X,Y],Z) + c([Z,X],Y)+c([Y,Z],X) = 0$ for any $X,Y,Z \in \mfg$.} $c: \mfg \times \mfg \to \RR$ such that for every $X,Y \in \mfg$ and $w \in W$, $$\Vert([\xi_k(X), \xi_k(Y)] - \xi_k([X,Y]) - ik c(X,Y) \Id_W)w \Vert = \mco(1).$$ \end{definition}
The identification of the spaces $\mch_k$ with subspaces $\mch_k^{\TT^n} \subset L^2(\TT^n)$ via linear isometries $U_k : \mch_k \to \mch_k^{\TT^n}$, as in Theorem \ref{str_conv_thm}, gives rise to contractions of certain semiclassical (projective) quasi-representations of certain subalgebras of $(C^\infty(M), \{\cdot, \cdot\})$. Throughout, we assume that $U_k$ maps joint eigensections of $T_k^1,...,T_k^n$ to exponentials\footnote{We complete the set $\mcb_k$ to an orthonormal basis $\hat \mcb_k$ of $\mch_k$ consisting of joint eigensections of $T_k^1, ..., T_k^n$, and define $U_k$ so that it maps the elements of $\hat \mcb_k \setminus \mcb_k$ to unused exponentials.}. Let us first mention the "trivial" contraction.
\begin{example} The invertible linear maps $$t_k = \frac 1 k \Id : C^\infty(M) \to C^\infty(M)$$ produce a contraction of $(C^\infty(M), \{\cdot, \cdot\})$ to the commutative Lie algebra with underlying vector space $C^\infty(M)$. Define $$Q_k^{\TT^n} : C^\infty(M) \to \End(L^2(\TT^n)),\ Q_k^{\TT^n}(g) = ik [T_k(g)]^{\TT^n}.$$
The norm correspondence (\ref{norm_corr}) and correspondence principle (\ref{corr_princ}) imply that $Q_k^{\TT^n}$ is a (non-projective) semiclassical quasi-representation.

Then by Theorem \ref{str_conv_thm} (identifying $\Lambda_{a_0} \simeq \TT^n$), $$\big{(}Q_k^{\TT^n} \circ t_k \big{)}(g) = i[T_k(g)]^{\TT^n}$$
converges strongly to the multiplication operator $\mcm_{ig|_{\Lambda_{a_0}}}$. Thus, in light of the correspondence principle (\ref{corr_princ}), the sequence $Q_k^{\TT^n}$ contracts to the representation $g \mapsto \mcm_{i g|_{\Lambda_{a_0}}}$ of the commutative Lie algebra $C^\infty(M)$. \end{example}
The results of Sect. \ref{main_res_sect} may be used to construct non-trivial contractions of semiclassical projective quasi-representations of certain subalgebras of $C^\infty(M)$, as follows. Let $(C^\infty_{\inv}(M), \{\cdot, \cdot\})$ be the subalgebra of $(C^\infty(M), \{\cdot,\cdot\})$ consisting of functions which commute with $f^1_{0}, ..., f^n_{0}$, that is, \begin{equation}\label{inv_subalg} C_{\inv}^\infty(M) = \{g \in C^\infty(M) \ | \ \{g, f^j_{0}\} = 0,\ j = 1,...,n\}.\end{equation}
Assume that \begin{equation}\label{decomp_subalg}\mfg = (V_{0} \oplus V_{\inv}, \{\cdot, \cdot\})\end{equation} is a subalgebra of $(C^\infty(M), \{\cdot,\cdot\})$, where $V_{0} \oplus V_{\inv}$ is a direct sum of vector spaces with $V_{\inv} \subset C^\infty_{\inv}(M)$ a Lie subalgebra, and such that for all $g \in V_{0}$, $h \in V_{\inv}$, it holds that \begin{equation}\label{comm_rel}\{g, h\} \in V_{0}.\end{equation}
For $0 \ne R \in \RR$, define linear maps $t_{k,R} : \mfg \to \mfg$ by \begin{equation}\label{contr_maps}t_{k,R} \big{|}_{V_0} = \frac R k \Id,\ t_{k,R} \big{|}_{V_{\inv}} = \Id.\end{equation}
Then $\mfg$ contracts via $t_{k,R}$ to $\mfg_\infty = (V_0 \oplus V_{\inv}, \{\cdot, \cdot\}_\infty)$, with $\{\cdot, \cdot\}_\infty$ specified by $$\{\cdot, \cdot\}_\infty \big{|}_{V_0 \times V_0} =0,\  \{\cdot, \cdot\}_\infty \big{|}_{(V_0 \oplus V_{\inv})\times V_{\inv}} = \{\cdot, \cdot\}.$$
In other words, $$\mfg_\infty = (V_{\inv}, \{\cdot, \cdot\}) \ltimes (V_0, 0),$$
where $(V_0, 0)$ is the commutative Lie algebra with underlying vector space $V_0$.

\begin{theorem}\label{contr_thm} Let $\mfg_\infty$ be the contraction of $\mfg$ (\ref{decomp_subalg}) via the maps $t_{k,R}$ (\ref{contr_maps}). Let $X_g$ denote the Hamiltonian vector field of $g \in C^\infty(M)$. For $g \in V_{\inv}$ write $$X_g|_{\Lambda_{a_0}} = \zeta_1(g) X_{f_0^1}|_{\Lambda_{a_0}} + ... + \zeta_n(g) X_{f_0^n}|_{\Lambda_{a_0}}.$$
Define a sequence $\xi_k : \mfg \to \End(C^\infty(\TT^n))$ by \begin{equation}\label{xi_k_def}\xi_k(g) =  \left\{\begin{array}{ll} i k[T_k(g)]^{\TT^n} & \text{if } g \in V_0,\\ ik \sum_{j=1}^n  \zeta_j(g) \left([T^j_k]^{\TT^n} - a^j_{k,0} \Id \right) & \text{if } g \in V_{\inv}.\end{array}\right.\end{equation}
Then $\xi_k$ is a semiclassical projective quasi-representation of $\mfg$ which contracts via $t_{k,R}$ to the representation $\xi^R_\infty : \mfg_\infty \to \End(C^\infty(\TT^n))$ specified by $$\xi^R_{\infty}(g) = \left\{\begin{array}{ll} iR \mcm_{g|_{\Lambda_{a_0}}} & \text{if } g \in V_0,\\  X_g\big{|}_{\Lambda_{a_0}} & \text{if } g \in V_{\inv}.\end{array}\right..$$
\end{theorem}
\begin{remark} In Theorem \ref{contr_thm}, the space $C^\infty(\TT^n)$ encodes sections of $\mch_k$ which are "highly localized" about $\Lambda_{a_0}$, in the following sense. Fix $\mci_k$ as in Proposition \ref{spec_properties} with $p_k = \lfloor k^{\frac 1 4} \rfloor$. Any $u\in C^\infty(\TT^n)$ satisfies $u = u_k + \mco(k^{-\infty})$, where $u_k = \sum_{m \in \mci_k} \widehat{u}(m) e_m$ can be identified with $s_k = \sum_{m \in \mci_k} \widehat{u}(m) s_{k,m} \in \mch_k$. Here, $\widehat{u}(m) = \mco(|m|^{-\infty})$ and $s_{k,m}$ are as in Theorem \ref{str_conv_thm}.\end{remark}
\begin{remark} The semiclassical projective quasi-representation (\ref{xi_k_def}) is exact. Namely, if $P_{\inv} : V_0 \oplus V_{\inv} \to V_{\inv}$ is the projection, then the $2$-cocycle associated with $\xi_k$ is defined by $c(g, h) = l(\{g, h\})$, where $l \in \mfg^*$ is the linear functional $$V_0 \oplus V_{\inv} \ni g \mapsto l(g)=  P_{\inv}(g)\big{|}_{\Lambda_{a_0}},$$
and the latter is viewed as a constant. Hence, the sequence of "shifted" maps $ \xi'_k(g) = \xi_k(g) + ikl(g) \Id$ is a non-projective semiclassical quasi-representation. However, for any $g \in \mfg$ such that $l(g) \ne 0$, if $u \ne 0$ then $ \xi'_k \circ t_{k,R}(g) u$ diverges. Thus, in general, $\xi'_k$ cannot be expected to contract to $\xi_\infty^R$. \end{remark}
In what follows, we specify a number of examples for Lie algebras $\mfg$ of the form (\ref{decomp_subalg}), which therefore admit contractions as specified in Theorem \ref{contr_thm}. We denote the normalized Lebesgue measure on $\TT^n$ by $\mu_n$.
\begin{example}\label{comp_supp_contr} A natural example to consider is the subalgebra of compactly supported functions $$\mfg = (C^\infty_c(U), \{\cdot, \cdot\}),$$
where $U \simeq D \times \TT^n$ is the neighborhood of $\Lambda_{a_0}$ specified in (\ref{al_sympl}). Then $$\begin{aligned} &C^\infty_c(U) = V_0 \oplus V_{\inv},\\ &V_{\inv} = C_c^\infty(U) \cap C_{\inv}^\infty(M),\\ & V_0 = \Big{\{}g \in C^\infty_c(U) \ \Big{|} \ \int_{\TT^n} g(y, \theta) d\mu_n(\theta) = 0\ \text{for all }y \in D\Big{\}}.\end{aligned}$$
Note that $\ker \xi_\infty^R$ is infinite dimensional\footnote{If $g \in \mfg$ vanishes in a neighborhood of $\Lambda_{a_0}$, or if $g \in V_{\inv}$, $X_g\big{|}_{\Lambda_{a_0}} = 0$, then $g\in \ker \xi_\infty^R$.}. \end{example}
In some settings, the entire space $C^\infty(M)$ admits a decomposition as in Example \ref{comp_supp_contr}, and then a similar contraction is possible.
\begin{example} Assume that $(M, \omega, f_0)$ is a toric integrable system, i.e., the Hamiltonian flows of the components $f_0^1$,..., $f_0^n$ of $f_0$ define an effective $\TT^n$ action on $M$. Then $$C^\infty(M) =V_0 \oplus V_{\inv} = C^\infty_\perp(M) \oplus C^\infty_{\inv}(M),$$
where $C^\infty_{\inv}(M)$ is as specified in (\ref{inv_subalg}), and $$C^\infty_{\perp}(M) = \Big{\{}g \in C^\infty(M) \ \Big{|} \ \int_{\TT^n} g(\theta \cdot x) d\mu_{n}(\theta) = 0 \text{ for all }x \in M\Big{\}}.$$ \end{example}
In general, $C^\infty(M)$ does not necessarily admit a decomposition as above. However, in our context, this is to some extent only a technical difficulty, since the contraction in Theorem \ref{contr_thm} only depends on semi-local data about $\Lambda_{a_0}$ (where such a decomposition is possible). %

Next, we note that a natural setting to consider Theorem \ref{contr_thm} is when $M$ carries a Hamiltonian group action (see Sect. \ref{S2_sect}).
\begin{example}\label{cpct_grp_ex} Assume that $(M, \omega)$ admits an effective Hamiltonian $G$-action, where $G$ is a compact Lie group. Identify $\mfg = \text{Lie}(G)$ with a subalgebra of $C^\infty(M)$, and note that as a vector space, $\mfg = \mft \oplus \mfm,$
where $\mft \subset \mfg$ is a Cartan subalgebra and $[\mft, \mfm] \subset \mfm$. Assume also that $\mft \subset C^\infty_{\inv}(M)$. Then we obtain a contraction of semiclassical projective quasi-representations of $\mfg$ to a genuine representation of $\mfg_\infty = \mft \ltimes \mfm_0$, where $\mfm_0$ is the abelian Lie algebra with the same underlying vector space as $\mfm$. \end{example}
Notably, at least in the setting of Example \ref{cpct_grp_ex}, the contraction specified in Theorem \ref{contr_thm} may also be observed at the level of group representation theory, in the following sense (see Sect. \ref{S2_sect} for a detailed example).
\begin{corollary}\label{grp_matr_coef} In the setting of Theorem \ref{contr_thm}, for every $\beta \in \RR$ and $g \in \mfg$, it holds that\footnote{The sequence $-i (\xi_k \circ t_{k,R})(g)$ converges to $-i \xi_\infty^R(g)$ in the sense of strong resolvent convergence (by Theorem \ref{contr_thm} combined with \cite{oliveira}, Proposition 10.1.18). This implies that $e^{\beta (\xi_k \circ t_{k,R})(g)} \sto e^{\beta \xi_\infty^R(g)}$, see \cite{oliveira}, Proposition 10.1.8.} $$e^{\beta (\xi_k \circ t_{k,R})(g)} \sto e^{\beta \xi_\infty^R(g)}.$$ Consider now the setting of Example \ref{cpct_grp_ex} and assume that $\mfg_\infty \simeq Lie(G_\infty)$ for some finite dimensional Lie group $G_\infty$, and that $\xi_\infty^R = d\Xi_\infty^R$, where $$\Xi^R_\infty : G_\infty \to \mcu(L^2(\TT^n))$$ is a unitary representation. If $g \in \mfm_0$, then $$\Xi_\infty^R(e^{\beta g}) =e^{\beta \xi_\infty^R(g)}= \mcm_{e^{i \beta R g_0}},\ g_0 = g\big{|}_{\Lambda_{a_0}}.$$ Finally, assume that $\Xi_k : G \to \mcu(\mch_k^{\TT^n})$ is a genuine representation such that\footnote{Here we view $\xi_k$ as a map $\xi_k : \mfg \to \End(\mch_k^{\TT^n})$.} for every $g \in \mfm$, $$\Vert \xi_k(g)- d\Xi_k(g)\Vert_{\op} = \mco(1).$$ Then, for $g \in \mfm$, we obtain (using \cite{charles}, Proposition 12) $$\Xi_k(e^{\beta t_{k,R} (g)}) = [T_k(e^{i \beta R g})]^{\TT^n} + \mco(k^{-1}).$$ In terms of matrix coefficients (see Remark \ref{matr_coef_rk}), for any $l,m \in \ZZ^n$ fixed, $$\langle \Xi_k(e^{\beta t_{k,R}(g)}) e_l, e_m \rangle = \langle e^{iR \beta g} s_{k,l}, s_{k,m} \rangle + \mco(k^{-1}) =\widehat{e^{i \beta R g_0}}(m-l) +\mco(k^{-1}).$$This allows to reproduce certain known estimates involving special functions (e.g., see Sect. \ref{S2_sect}).\end{corollary}
\subsection{Example: The contraction \texorpdfstring{$\mathfrak{su}(2) \to \mathfrak{e}(2)$}{su(2) to e(2)}}\label{S2_sect}
In the present subsection, we demonstrate Theorem \ref{contr_thm} (or more specifically, Example \ref{cpct_grp_ex}) and Corollary \ref{grp_matr_coef} in the setting of Berezin--Toeplitz quantization of the unit sphere $$M = S^2 = \{x_1^2 + x_2^2 +x_3^2 = 1\} \subset \RR^3.$$
We identify $S^2$ with $\CC P^1$, equipped with the Fubini-Study form $\omega_{FS}$ of total area $2\pi$, and consider $$\mfg = \mathfrak{su}(2)\simeq \Sp\{u_1, u_2, u_3\},\ u_l = -\frac {x_l} 2,\ l = 1,2,3,$$
with the decomposition $$\mathfrak{su}(2) = \mft \oplus \mfm,\ \mft = \Sp\{u_3\},\ \mfm = \Sp\{u_1, u_2\}.$$
The functions $u_1, u_2, u_3$ satisfy the commutation relations $$\{u_1, u_2\} = u_3,\ \{u_2, u_3\} = u_1,\ \{u_3, u_1\} = u_2.$$ For any $R \ne 0$, the map $t_{k,R}$ specified by $$t_{k,R}(u_3) = u_3,\ t_{k,R}(u_1) = \frac R k u_1,\ t_{k,R}(u_2) = \frac R k u_2,$$ produces a Lie algebra contraction $\mathfrak{su}(2)\xrightarrow{t_{k,R}} \mathfrak{e}(2)$, where the latter is the Lie algebra of the special Euclidean group $SE(2)$.

Let $L = \mco(1)$ denote the dual of the tautological line bundle on $\CC P^1$. Then $L \to M$ is a prequantum line bundle, and the quantum spaces $\mch_k = H^0(M, L^{\otimes k})$ can be identified with $\CC_k[z_0, z_1]$, the space of complex bivariate homogeneous polynomials of degree $k$. In particular, $\dim \mch_k = k+1$.
The operators $$J_{k,l} = -(k+2) T_k(u_l),\ l = 1,2,3$$ are so-called spin operators. Namely, they satisfy the commutation relations (\cite{lefloch}, Exercise 5.2.5, Remark 5.2.6) $$[J_{k,1}, J_{k,2}] = i J_{k,3},\ [J_{k,2}, J_{k,3}] = iJ_{k,1},\ [J_{k,3}, J_{k,1}] = iJ_{k,2},$$
and give rise to an irreducible representation of the Lie algebra $\mathfrak{su}(2)$.

Write $k = 2j$, $j \in \frac 1 2 \NN$. For $m \in \{-j, -j+1, ..., j\}$ denote $$e_{k,m} = \sqrt{\frac{k+1}{2\pi} \binom {k} {j+m}} z_0^{j+m} z_1^{j-m}.$$
Then $\{e_{k,m}\ | \ m = -j, -j+1, ..., j\}$ is an orthonormal basis of $\mch_k$ such that $$J_{k,3} e_{k,m} = m e_{k,m}.$$
The spectrum of $T_k(u_3)$ is given by $$\sigma_k = \{\lambda_{k,m} \ | \ m = -j, ..., j\},\ \lambda_{k,m} = -\frac m {k+2}.$$
Fix a regular value $a_0 \in \big{(}-\frac 1 2, \frac 1 2 \big{)}$ of $u_3$, and let $m_k \in \{-j, -j+1, ..., j\}$ satisfy that $\lambda_{k, m_k}$ is an element of $\sigma_k$ nearest to $a_0$. We can identify $\mch_k$ with a subspace of $L^2(\TT)$ by the assignment $$e_{k,m_k + p} \mapsto e_p \in L^2(\TT),\ p = -j -m_k, -j -m_k +1,.., j-m_k.$$ The semiclassical projective quasi-representation of Theorem \ref{contr_thm} is given by $$\begin{aligned} &\xi_k(u_3) = ik\Big{(}[T_k(u_3)]^{\TT} +\frac{m_k}{k+2}\Id\Big{)},\\ &\xi_k(u_1) = ik [T_k(u_1)]^{\TT},\ \xi_k(u_2) = ik[T_k(u_2)]^{\TT}.\end{aligned}$$
Fix action-angle variables $y \in (-\pi, \pi)$, $\theta \pmod 1$ such that $\omega_{FS} = dy \wedge d\theta$. Then $$u_1 = -\frac{\cos(2\pi \theta)}2 \sqrt{1-4u_3^2},\ u_2 = -\frac{\sin(2\pi \theta)}2 \sqrt{1-4u_3^2},\ u_3 = -\frac y {2\pi}.$$
\begin{corollary}\label{su2e2_cor} By Theorem \ref{contr_thm}, we obtain via $t_{k,R}$ a contraction of the semiclassical projective quasi-representation $\xi_k : \mathfrak{su}(2)\to \End(C^\infty(\TT))$ to the irreducible representation of the Euclidean Lie algebra $\mathfrak{e}(2)$ which is generated by the skew-Hermitian operators $$\begin{aligned} &\mcx_{R,1} = -i \frac R 2 \sqrt{1-4a_0^2}\mcm_{\cos(2\pi \theta)},\\ &\mcx_{R,2} = -i \frac R 2 \sqrt{1-4a_0^2}\mcm_{\sin(2\pi \theta)},\\ &\mcx_{R,3} = -\frac 1 {2\pi} \partial_{\theta}.\end{aligned}$$\end{corollary}
One readily verifies that $$[\mcx_{R,1}, \mcx_{R,2}] = 0,\ [\mcx_{R,2}, \mcx_{R,3}] = \mcx_{R,1},\ [\mcx_{R,3}, \mcx_{R,1}] = \mcx_{R,2},$$
as required. Note that every irreducible unitary representation of the special Euclidean group $SE(2)$ of dimension greater than $1$ is unitarily equivalent to a representation generated by $\mcx_{R,1}$, $\mcx_{R,2}$, $\mcx_{R,3}$ for some $R > 0$ (\cite{ito}).

Next, we consider the convergence of matrix coefficients of the irreducible unitary representation of $SU(2)$ on $\mch_k$ to those of the corresponding irreducible unitary representation of $SE(2)$ on $L^2(\TT)$. The operator $$\mcr_k(\beta)= e^{-i \beta J_{k,2}},\ \beta \in \RR,$$
is a rotation operator, and the matrix coefficients $$d^j_{m', m}(\beta) = \langle \mcr_k(\beta) e_{k,m}, e_{k,m'} \rangle,\ m,m' \in \{-j, -j+1, ..., j\}$$
are known as Wigner small-d functions. The following well-known asymptotic estimate is a special case of Corollary \ref{matr_coef_cor}, for group matrix coefficients, as detailed in Corollary \ref{grp_matr_coef}. 
\begin{lemma} Fix $p, p' \in \ZZ$ and for $q \in \ZZ$ let $J_q$ denote the Bessel function of the 1st kind of order $q$. Then for $\beta \in \RR$, $$d^j_{m_k + p', m_k + p}\Big{(}\frac \beta j \Big{)} = (-1)^{p'-p} J_{p'-p}\Big{(}\beta \sqrt{1-4a_0^2}\Big{)} + \mco(k^{-1}).$$ \end{lemma}
\begin{proof}
It suffices to prove for $d^j_{m_k + p', m_k + p}(\frac \beta{j+1})$. Note that $$\mcr_k\Big{(}\frac \beta {j+1}\Big{)} = e^{-i\beta T_k(x_2)} = T_k(e^{-i \beta x_2}) + \mco(k^{-1}),$$
hence for fixed $p, p' \in \ZZ$,
\begin{multline*} d^j_{m_k + p', m_k + p}\Big{(}\frac \beta{j+1}\Big{)} = \langle e^{-i \beta x_2} e_{k, m_k + p}, e_{k,m_k + p'} \rangle + \mco(k^{-1}) \\= \widehat{e^{-i \beta x_{2,0}}}(p'-p) + \mco(k^{-1}),\end{multline*}
where $$x_{2,0}(\theta) = \sqrt{1-4a_0^2} \sin(2\pi \theta).$$
Next, writing $q = p'-p$, we see that \begin{multline*}\widehat{e^{-i \beta x_{2,0}}}(q) = \int_{-\frac 1 2}^{\frac 1 2} e^{-i \beta \sqrt{1-4a_0^2} \sin(2\pi \theta) -2\pi q \theta i} d\theta \\= 2 \int_0^{\frac 1 2} \cos\Big{(}\beta \sqrt{1-4a_0^2} \sin(2\pi \theta) + 2\pi q \theta\Big{)} d\theta\\ = \frac 1 \pi \int_0^\pi \cos\Big{(}\beta \sqrt{1-4a_0^2} \sin(\tau) + q \tau\Big{)} d\tau \\= J_q\Big{(}-\beta \sqrt{1-4a_0^2} \Big{)} = (-1)^q J_q\Big{(}\beta \sqrt{1-4a_0^2} \Big{)}.\end{multline*} \end{proof}
We conclude the discussion of this subsection by noting that when $a_0 =0$, if $k\in 2\NN$ then $m_k = 0$ (otherwise, $m_k = -\frac 1 2$ or $m_k = \frac 1 2$). Consequently, $$u_l \mapsto \frac{k+2} k \xi_k(u_l) = -i[J_{k,l}]^{\TT}= \xi_k(u_l) + \mco(1),\ l=1,2,3$$
is a genuine representation of $\mathfrak{su}(2)$ which contracts to a representation of $\mathfrak{e}(2)$ as in Corollary \ref{su2e2_cor}. 
The contraction of irreducible representations of $\mathfrak{su}(2)$ to irreducible representations of $\mathfrak{e}(2)$ is a classical example of contraction of Lie algebra representations, already considered in the original work of İnönü and Wigner (\cite{iw}). By considering $M = S^2$, the fixed Liouville torus $\{x_3 = 0\}$, and the subsequence $k \in 2\NN$, Theorem \ref{contr_thm} yields a realization of this contraction via Berezin--Toeplitz quantization. This realization is closely related to the one constructed in \cite{dooley_rice}\footnote{In \cite{dooley_rice}, contractions of representations of $\mathfrak{so}(d+1)$ to those of $\mathfrak{e}(d)$ are obtained for any $d$.}.
\begin{remark}
The example of Theorem \ref{contr_thm} and Corollary \ref{grp_matr_coef} presented in this subsection extends naturally to the setting of representations of $\mathfrak{su}(d+1)$ obtained through the Berezin--Toeplitz quantization of $\mathbb{CP}^{d}$, $d > 1$.
\end{remark}
\section{Bohr-Sommerfeld conditions}\label{bohr_sommer_sect}
In this section we establish Proposition \ref{spec_properties}.
We begin by recalling that according to the Arnold-Liouville theorem, there exists $\delta > 0$, a neighborhood $U_\delta$ of $\Lambda_{a_0}$, and a symplectomorphism
\begin{equation*} \begin{aligned} &\Psi : \left(U_\delta, \omega|_{U_{\delta}}\right) \to (D_\delta\times \TT^n,\ \omega_{\std}),\ D_\delta = (-\delta, \delta)^n,\\ & \omega_{\std} = \sum_{j=1}^n dy^j \wedge d\theta^j,\end{aligned}\end{equation*}
where $y = (y^1, ..., y^n)$ are the coordinates on $D_\delta$ and $\theta = (\theta^1, ..., \theta^n) \mod \ZZ^n$ are periodic coordinates on $\TT^n = \RR^n / \ZZ^n$. Additionally, $$\Psi(\Lambda_{a_0}) = \{0\}\times \TT^n,$$ and there exists a  diffeomorphism $\varphi_0 : D_\delta \to D_{\delta,0} =  f_0(U_\delta)$ such that $$f_0 \circ \Psi^{-1}(y,\theta) = \varphi_0(y).$$
The level sets $$\Lambda_a = \{f_0 = a\},\ a \in D_{\delta, 0}$$
are so-called Lagrangian tori.
Finally, for $j=1,...,n$, the Hamiltonian vector field of $f^j_{0} \circ \Psi^{-1}$ is given by
\begin{equation*} \nu_j(y) \cdot \partial_\theta = \nu_j^1(y) \partial_{\theta^1} + ... + \nu_j^n(y) \partial_{\theta^n},\ \nu_j \in C^\infty(D_\delta, \RR^n),\end{equation*}
and it holds that
\begin{equation*}  \varphi_0'(y)= \nu(y) = \left(\begin{array}{ccc} \text{---} & \nu_1(y) & \text{---} \\   & \vdots & \\ \text{---} & \nu_n(y) & \text{---} \end{array}\right).\end{equation*}

Next, for $a \in D_{\delta,0}$ and $j =1, ...,n$, define the oriented loop $\gamma^j_a$ by $$\Psi(\gamma^j_a(t)) = (\varphi_0^{-1}(a), 0, ..., 0, t, 0, ...,0),\ t\in[0,1].$$
We rely on the following semiclassical \textit{Bohr-Sommerfeld conditions} satisfied by the joint spectrum $\sigma_k$ of $T^1_{k}, ..., T^n_{k}$.
\begin{theorem}[\cite{charles_bohr}, Theorem 3.1]\label{charles_bohr_thm} There exists a formal series $$\sum_{l \ge -1} \hbar^l \rho_l,\ \rho_l \in C^\infty(D_{\delta,0}, \RR^n),$$ such that for any open set $O \subset \RR^n$ with compact closure $\overline{O} \subset D_{\delta,0}$ and any sequences $$(a_{q}, k_q), (a'_q, k_q) \in O \times \NN,\ \lim_{q \to \infty}k_q = \infty,$$
the following holds:
\begin{itemize}
\item{$a_q \in \sigma_{k_q} +\mco(k_q^{-\infty})\Longleftrightarrow \rho(a_q,k_q) \in  k_q^{-1}\ZZ^n + \mco(k_q^{-\infty})$.}
\item{If $a_q, a'_q \in \sigma_{k_q}$ and $a_q = a'_q + \mco(k_q^{-\infty})$, then for every $k_q$ sufficiently large, the multiplicity of $a_q$ is $1$ and $a_q = a'_q$.}
\end{itemize}
Here, $\rho(a, k) \in C^\infty(D_{\delta,0}, \RR^n)$ has the expansion  $\rho(a,k)= k^{-1} \sum_{l \ge -1} k^{-l} \rho_l$, and it holds that $$\rho_{-1} = (\rho^1_{-1}, ..., \rho^n_{-1})$$ with $\rho^l_{-1}(a)$ the holonomy of $\gamma^l_{ a}$ in $L$, that is, the parallel transport along $\gamma^l_{ a}$ is given by multiplication by $e^{2\pi i \rho^l_{-1}(a)}$.\end{theorem}
The matrix $\nu$ introduced above is related to $\rho_{-1}$, as follows.
\begin{lemma}\label{hol_deriv_lemma} The function $\rho_{-1}$ satisfies
\begin{equation*} \rho_{-1}'(a_0) = \frac 1 {2\pi} (\nu(0))^{-1}.\end{equation*} \end{lemma}
\begin{proof}
Let $\varepsilon >0$ be small enough, and let $v_1, ..., v_n$ be the standard basis of $\mathbb R^n$. Fix $l,j\in \{1,...,n\}$ and consider the difference $$r_{j,l,\varepsilon} = \rho^l_{-1}(a_0 + \varepsilon v_j) - \rho^l_{-1}( a_0).$$
Let $[0,y] \subset D_\delta$ denote the interval connecting the points $0,y \in D_\delta$, and write $$y_{j,\varepsilon} =(y_{j,\varepsilon}^1, ..., y_{j,\varepsilon}^n) = \varphi_0^{-1}(a_0 + \varepsilon v_j).$$ Note that $\varphi_0^{-1}( a_0) = 0$, and introduce the cylinder $$C_{j,l, \varepsilon} = [0, y_{j,\varepsilon}] \times \{(0,...,0,t,0,...,0) \in \TT^n\ | \ t\in [0,1]\} \subset D_\delta \times \TT^n,$$
oriented such that $\partial C_{j,l,\varepsilon} = \Psi(\gamma^l_{a_0 + \varepsilon v_j}) - \Psi(\gamma^l_{a_0})$.

Then $$r_{j,l,\varepsilon} = \frac 1{2\pi}\int_{\Psi^{-1}(C_{j,l, \varepsilon})} \omega = \frac 1 {2\pi} \int_{C_{j, l,\varepsilon}} \omega_{\std} = \frac 1 {2\pi} y_{j,\varepsilon}^l.$$
Thus, writing $$\varphi_0^{-1} = \big{(}(\varphi_0^{-1})^1, ..., (\varphi_0^{-1})^n\big{)},$$
we see that $$\frac{\partial \rho^l_{-1}}{\partial a^j}( a_0) = \frac 1 {2\pi} \frac{\partial (\varphi_0^{-1})^l}{\partial a^j}( a_0),$$
and so $$\rho_{-1}'( a_0) = \frac 1 {2\pi}(\varphi_0^{-1})'(\varphi_0(0)) = \frac 1 {2\pi}(\varphi'_0(0))^{-1} = \frac 1 {2\pi}(\nu(0))^{-1}.$$
\end{proof}
We can now readily deduce Proposition \ref{spec_properties}.
\begin{proof}[Proof of Proposition \ref{spec_properties}]
Since $\rho'_{-1}(a_0)$ is invertible, there exists an open neighborhood $O$ of $a_0$, with compact closure $\overline{O} \subset D_{\delta, 0}$, such that $$\rho_{O, -1} = \rho_{-1}\big{|}_O : O \to \rho_{-1}(O)$$ is a diffeomorphism, and for every $k$ large enough $$\rho_O(\cdot, k) = \rho(\cdot, k)\big{|}_O : O \to \rho(O, k)$$ is a diffeomorphism. Fix an open neighborhood $O_1$ of $\rho_{-1}(a_0)$ such that for every sufficiently large $k$, it holds that $\rho(a_0, k) \in O_1$, and the closure $$\overline{O_1} \subset \rho(O, k) \cap \rho_{-1}(O)$$
is compact. Let $r(\cdot, k) = \left(\rho_O(\cdot, k)\right)^{-1}$, and $r_0 = \left(\rho_{O,-1}\right)^{-1}$. Then $$r(\cdot, k) = r_0 + \mco(k^{-1}).$$ Consider $b_{k,0} \in k^{-1} \ZZ^n$ such that $$\dist(\rho( a_0, k), k^{-1} \ZZ^n) = \lvert \rho( a_0,k)-b_{k,0}\rvert.$$ For every $k$ large enough it holds that $b_{k,0} \in O_1$ and so by Theorem \ref{charles_bohr_thm} $$r( b_{k,0},k) \in \sigma_k + \mco(k^{-\infty}),$$
and additionally, $$r(b_{k,0}, k)- a_0 = r(\rho( a_0, k) + \mco(k^{-1}), k) - a_0 = \mco(k^{-1}),$$ which implies that $\dist(a_0, \sigma_k) = \mco(k^{-1})$. That is, the joint eigenvalue $a_{k,0} \in \sigma_k$ closest to $a_0$ satisfies $$a_{k,0} = a_0 + \mco(k^{-1}).$$

Next, fix a non-decreasing sequence $p_k \ge 0$, $p_k = o\big{(}k^{\frac 1 2}\big{)}$ and set $$\mci_{k} = \{m \in \ZZ^n\ | \ |m^j| \le  p_k,\ j = 1,...,n\}.$$ If $m\in \mci_{k}$, then since (by Theorem \ref{charles_bohr_thm}) $$\rho(a_{k,0},k) \in k^{-1}\ZZ^n + \mco(k^{-\infty})$$ and $\rho(a_{k,0}, k) + \frac m k \in O_1$ for all sufficiently large $k$, we see that (by Theorem \ref{charles_bohr_thm})
\begin{equation*} r\Big{(}\rho(a_{k, 0},k) + \frac {m} k, k\Big{)} \in \sigma_k + \mco(k^{-\infty}).\end{equation*}
Also, setting $m_k = \frac {m} k$ (and noting the previous lemma),
\begin{multline*} r\big{(}\rho( a_{k, 0},k) + m_k, k\big{)} - a_{k, 0} = r'(\rho(a_{k, 0},k),k) m_k +\mco(|m_k|^2) \\= r_0'\big{(}\rho_{-1}(a_{k, 0}) + \mco\big{(}k^{-1}\big{)}\big{)}m_k + \mco(|m_k|^{2}) = r_0'(\rho_{-1}(a_{k, 0}))m_k + \mco(|m_k|^{2}) \\= r_0' (\rho_{-1}( a_0))m_k + \mco(|m_k|^{2}) = (\rho_{-1}'( a_0))^{-1}m_k + \mco(|m_k|^{2}) \\= 2\pi \nu(0)  m_k + \mco(|m_k|^{2}).\end{multline*}
Clearly, the $\mco(|m_k|^2)$ remainder can be bounded uniformly for $m \in \mci_k$.

We conclude that there exists $ a_{k, m} \in \sigma_k$ such that $$a_{k, m} + \mco(k^{-\infty}) = r\big{(}\rho( a_{k, 0}, k) + m_k, k\big{)} =  a_{k, 0} + 2\pi  \nu(0) m_k + \mco(|m_k|^{2}),$$
that is,$$a_{k, m} = a_{k, 0} + \frac{2\pi} k \nu(0)  m + \mco\Big{(}\frac{|m|^2}{k^2}\Big{)}.$$
Also, if $k$ is sufficiently large then for all $m \in \mci_k$, the multiplicity of $a_{k,m}$ is $1$ (by Theorem \ref{charles_bohr_thm}).

The equality (for all $k$ large enough) $$\sigma_{k,0} =  \{ a_{k, m} \ | \ m \in \mci_{k}\}$$ also readily follows from Theorem \ref{charles_bohr_thm}. Indeed, otherwise there would exist a sequence $k_q \to \infty$ and a sequence $a_{q} \in \sigma_{k_q, 0} \setminus \{a_{k_q,m} \ | \ m \in \mci_{k_q}\}$. Then $$a_{q} = a_{k_q,0} + \frac \pi {k_q} \nu(0) y_{q},\ y_{q} \in \{y \in D_\delta \ | \ |y^j| \le 2 \lfloor p_{k_q} \rfloor + 1,\ j = 1,...,n\}.$$
In particular (by Theorem \ref{charles_bohr_thm}) there exists $m_{q} \in \ZZ^n$, $m_{q} = o(k_q^{\frac 1 2})$ such that $$\rho(a_{q}, k_q) = \rho(a_{k_q,0}, k_q) + \frac{m_{q}} {k_q} + \mco(k_q^{-\infty}).$$
Then applying $r(\cdot, k_q)$ to both sides, we see that $$a_{k_q,0} + \frac \pi {k_q} \nu(0) y_{q}=a_{q}= a_{k_q,0} + \frac{2\pi} {k_q} \nu(0) m_{q} + o(k_q^{-1}).$$
This implies that $m_{q} \in \mci_{k_q}$ (for $k_q$ large enough), and so $$\rho(a_{q}, k_q) = \rho(a_{k_q,m_{q}}, k_q) + \mco(k_q^{-\infty}),$$ hence $$a_{q} = a_{k_q,m_{q}} + \mco(k_q^{-\infty}),$$ and (by Theorem \ref{charles_bohr_thm}) $a_{q} = a_{k_q,m_{q}}$, a contradiction.

Finally, assume by contradiction that for any $c > 0$, $k_0 \in \NN$ there exists $k > k_0$ and $ a \in \sigma_k \setminus \sigma_{k, 0}$ such that $\dist( a, \sigma_{k, 0}) <\frac c k$. Then we can obtain sequences $k_q \to \infty$ and $ a_{q} \in \sigma_{k_q} \setminus \sigma_{k_q, 0}$ such that $$\dist(a_{q}, \sigma_{k_q, 0}) = o(k_q^{-1}).$$ Let $\hat a_{q} \in \sigma_{k_q, 0}$ be such that $$\dist( a_{q}, \sigma_{k_q, 0}) = | a_{q} -  \hat a_{q}|.$$ Then $$\rho(a_{q}, k_q) \in k_q^{-1} \ZZ^n + \mco(k_q^{-\infty}),\ \rho( \hat a_{q}, k_q) \in k_q^{-1} \ZZ^n + \mco(k_q^{-\infty}),$$
hence $$\rho( a_{q}, k_q) = \rho( \hat a_{q}, k_q) + \frac{m_q}{k_q} + \mco(k_q^{-\infty})$$
for some $m_q \in \ZZ^n$. It follows that for all $q$ large enough, $$|\rho( a_{q}, k_q) - \rho( \hat a_{q}, k_q)| \ge \frac{|m_q|}{2k_q}.$$
On the other hand, $$|\rho( a_{q}, k_q) - \rho( \hat a_{q}, k_q)| = \mco(1) | a_{q} - \hat a_{q}| = o(k_q^{-1}),$$
hence $m_q = 0$ which implies that $a_q = \hat a_q + \mco(k_q^{-\infty})$, but then (by Theorem \ref{charles_bohr_thm}) $a_q = \hat a_q$, a contradiction.
\end{proof}
\section{Proof of Theorem \ref{str_conv_thm}}
In this section, we establish Theorem \ref{str_conv_thm}. We use the fact that under suitable assumptions, convergence of matrix coefficients implies strong convergence. 
\begin{lemma}[\cite{shabtai2}, Sect. 4.1]\label{str_conv_explanation1} If $\mch$ is a separable complex Hilbert space equipped with an orthonormal basis $\mcb = \{v_m \ | \ m \in \NN\}$ and $A_k \in \End(\mch)$ is a uniformly bounded sequence such that for any $m, l \in \NN$, \begin{equation}\label{wk_to_str} \lim_{k \to \infty} \langle A_k v_m, v_l \rangle = \langle A v_m, v_l \rangle,\ \lim_{k \to \infty} \Vert A_k v_m \Vert = \Vert A v_m \Vert,\end{equation}
then $A_k$ converges strongly to $A$.  \end{lemma}
\begin{remark}\label{str_res_conv_exp} If $A_k$, $A$ are (possibly unbounded) self-adjoint operators on $\mch$ such that (\ref{wk_to_str}) is satisfied by the resolvents $R_i(A_k)$, $R_i(A)$, then $A_k$ converges to $A$ in the sense of strong resolvent convergence. In this way, we obtain the convergence mentioned in Remark \ref{str_conv_rk}.\end{remark}
In light of Lemma \ref{str_conv_explanation1}, we see that in order to prove Theorem \ref{str_conv_thm}, it suffices to construct the orthonormal sets $\mcb_k = \{s_{k,m} \ | \ m \in \mci_k\}$ (see Corollary \ref{orth_set}), and then for $s_{k,m}, s_{k,l} \in \mcb_k$ with $m,l \in \ZZ^n$ fixed, compute the limits (\ref{wk_to_str}) for the relevant operators (Propositions \ref{matr_coef_prop1}, \ref{matr_coef_prop2}).

Throughout, $\mu = \frac{|\omega^{\wedge n}|}{n!}$ denotes the Liouville measure.
\subsection{Microsupports of sections and operators}\label{MS_sect}
The \textit{microsupport} (\cite{charles}) of a sequence of holomorphic sections $s_k \in \mch_k$ is defined as follows. First, we call $(s_k)_{k \in \NN}$ \textit{admissible} if there exists $N > 0$ such that $\Vert s_k \Vert = \mco(k^N)$. An admissible sequence $(s_k)_{k \in \NN}$ is called \textit{negligible at $x_0 \in M$} if there exists a neighborhood $V$ of $x_0$ such that $\sup_{x \in V} |s_k(x)|_x = \mco(k^{-\infty})$, where $|\cdot|_x$ is the norm on $L^{\otimes k}_x$ induced by the Hermitian product.
\begin{definition} The microsupport of an admissible sequence $(s_k)_{k \in \NN}$ is the set $\MS(s_k) \subset M$ specified by
\begin{equation*} \MS(s_k) = M \setminus \{x \in M \ | \ (s_k)_{k \in \NN}\text{ is negligible at }x\}.\end{equation*} \end{definition}

The microsupport of an admissible sequence $s_p \in \mch_{k_p}$ is defined similarly (it is always assumed that $\lim_{p \to \infty} k_p = \infty$).

We also consider the microsupport of sequences $$Q_k : C^\infty(M, L^{\otimes k}) \to C^\infty(M, L^{\otimes k})$$ which satisfy $\Pi_k Q_k \Pi_k = Q_k$ (here $\Pi_k : L^2(M, L^{\otimes k}) \to \mch_k$ is the Bergman projection). To this end, equip $M \times M$ with the symplectic form $\pi_1^* \omega - \pi_2^* \omega$, where $\pi_1, \pi_2$ are the projections on the left and right factors. Then $$\MS(Q_k) \subset M \times M$$ is defined to be the microsupport of the sequence of Schwartz kernels of $Q_k$ (assuming that they form an admissible sequence). In this context, we identify the operators $Q_k \in \End(\mch_k)$ and $\Pi_k Q_k \Pi_k : C^\infty(M, L^{\otimes k}) \to C^\infty(M, L^{\otimes k})$. %

The microsupports of Toeplitz operators are subsets of the diagonal $$\Delta_M \subset M \times M.$$ More precisely, fix $g(\cdot, k) \in C^\infty(M)$ admitting the expansion $$g(\cdot, k) \sim \sum_{l=0}^\infty g_l k^{-l}$$ in the $C^\infty$ topology. Then $\MS(T_k(g(\cdot, k))) \subset \Delta_M$. If we identify $\Delta_M$ with $M$, then (\cite{charles})
\begin{equation}\label{ms_formula} \MS(T_k(g(\cdot, k))) = \overline{\cup_{l \ge 0} \supp(g_l)}.\end{equation}

In what follows, we note some basic facts involving the microsupports of spectral projections such as $\Pi_{k,g,b_0} = \mathbbm 1_{(b_0, \infty)}(T_k(g(\cdot, k)))$, where $b_0 \in g_0(M)$ is a regular value of $g_0$. First, we consider a "smooth approximation" of $\Pi_{k,g,b_0}$.
\begin{lemma}\label{smoothed_proj_ms}Let $\varepsilon >0$ and assume that $b_0 \in g_0(M)$ is a regular value of $g_0$. Let $\psi_\varepsilon : \RR \to [0,1]$ be a smooth function such that $\psi_\varepsilon(t) = 0$ for all $t \le b_0 -\varepsilon$, $\psi_\varepsilon(t)= 1$ for all $t \ge b_0$. Then $$\MS\left(\psi_\varepsilon\big(T_{k}(g(\cdot, k))\big)\right) \subset \{g_0 \ge b_0-\varepsilon\}.$$\end{lemma}
\begin{proof}
There exists (\cite{charles}, Proposition 12) $\Psi_{k,\varepsilon} \in C^\infty(M)$ with an expansion $$\Psi_{k,\varepsilon} \sim \sum_{l=0}^\infty k^{-l} \psi_{l,\varepsilon}$$ in the $C^\infty$ topology, such that $$\psi_\varepsilon(T_k(g(\cdot, k)))= T_k(\Psi_{k,\varepsilon})+ \mco(k^{-\infty}).$$
If $g_0(x) <  b_0- \varepsilon$, then $\psi_\varepsilon$ vanishes in a neighborhood of $g_0(x)$, hence (using \cite{charles}, p27) there exists a neighborhood $V$ of $x$ such that $\psi_{l,\varepsilon}\big{|}_V \equiv 0$ for all $l \ge 0$. Thus, in light of (\ref{ms_formula}), $x \not \in \MS\big(\psi_\varepsilon(T_k(g(\cdot,k))) \big)$.
\end{proof}
\begin{corollary}\label{ms_cor} The microsupport of $\Pi_{k,g,b_0} = \mathbbm 1_{(b_0,\infty)}(T_k(g(\cdot, k)))$ satisfies $$\MS(\Pi_{k,g,b_0}) \subset \{g_0 \ge b_0\}\times \{g_0 \ge b_0\}.$$ The same holds for $\mathbbm 1_{[b_0, \infty)}(T_{k}(g(\cdot,k)))$. \end{corollary}
\begin{proof}
We prove for $\Pi_{k,g,b_0}$, but the proof for $\mathbbm 1_{[b_0, \infty)}(T_k(g(\cdot,k)))$ is identical.
Let $\psi_\varepsilon$ be as in Lemma \ref{smoothed_proj_ms}, and write $A_{k,\varepsilon} = \psi_\varepsilon(T_{k}(g(\cdot,k)))$. Then $$\Pi_{k,g,b_0} =A_{k,\varepsilon} \Pi_{k,g,b_0} = \Pi_{k,g,b_0}A_{k,\varepsilon},$$ hence (using \cite{charles}, p24)
\begin{align*} &\MS(\Pi_{k,g,b_0}) = \MS(A_{k,\varepsilon} \Pi_{k,g,b_0})\subset\\ &\{(x_1,x_3) \ | \ \exists x_2 \in M\ \text{such that } (x_1, x_2) \in \MS(A_{k,\varepsilon})\text{ and }(x_2, x_3) \in \MS(\Pi_{k,g,b_0})\}\\ &\subset \MS(A_{k,\varepsilon}) \times M,\end{align*}
where by a slight abuse of notation, in the second line $\MS(A_{k,\varepsilon})$ is viewed as a subset of $M \times M$, and in the third line it is viewed as a subset of $M$. Repeating this argument,
\begin{equation*} \MS(\Pi_{k,g,b_0}) = \MS(\Pi_{k,g,b_0}A_{k,\varepsilon}) \subset M \times \MS(A_{k,\varepsilon}),\end{equation*}
therefore (noting Lemma \ref{smoothed_proj_ms})
\begin{align*} &\MS(\Pi_{k,g,b_0}) \subset \left(\MS(A_{k,\varepsilon}) \times M\right) \cap \left( M \times \MS(A_{k,\varepsilon}) \right)\\  &= \MS(A_{k,\varepsilon}) \times \MS(A_{k,\varepsilon}) \subset \{g_0 \ge b_0 -\varepsilon\} \times \{g_0 \ge b_0 -\varepsilon\}.\end{align*}
Since $\varepsilon>0$ is arbitrary, we obtain the required.
\end{proof}
\begin{corollary} Similarly, $$\MS(\mathbbm 1_{(-\infty, b_0)}(T_k(g(\cdot, k)))) \subset \{g_0 \le b_0\} \times \{g_0 \le b_0\}.$$
The same holds for $\mathbbm 1_{(-\infty, b_0]}(T_k(g(\cdot, k)))$. \end{corollary}
\begin{proof}
Replacing $g(\cdot,k)$ with $-g(\cdot,k)$ and $b_0$ with $-b_0$, if $(-g_0)(x) < -b_0$ or $(-g_0)(y) < -b_0$ then there exists a neighborhood $\mcn \subset M \times M$ of $(x,y)$ such that
\begin{equation*} \sup_{\mcn}\left|\mathbbm 1_{(-b_0, \infty)}(-T_{k}(g(\cdot,k)))\right| = \mco(k^{-\infty}),\end{equation*}
where
\begin{gather*} \mathbbm 1_{(-b_0, \infty)}(-T_{k}(g(\cdot,k))) = \Pi_k - \mathbbm 1_{(-\infty, -b_0]}(-T_{k}(g(\cdot,k))) \\= \Pi_k - \mathbbm 1_{[b_0, \infty)}(T_{k}(g(\cdot,k))),\end{gather*}
that is, \begin{multline*}\sup_{\mcn} \left| \Pi_k - \mathbbm 1_{[b_0, \infty)}(T_{k}(g(\cdot,k)))\right|= \sup_{\mcn}\left|\mathbbm 1_{(-\infty, b_0)}(T_k(g(\cdot, k)))\right| = \mco(k^{-\infty}).\end{multline*}\end{proof}
The microsupport of a Toeplitz operator associated with any $h \in L^\infty(M)$ is a subset of the diagonal $\Delta_M \subset M \times M$. More precisely,
\begin{lemma}\label{L_infty_MS} Let $h \in L^\infty(M)$. Then\footnote{Recall that $\esssupp(h) = M \setminus \bigcup\{V \subset M \ | \ V\text{ is open and }h = 0 \text{ a.e. in }V\}$.} $$\MS(T_k(h)) \subset  \{(x,x) \in \Delta_M \ | \ x \in \esssupp(h)\}.$$ \end{lemma}
\begin{proof}
Recall that $\mu = \frac{|\omega^{\wedge n}|}{n!}$ is the Liouville measure. The Schwartz kernel $\mck_{k,h}$ of $T_k(h)$ is given by $$\mck_{k,h}(x,y) = \int_M h(z) \mck_k(x,z)\cdot \mck_k(z, y) d\mu(z),$$
where $\mck_k$ is the Bergman kernel and the dot stands for contraction with respect to the fiberwise Hermitian metric of $L^{\otimes k}$ (\cite{lefloch}, Proposition 6.3.5). First, the fact that $\MS(T_k(h)) \subset \Delta_M$ may be seen as follows. If $x_0 \ne y_0$, we fix balls $B_{x_0, \delta}$ and $B_{y_0, \delta}$ centered at $x_0$ and $y_0$ (respectively) with sufficiently small radius $\delta > 0$. Then for every $(x,y) \in B_{x_0, \frac 1 2 \delta} \times B_{y_0, \frac 1 2 \delta}$, using the fact that $\MS(\Pi_k) = \Delta_M$, \begin{multline*} |\mck_{k,h}(x,y)|_{(x,y)} \le \int_{B_{x_0, \delta}} |h(z)||\mck_k(x,z)|_{(x,z)} |\mck_k(z,y)|_{(z,y)}d\mu(z) \\ + \int_{B_{y_0, \delta}} |h(z)| |\mck_k(x,z)|_{(x,z)} |\mck_k(z,y)|_{(z,y)} d\mu(z) \\+ \int_{M \setminus (B_{x_0, \delta} \cup B_{y_0, \delta})} |h(z)| |\mck_k(x,z)|_{(x,z)} |\mck_k(z,y)|_{(z,y)} d\mu(z) = \mco(k^{-\infty}),\end{multline*}
where the $\mco(k^{-\infty})$ estimate is uniform over $(x,y) \in B_{x_0,\frac 1 2 \delta} \times B_{y_0, \frac 1 2 \delta}$. That is, $$\sup_{B_{x_0, \frac 1 2 \delta} \times B_{y_0, \frac 1 2 \delta}} |\mck_{k,h}(x,y)| = \mco(k^{-\infty}),$$
therefore $(x_0, y_0) \not \in \MS(T_k(h))$.

Similarly, if $x_0 \not \in \esssupp(h)$ then we fix a small ball $B_{x_0, \delta}$ centered at $x_0$ such that $h = 0$ almost everywhere in $B_{x_0, \delta}$. Then for any $x,y \in B_{x_0,\frac 1 2 \delta}$, $$|\mck_{k,h}(x,y)|_{(x,y)} = \Big{|} \int_{M \setminus B_{x_0, \delta}} h(z) \mck_k(x,z) \cdot \mck_k(z, y)d\mu(z) \Big{|} = \mco(k^{-\infty}),$$
and the $\mco(k^{-\infty})$ is uniform in $x,y \in B_{x_0, \frac 1 2 \delta}$, therefore $(x_0,x_0) \not \in \MS(T_k(h))$. 
\end{proof}
In particular, Lemma \ref{L_infty_MS} allows us to conclude the following.
\begin{corollary}\label{basic_MS_lemma} Let $\delta > 0$ be sufficiently small. Then
\begin{gather*} \Vert T_k(\mathbbm 1_{\{g_0 \ge b_0+ \delta\}}) \Pi_{k, g,b_0} - T_k(\mathbbm 1_{\{g_0 \ge b_0+ \delta\}}) \Vert_{\op} = \mco(k^{-\infty}),\\ \Vert T_k(\mathbbm 1_{\{g_0 \le b_0-\delta\}}) \Pi_{k,g, b_0} \Vert_{\op} = \mco(k^{-\infty}).\end{gather*} \end{corollary}
\begin{proof}
We note that
\begin{equation*} \MS(\Pi_{k, g,b_0}) \subset \{g_0 \ge b_0\} \times \{g_0 \ge b_0\}.\end{equation*}
Hence,
\begin{multline*} \MS(T_k(\mathbbm 1_{\{g_0 < b_0-\delta\}}) \Pi_{k, g,b_0} ) \subset\\ \{(x_1, x_3) \ | \ \exists x_2: \ (x_1, x_2) \in \MS(T_k(\mathbbm 1_{\{g_0 <b_0 -\delta\}})),\ (x_2, x_3) \in \MS(\Pi_{k,g, b_0})\}.\end{multline*}
Assume by contradiction that $ \MS(T_k(\mathbbm 1_{\{g_0 <b_0- \delta\}}) \Pi_{k, g,b_0} )  \ne \emptyset$. Then there exists $x_2\in M$ such that (noting the above lemma)
\begin{multline*} (x_1, x_2) \in \MS(T_k(\mathbbm 1_{\{g_0 <b_0- \delta\}})) \subset\\ \Delta_M \cap \left(\{g_0 \le b_0 -\delta\} \times \{g_0 \le b_0- \delta\}\right) = \{(x, x) \ | \ g_0 (x) \le b_0-\delta\}.\end{multline*}
Thus, $x_1 = x_2$ and $g_0(x_2) \le b_0- \delta$. At the same time,
\begin{equation*} (x_2, x_3) \in \MS(\Pi_{k,g, b_0}) \subset \{g_0 \ge b_0\} \times \{g_0 \ge b_0\},\end{equation*}
therefore $g_0(x_2) \ge b_0$, a contradiction. Similarly,
\begin{multline*} T_k(\mathbbm 1_{\{g_0 > b_0+\delta\}})-T_k(\mathbbm 1_{\{g_0 > b_0+\delta\}}) \Pi_{k, g,b_0}  \\= T_k(\mathbbm 1_{\{g_0 > b_0+\delta\}}) \mathbbm 1_{(-\infty, b_0]}(T_k(g(\cdot,k))) = \mco(k^{-\infty}).\end{multline*}
\end{proof}
\subsection{Lagrangian sections}\label{lagr_sect}
The joint eigensections of $T_k^1, ..., T_k^n$ corresponding to joint eigenvalues near $a_0$ are important examples of so-called Lagrangian sections, as introduced in \cite{charles_bohr}. The symbol calculus of Lagrangian sections is a key ingredient in the proof of Propositions \ref{matr_coef_prop1}, \ref{matr_coef_prop2}. Hence, in what follows we recall some relevant properties and features of Lagrangian sections (we refer the reader to \cite{charles_bohr} for the details). We continue using the notations of Sect. \ref{bohr_sommer_sect}. 

Fix $\delta > 0$ sufficiently small, and denote $\Gamma = \varphi_0\big{(}D_{\frac \delta 2}\big{)}$, where $D_{\frac \delta 2} = \big{(}-\frac \delta 2, \frac \delta 2\big{)}^n$. In what follows, we identify $U_\delta = f_0^{-1}(D_{\delta,0})$ with $D_{\delta,0} \times \TT^n$ as smooth manifolds, so that $f_0^{-1}(a) = \Lambda_a$ is identified with $\{a\} \times \TT^n$. A Lagrangian section $$(s_{p}, a_p, k_p) \in \mch_{k_p} \times \Gamma \times \NN,\ p \ge p_0,$$ is a sequence of triplets, with $s_{p}\in \mch_{k_p}$ an admissible sequence\footnote{We always assume that $\lim_{p \to \infty} k_p = \infty$.} associated with $\Lambda_{a_p} = f_0^{-1}(a_p)$. This means that over sets of the form $D_{\delta, 0} \times V$, where $V \subset \TT^n$ is open and contractible, $s_{p}$ may be represented in a certain prescribed way (\cite{charles_bohr}, eq. (25)). The symbol of the Lagrangian section $(s_{p}, a_p, k_p)$ is a section of a flat $\CC[[\hbar]]$-bundle $K \to \Gamma \times \TT^n$ of rank $1$, i.e., a map $u : \Gamma \times \TT^n \to K$ which is locally of the form $$u(a, \theta) = (a, \theta, u_V(a, \theta, \hbar)),\ u_V \in C^\infty(\Gamma \times V)[[\hbar]].$$
The bundle $K \to \Gamma \times \TT^n$ is equipped with a formal Hermitian form $$\langle \cdot, \cdot \rangle_{K_x} : K_x \times K_x \to \CC[[\hbar]],\ x \in \Gamma \times \TT^n,$$
so if $u,v$ are smooth sections of $K$, then $$\langle u, v \rangle_K \in C^\infty(\Gamma \times \TT^n)[[\hbar]].$$
The inner product of Lagrangian sections and the formal Hermitian product of their symbols are related, as follows.
\begin{proposition}[\cite{charles_bohr}, Proposition 2.6, and \cite{charles_bohr}, eq. (29)]\label{inner_prod_lagr} Let $(s_{p}, a_p, k_p)$ and $(s'_{p}, a_p, k_p)$ be Lagrangian sections with the same sequences $(a_p, k_p) \in \Gamma \times \NN$, and with symbols $u,u' : \Gamma \times \TT^n \to K$. Denote $$\langle u, u' \rangle_K = \langle u, u' \rangle_{K,0} + \mco(\hbar),\ \langle u, u' \rangle_{K,0} \in C^\infty(\Gamma \times \TT^n).$$ Then $$\langle s_{p}, s'_{p} \rangle = \int_{\Lambda_{a_p}} \langle u, u' \rangle_{K,0}(a_p, \theta)d \mu_{\Lambda_{a_p}}(\theta) + \mco(k_p^{-1}).$$
Here, $\mu_{\Lambda_a}$ is the measure on $\Lambda_a$ induced from the Riemannian structure of $M$.\end{proposition}
We also use the following fact about microsupports of Lagrangian sections.
\begin{lemma}\label{lagr_ms} Assume that $D \subset D_{\delta,0}$ is compact, and that $(s_{p}, a_p, k_p)$ is a Lagrangian section such that $a_p \in D$ for every sufficiently large $p$. Then $\MS(s_{p}) \subset f_0^{-1}(D)$. In particular, if $\lim_{p \to \infty} a_p = a_0$, then $\MS(s_{p}) \subset \Lambda_{a_0}$. \end{lemma}
\begin{proof}
This follows directly from the local representation of $s_{p}$ (\cite{charles_bohr}, eq. (25)).

Alternatively, if $x \not \in f_0^{-1}(D)$, then we may choose a smooth bump function $\chi_x : M \to [0,1]$ which equals $1$ near $x$, and which is compactly supported in an open ball $B_x$ such that $\overline{B_x} \cap f_0^{-1}(D) = \emptyset$. Then $T_k(\chi_x)$ is a Toeplitz operator which is elliptic at $x$, hence by \cite{charles}, Proposition 9, it suffices to show that $T_{k_p}(\chi_x) s_{p}$ is negligible in order to conclude that $x \not \in \MS(s_{p})$.
Noting \cite{charles_bohr}, Proposition 2.7 and \cite{charles_bohr}, p. 1547, we see that $(T_{k_p}(\chi_x) s_{p}, a_p, k_p)$ is a Lagrangian section whose symbol vanishes in an open neighborhood of $D \times \TT^n \subset \Gamma \times \TT^n$. This implies that $T_{k_p}(\chi_x) s_{p} = \mco(k_p^{-\infty})$, as needed.
\end{proof}
In the remainder of this subsection, we specify certain Lagrangian sections of the form $(\tilde s_{k,m}, a_{k,0}, k)$ which approximate the joint eigensections of $T_k^1, ..., T_k^n$ corresponding to $a_{k,m}\in \sigma_k$, where $m \in \ZZ^n$ is fixed. Then, in Sect. \ref{matr_coef_sect1}, \ref{matr_coef_sect2}, we will use these sections to compute the limits of the matrix coefficients specified in Proposition \ref{matr_coef_prop1}, \ref{matr_coef_prop2}.
\begin{proposition}[\cite{charles_bohr}, Proposition 3.8] Let $s_{k,0}$ be a joint unit eigensection of $T^1_{k}$, ..., $T^n_{k}$ corresponding to the joint eigenvalue $a_{k, 0}\in \sigma_k$.
Then there exists a sequence $c_k \in \CC$ such that
\begin{equation*} s_{k,0} = c_k \tilde s_{k,0} + \mco(k^{-\infty}),\end{equation*}
where $(\tilde s_{k,0}, a_{k,0}, k)$ is a Lagrangian section.\end{proposition}
We also need the following.
\begin{lemma}\label{c_k_abs} The sequence $c_k$ satisfies $|c_k|^2 = |\det(\nu(0))| + \mco(k^{-1})$. \end{lemma}
\begin{proof}
Denote $X_{a,j} = X_j\big{|}_{\Lambda_{a}}$, where $X_j$ is the Hamiltonian vector field of $f^j_{0}$. Then (\cite{charles_bohr}, eq. (36))
\begin{equation*} 1 = |c_k|^2\Big{(} \int_{\Lambda_{a_{k,0}}} \eta_{a_{k,0}} + \mco(k^{-1})\Big{)} + \mco(k^{-\infty}),\end{equation*}where $\eta_{a} \in |\Omega|(\Lambda_{ a})$ is the density defined by $\eta_{ a}(X_{ a,1}\wedge ... \wedge X_{ a,n}) = 1$. In the coordinates induced by $\Psi$ (specified in (\ref{al_sympl})),
\begin{equation*} X_{ a,j} = \sum_{l=1}^n \nu^l_{j}(\varphi_0^{-1}( a))\partial_{\theta^l},\end{equation*}
therefore $\eta_{ a} =\frac 1 {|\det \nu(\varphi_0^{-1}( a))|} |d\theta|$ and since $ a_{k, 0} =  a_0 + \mco(k^{-1})$,
\begin{equation*} \int_{\Lambda_{ a_{k,0}}} \eta_{ a_{k,0}} = \frac 1 {|\det (\nu(0))|} + \mco(k^{-1}).\end{equation*}
This readily implies the required.
\end{proof}
The rest of the sections $(\tilde s_{k,m}, a_{k,0},k)$ are obtained from $(\tilde s_{k,0}, a_{k,0}, k)$. Let $u_0 : \Gamma \times \TT^n \to K$ denote the symbol of $(\tilde s_{k,0}, a_{k,0}, k)$. Fix $m \in \ZZ^n$. Then $$u_m = e_m u_0,\ e_m(\theta) = e^{2\pi i m \cdot \theta}$$
is also a section of $K \to \Gamma \times \TT^n$, and as shown in \cite{charles_bohr}, there exists a corresponding Lagrangian section \begin{equation}\label{lagr_qm}(\tilde s_{k,m}, a_{k,0}, k)\end{equation} whose symbol is $u_m$. In light of Proposition \ref{inner_prod_lagr}, $$\Vert \tilde s_{k,m} \Vert = \Vert \tilde s_{k,0} \Vert + \mco(k^{-1}).$$
The sections $\tilde s_{k,m}$ are approximate joint eigensections of $T_k^1, ..., T_k^n$, as follows.
\begin{lemma} Fix $m \in \ZZ^n$. Denote $a_{k, m} = (a^1_{k,m}, ..., a^n_{k,m})$. Then
\begin{equation*} T^j_{k} \tilde s_{k,m} = a^j_{k,m} \tilde s_{k,m} + \mco(k^{-2}),\ j=1,...,n.\end{equation*}\end{lemma}
\begin{proof}
Recall that $u_0 : \Gamma \times \TT^n \to K$ is the symbol of $(\tilde s_{k,0}, a_{k,0}, k)$, and that $u_m = e_m u_0$ is the symbol of $(\tilde s_{k,m}, a_{k,0}, k)$.

Let $V \subset \TT^n$ be any open contractible set. On $\Gamma \times V$, the symbol $u_m$ of $(\tilde s_{k,m}, a_{k,0}, k)$ is given by $$u_m(a, \theta) = (a, \theta, u_{m,V}(a, \theta, \hbar)),\ u_{m,V} \in C^\infty(\Gamma \times V)[[\hbar]],$$
where (since $u_m = e_m u_0$) $$u_{m,V} = e_m u_{0,V}.$$
The Toeplitz operator $T^j_{k}$ induces a map (\cite{charles_bohr}, Sect. 2)$$T^j:C^\infty(\Gamma \times V)[[\hbar]]\to C^\infty(\Gamma \times V)[[\hbar]],\ j=1,...,n.$$
Specifically (\cite{charles_bohr}, Theorem 2.8, cf. \cite{charles_bohr}, eq. (30)), for $v \in C^\infty(\Gamma \times \TT^n)[[\hbar]]$, \begin{equation}\label{symbol_map}T^j v = a^j v - i \hbar (\mcl_{X_{j}}v + F^j v) + \mco(\hbar^2).\end{equation}
Here, $X_j$ is the Hamiltonian vector field of $f_0^j$, and $F^j \in C^\infty(\Gamma \times V)$ is a certain explicitly specified smooth function.

Now, since $T^j_{k} s_{k,0} = a^j_{k,0} s_{k,0},$ and $s_{k,0} = c_k \tilde s_{k,0} + \mco(k^{-\infty})$, it holds that $T^j u_{0,V} = a^j u_{0,V}$ (see \cite{charles_bohr}, Proposition 3.8 and p. 1554).
Thus, we obtain by direct computation (using eq. (\ref{symbol_map})) that $$T^j u_{m,V} = e_m T^j u_{0,V} + \hbar 2\pi m \cdot \nu_j( \varphi_0^{-1}(a)) u_{m,V} + \mco(\hbar^2).$$
That is, $$T^j u_{m,V} =  \big{(}a^j  + \hbar 2\pi m \cdot \nu_j( \varphi_0^{-1}(a))\big{)} u_{m,V} + \mco(\hbar^2).$$
Since $V \subset \TT^n$ is an arbitrary open contractible set, this implies that $$T_k^j \tilde s_{k,m} = \Big{(} a_{k,0}^j + \frac {2\pi m \cdot \nu_j(\varphi_0^{-1}(a_{k,0}))} k \Big{)} \tilde s_{k,m} + \mco(k^{-2}).$$
Finally, noting that $a_{k,0} = a_0 + \mco(k^{-1})$ and Proposition \ref{spec_properties}, we obtain the required.
\end{proof}
We will consider matrix coefficients with respect to the following specific joint eigensections.
\begin{corollary}\label{approx_cor} Fix $m \in \ZZ^n$. Then for all sufficiently large $k$, there exists a joint unit eigensection $s_{k,m} \in \mch_k$ of $T^1_{k}, ..., T^n_{k}$ such that
\begin{equation*} T^j_{k} s_{k,m} = a^j_{k,m} s_{k,m},\ s_{k,m}= \sqrt{|\det\nu(0)|} \tilde s_{k,m} + \mco(k^{-1}).\end{equation*} \end{corollary}
\begin{proof}
Let $s'_{k,m}$ be a joint unit eigensection of $T^1_{k}, ..., T^n_{k}$ with $$T^j_{k} s'_{k,m} = a^j_{k,m} s'_{k,m},$$ and note that the multiplicity of $a_{k,m} = (a^1_{k,m}, ..., a^n_{k,m})$ is $1$. Write
\begin{equation*} \tilde s_{k,m} = c'_{k,m} s'_{k,m} + s_{k,m}^\perp,\ s'_{k,m} \perp s_{k,m}^\perp.\end{equation*}
Then on one hand by the previous lemma, for $j = 1,...,n$, $$T^j_{k} \tilde s_{k,m} = a^j_{k,m} \tilde s_{k,m} + \mco(k^{-2}) = a^j_{k,m} c'_{k,m} s'_{k,m} + a^j_{k,m} s_{k,m}^{\perp} + \mco(k^{-2}),$$
and on the other hand $$T^j_{k} \tilde s_{k,m} = a^j_{k,m} c'_{k,m} s'_{k,m} + T^j_{k} s_{k,m}^\perp.$$
Thus,
\begin{equation*} \left(T^j_{k} - a^j_{k,m} \right) s_{k,m}^\perp = \mco(k^{-2}),\end{equation*}
hence
\begin{equation*} \Vert (T^j_{k} - a^j_{k,m})s^\perp_{k,m} \Vert^2 = \mco(k^{-4}).\end{equation*}
Now, we note that by Proposition \ref{spec_properties}, there exists some $c' > 0$ such that for every $\lambda \in \sigma_k \setminus \{a_{k,m}\}$, it holds that $|\lambda - a_{k,m}| > \frac{c'}k$. Consequently, writing $s_{k,m}^\perp$ using an orthonormal joint eigenbasis of $T^1_{k}, ..., T^n_{k}$, and summing over $j=1,...,n$, we deduce that
\begin{equation*} \Vert s^\perp_{k,m} \Vert = \mco(k^{-1}),\end{equation*}
so that
\begin{equation*} \tilde s_{k,m} = c'_{k,m} s'_{k,m} + \mco(k^{-1}).\end{equation*}
This readily implies the required, once we define $s_{k,m}$ by $s'_{k,m} = \frac{\bar c'_{k,m}}{|c'_{k,m}|}s_{k,m}$ and note that using Lemma \ref{c_k_abs}, $|c'_{k,m}| = \frac 1 {\sqrt{|\det \nu(0)|}} + \mco(k^{-1})$.
\end{proof}
\begin{corollary}\label{orth_set} Given any $\mci_k$ as in Proposition \ref{spec_properties}, for all sufficiently large $k$ we can choose (by a standard diagonal argument) a set of joint unit eigensections $$\mcb_k = \{s_{k,m} \ | \ m \in \mci_k\}$$
such that for any $N \in \NN$, for all sufficiently large $k$, if $m = (m^1, ..., m^n) \in \ZZ^n$ satisfies $|m^j| \le N$, $j =1,...,n$, then $s_{k,m}$ is the joint unit eigensection obtained via Corollary \ref{approx_cor}.\end{corollary}

Finally, we note the following "continuity" property.
\begin{lemma}\label{L_infty_limsup} Assume that $(s_k, a_{k,0}, k)$ is a Lagrangian section with symbol $v : \Gamma \times \TT^n \to K$. Write $|v|^2 = v_0 + \mco(\hbar)$, where $v_0 \in C^\infty(\Gamma \times \TT^n)$, and set $$M_{v_0, a_0} = \sup_{\mathbb T^n} |v_0(a_0, \cdot)|.$$
Fix $h \in L^\infty(M)$, and let $K_{a_0} = \esssupp(h) \cap \Lambda_{a_0}$. Then $$\limsup_k \Vert h s_k \Vert  \le \Vert h \Vert_\infty \sqrt{ M_{v_0, a_0}\mu_{\Lambda_{a_0}}(K_{a_0})}.$$ \end{lemma}
\begin{proof}
Fix $\varepsilon > 0$. Fix $\delta > 0$, and let $\tau : M \to [0,1]$ be a smooth function such that $\tau = 1$ in an open set $B \subset M$ such that $K_{a_0} \subset B$, and $$\int_{\Lambda_{a_0}} \tau^2|_{\Lambda_{a_0}} d\mu_{\Lambda_{a_0}} < \mu_{\Lambda_{a_0}} (K_{a_0}) + \delta.$$ Let $\mu = \frac{|\omega^{\wedge n}|}{n!}$ be the Liouville measure. Then for any  $h_0 \in L^\infty(M)$ and $s \in \mch_k$, $$\Vert h_0 s \Vert^2 = \int_M |h_0|^2 |s|^2 d\mu = \langle |h_0|^2 s, s \rangle = \langle T_k(|h_0|^2) s, s \rangle.$$ In our present case, $$\Vert h s_k \Vert \le \Vert \tau h s_k \Vert + \Vert(1-\tau) h s_k\Vert =  \Vert \tau h s_k \Vert + \sqrt{\langle T_k( (1-\tau)^2 |h |^2)s_k, s_k \rangle}.$$
We claim that $\MS(T_k((1-\tau)^2|h|^2)s_k) = \emptyset$. Indeed, recall that $\MS(s_k) \subset \Lambda_{a_0}$ by Lemma \ref{lagr_ms}.  Also, \begin{multline*}\MS(T_k((1-\tau)^2 |h|^2) s_k) \subset\\ \{x \in M \ | \ \exists y \in \MS(s_k)\text{ with }(x,y) \in \MS(T_k((1-\tau)^2 |h|^2))\}.\end{multline*}
If $y \in \MS(s_k)$ and $(x,y) \in \MS(T_k((1-\tau)^2 |h|^2))$, then noting Lemma \ref{L_infty_MS}, $$x = y,\ \text{and } y \in \Lambda_{a_0},\ \text{and } y \in \esssupp((1-\tau)^2 |h|^2).$$ However (by the general properties of the notion of essential support) \begin{multline*}\esssupp((1-\tau)^2 |h|^2) \subset \esssupp((1-\tau)^2) \cap \esssupp(|h|^2)\\ = \esssupp(1-\tau)\cap \esssupp(h),\end{multline*}
so we obtain that $$y \in \esssupp(1-\tau) \cap \esssupp(h)  \cap \Lambda_{a_0} = \emptyset,$$ which is impossible. Thus, $T_k((1-\tau)^2 |h|^2) s_k \in \mch_k$ is a negligible sequence, therefore $$\Vert T_k((1-\tau)^2 |h|^2) s_k \Vert = \mco(k^{-\infty}).$$

Next, $$\Vert \tau h s_k \Vert^2 = \int_M |h|^2 \tau^2 |s_k|^2 d\mu \le \Vert h \Vert_\infty^2 \Vert \tau s_k \Vert^2 = \Vert h \Vert_\infty^2 \langle T_k(\tau^2) s_k, s_k \rangle.$$
Since $(T_k(\tau^2) s_k, a_{k,0}, k)$ is a Lagrangian section, we apply Proposition \ref{inner_prod_lagr} to obtain \begin{multline*}\langle T_k(\tau^2) s_k, s_k \rangle = \int_{\Lambda_{a_{k,0}}} \tau^2(a_{k,0}, \theta) v_0(a_{k,0}, \theta) d\mu_{\Lambda_{a_{k,0}}}(\theta) + \mco(k^{-1})\\ = \int_{\Lambda_{a_0}} \tau^2(a_0, \theta) v_0(a_0, \theta) d\mu_{\Lambda_{a_0}}(\theta) + \mco(k^{-1})\\ \le M_{v_0,a_0} (\mu_{\Lambda_{a_0}}(K_{a_0}) + \delta) + \mco(k^{-1}).\end{multline*}
We conclude that $$ \Vert h s_k \Vert \le  \Vert h \Vert_\infty  \sqrt{M_{v_0, a_0}(\mu_{\Lambda_{a_0}}(K_{a_0}) + \delta)} + o(1).$$
Choosing $\delta > 0$ sufficiently small, we conclude that for all sufficiently large $k$, $$\Vert h s_k \Vert < \Vert h \Vert_\infty  \sqrt{M_{v_0, a_0}\mu_{\Lambda_{a_0}}(K_{a_0})} + \varepsilon,$$
as required.
\end{proof}
We conclude this subsection with a proof of Examples \ref{loc_example1}, \ref{loc_example2}.
\begin{proof}[Proof of Example \ref{loc_example1}]
Let $h : M \to \RR$ be bounded and measurable, and fix $\lambda\in \RR$, as in Example \ref{loc_example1}. Assume that for a certain $\varepsilon > 0$, $$\lambda > \sup_{\Lambda_{a_0, \varepsilon}}h.$$
Fix a smooth function $\tau : M \to [0,1]$ such that $\supp(\tau) \subset \Lambda_{a_0, \varepsilon}$ and $\tau\big{|}_{\Lambda_{a_0, \frac \varepsilon 2}} = 1$. Then $$\Vert T_{k,a_0}(1-\tau) \Vert_{\op} = \mco(k^{-\infty}).$$
Indeed, by a similar argument to Corollary \ref{ms_cor} (or otherwise using Lemma \ref{lagr_ms} and expressing $s_{k,m}$ using a single Lagrangian section $(\hat s_{k,m}, a_{k,m}, k)$, as in the proof of Lemma \ref{inv_zero_decay}), we obtain that $$\MS(\Pi_{k, a_0}) \subset \Lambda_{a_0} \times \Lambda_{a_0}.$$
Hence, in light of the behaviour of microsupports with respect to compositions of operators, $$\MS(\Pi_{k,a_0} T_k(1-\tau) \Pi_{k,a_0}) = \emptyset,$$
which implies the required. Similarly, $$\Vert T_{k,a_0}((1-\tau)(\lambda - h))\Vert_{\op} = \mco(k^{-\infty}).$$ Thus, if $s_k \in \mch_{k,a_0}$ is a unit section, then for some $\delta > 0$, \begin{multline*} \langle(\lambda - T_{k,a_0}(h))s_k, s_k \rangle = \langle T_{k,a_0}\big{(}(1-\tau)(\lambda-h)\big{)}s_k, s_k \rangle + \langle T_{k,a_0}\big{(}\tau (\lambda - h)\big{)}s_k, s_k \rangle \\ \ge \langle T_{k,a_0}(\tau(\lambda-h) )s_k, s_k \rangle - \Vert T_{k,a_0}\big{(}(1-\tau)(\lambda - h)\big{)}\Vert_{\op}  \\= \int_M \tau(\lambda - h)|s_k|^2 d\mu - \Vert T_{k,a_0}\big{(}(1-\tau)(\lambda - h)\big{)}\Vert_{\op}\\ =  \int_{\Lambda_{a_0, \varepsilon}} \tau(\lambda - h)|s_k|^2 d\mu- \Vert T_{k,a_0}\big{(}(1-\tau)(\lambda - h)\big{)}\Vert_{\op}\\ \ge  \delta \int_{\Lambda_{a_0, \varepsilon}} \tau |s_k|^2 d\mu- \Vert T_{k,a_0}\big{(}(1-\tau)(\lambda - h)\big{)}\Vert_{\op} \\= \delta \langle \tau s_k, s_k \rangle - \Vert T_{k,a_0}\big{(}(1-\tau)(\lambda - h)\big{)}\Vert_{\op}\\ \ge   \delta - \Vert T_{k,a_0}\big{(}(1-\tau)(\lambda - h)\big{)}\Vert_{\op} - \Vert T_{k,a_0}(1-\tau)\Vert_{\op} = \delta + \mco(k^{-\infty}).  \end{multline*}
In particular, for every sufficiently large $k$, for any unit section $s_k \in \mch_{k,a_0}$, it holds that $$\langle (\lambda - T_{k,a_0}(h)) s_k, s_k \rangle > \frac \delta 2,$$
which implies that $\lambda \not \in \sigma(T_{k,a_0}(h))$. The case $\lambda < \inf_{\Lambda_{a_0, \varepsilon}}h$ follows by replacing $h$, $\lambda$ with $-h$, $-\lambda$.
\end{proof}
\begin{proof}[Proof of Example \ref{loc_example2}]
We note that $[T_{k,a_0}(h)]^{\TT^n} \sto \mcm_{h|_{\Lambda_{a_0}}}$, since $$[T_{k,a_0}(h)]^{\TT^n} = \Pi_{k,a_0}^{\TT^n} [T_k(h)]^{\TT^n} \Pi_{k,a_0}^{\TT^n},$$
and clearly $\Pi_{k,a_0}^{\TT^n} \sto \Id_{L^2(\TT^n)}$. In light of Corollary \ref{str_conv_spec_imp}, we conclude that $$\lim_{k \to \infty}\sup_{\lambda \in h(\Lambda_{a_0})} \dist\big{(}\lambda, \sigma(T_{k,a_0}(h))\big{)} =0.$$
In the other direction, we note that Example \ref{loc_example1} implies that for any $\varepsilon > 0$, for any sufficiently large $k$, it holds that $$\sigma(T_{k,a_0}(h)) \subset [\min_{\Lambda_{a_0}} h - \varepsilon, \max_{\Lambda_{a_0}} h + \varepsilon].$$
\end{proof}
\subsection{Matrix coefficients of Toeplitz operators}\label{matr_coef_sect1}
In this subsection, we establish the following proposition.
\begin{proposition}\label{matr_coef_prop1} Let $A \subset M$ be an open set such that $\partial A \cap \Lambda_{a_0}$ has measure zero in $\Lambda_{a_0}$. Let $h \in C(M)$ be a continuous function, and write $$h_A = h \mathbbm 1_A,\ h_{A, 0} = h_A\big{|}_{\Lambda_{a_0}}.$$ We identify $\Lambda_{a_0}$ with $\TT^n$,  so that $h_{A,0} \in L^\infty(\TT^n)$. Fix $m, l \in \ZZ^n$ and let $s_{k,m}$, $s_{k,l}$ be the joint eigensections specified in Corollary \ref{orth_set}. Then $$\begin{aligned} &\lim_{k \to \infty} \langle h_A s_{k,l}, s_{k,m} \rangle = \widehat{h_{A,0}}(m-l),\\ &\lim_{k \to \infty} \Vert T_k(h_A) s_{k,m} \Vert = \Vert h_{A,0} \Vert_{L^2(\TT^n)}.\end{aligned}$$\end{proposition}

We begin by addressing matrix coefficients of Toeplitz operators with smooth symbols. Recall for $m \in \ZZ^n$, $$\widehat{h_0}(m) = \langle h_0, e_m \rangle_{L^2(\TT^n)}$$
is the $m$-th Fourier coefficient of $h_0 \in L^2(\TT^n)$, where $\TT^n$ is equipped with the normalized Lebesgue measure. The following lemma readily follows from \cite{charles}.
\begin{lemma}\label{smooth_appr} Let $h \in C^\infty(M)$, and $m,l \in \ZZ^n$. Let $(\tilde s_{k,m}, a_{k,0}, k)$, $(\tilde s_{k,l}, a_{k,0},k)$ be the Lagrangian sections specified in (\ref{lagr_qm}). Then $$\langle h \tilde s_{k,l}, \tilde s_{k,m} \rangle = \frac 1 {|\det(\nu(0))|} \widehat{h_0}(m-l) + \mco(k^{-1}),$$
where $h_0 = h|_{\Lambda_{a_0}}$ and $\Lambda_{a_0}$ is identified with $\TT^n$. \end{lemma}
\begin{proof}
Recall that the symbols $u_l, u_m : \Gamma \times \TT^n \to K$ of the sections $(\tilde s_{k,l}, a_{k,0}, k)$, $(\tilde s_{k,m}, a_{k,0}, k)$ satisfy $$u_l = e_l u_0,\ u_m = e_m u_0,\ e_l(\theta) = e^{2\pi i l \cdot \theta},\ e_m(\theta) = e^{2\pi i m \cdot \theta},$$
where $u_0 : \Gamma \times \TT^n \to K$ is the symbol of the Lagrangian section $(\tilde s_{k,0}, a_{k,0}, k)$. Note that $$\langle h \tilde s_{k,l}, \tilde s_{k,m} \rangle = \langle T_k(h) \tilde s_{k,l}, \tilde s_{k,m} \rangle,$$
and that $(T_k(h) \tilde s_{k,l}, a_{k,0}, k)$ is a Lagrangian section with symbol $v : \Gamma \times \TT^n \to K$ which satisfies $$v = h|_{\Gamma \times \TT^n} u_l + \mco(\hbar).$$
Thus, in light of Proposition \ref{inner_prod_lagr}, \begin{multline*}\langle h \tilde s_{k,l}, \tilde s_{k,m} \rangle = \int_{\Lambda_{a_{k,0}}} \langle h e_l u_0, e_m u_0 \rangle_{K,0}(a_{k,0}, \theta) d\mu_{\Lambda_{a_{k,0}}}(\theta) + \mco(k^{-1})\\ = \int_{\Lambda_{a_{k,0}}} h(a_{k,0}, \theta) e_{l-m}(\theta) \langle u_0, u_0 \rangle_{K, 0}(a_{k,0}, \theta) d\mu_{\Lambda_{a_{k,0}}}(\theta) + \mco(k^{-1}).\end{multline*}
We proceed as in the proof of Lemma \ref{c_k_abs}. Namely, in light of \cite{charles_bohr}, eq. (36), it holds that $$\langle u_0, u_0 \rangle_{K,0}(a, \theta) d\mu_{\Lambda_a}(\theta) = \eta_a(\theta),$$
with $\eta_a \in |\Omega|(\Lambda_a)$ the density given by $$\eta_a = \frac 1 {|\det(\nu(\varphi_0^{-1}(a)))|} |d\theta|.$$Recall also that (by Proposition \ref{spec_properties}), $a_{k,0} = a_0 + \mco(k^{-1})$. We conclude that \begin{multline*} \int_{\Lambda_{a_{k,0}}} h(a_{k,0}, \theta) e_{l-m}(\theta) \langle u_0, u_0 \rangle_{K, 0}(a_{k,0}, \theta) d\mu_{\Lambda_{a_{k,0}}}(\theta) =\\ \int_{\Lambda_{a_{k,0}}} h(a_{k,0}, \theta) e_{l-m}(\theta) \eta_{a_{k,0}}(\theta) =\\ \frac 1{|\det(\nu(\varphi_0^{-1}(a_{k,0})))|} \int_{\Lambda_{a_{k,0}}} h(a_{k,0}, \theta) e_{l-m}(\theta) |d\theta| = \\ \frac 1 {|\det(\nu(0))|}\int_{\Lambda_{a_0}} h(a_0, \theta) e_{l-m}(\theta) |d\theta| + \mco(k^{-1}),\end{multline*}
as needed.
\end{proof}
We also note the following.
\begin{corollary}\label{smooth_prod} Let $h_1, h_2 \in C^\infty(M)$, and denote $$h_0 = h\big{|}_{\Lambda_{a_0}},\ h = h_1 h_2.$$ Then $$ \langle T_k(h_1) \tilde s_{k,l}, T_k(h_2) \tilde s_{k,m} \rangle = \frac 1 {|\det(\nu(0))|} \widehat{h_0}(m-l) + \mco(k^{-1}).$$
This follows from Lemma \ref{smooth_appr}, together with the quasi-multiplicativity property of the mapping $T_k : C^\infty(M) \to \End(\mch_k)$, namely, $$T_k(h_1) T_k(h_2) = T_k(h_1 h_2) + \mco(k^{-1}).$$\end{corollary}
Next, we use Lemma \ref{smooth_appr} and Corollary \ref{smooth_prod} to prove Proposition \ref{matr_coef_prop1} by smooth approximation.
\begin{proof}[Proof of Proposition \ref{matr_coef_prop1}] Fix an open set $A \subset M$ such that $\mu_{\Lambda_{a_0}}(\partial A \cap \Lambda_{a_0}) = 0$ (where $\mu_{\Lambda_{a_0}}$ is as in Proposition \ref{inner_prod_lagr}), and denote $$A_0 = A \cap \Lambda_{a_0}.$$
Let $h \in C(M)$, and define $$h_{A} = h \mathbbm 1_{A},\ h_{A,0} = h_A\big{|}_{\Lambda_{a_0}}.$$ 
Note that we may assume without loss of generality that $h \in C^\infty(M)$. Indeed, if $h$ is only continuous, then given $\varepsilon > 0$ we can choose $h_\varepsilon \in C^\infty(M)$ such that $$\max_M |h - h_\varepsilon| < \frac \varepsilon 4.$$
Then writing $$h_{A, \varepsilon} = h_\varepsilon \mathbbm 1_A,\ h_{A, \varepsilon, 0} = h_{A, \varepsilon} \big{|}_{\Lambda_{a_0}},\ h_{k,A}(m,l) = \langle h_A s_{k,l}, s_{k,m} \rangle,$$
where $l,m \in \ZZ^n$ are fixed, we see that \begin{multline*} | h_{k,A}(m,l) - \widehat{h_{A,0}}(m-l) | \le |\langle (h_A - h_{A, \varepsilon}) s_{k,l}, s_{k,m} \rangle| +\\ |\langle h_{A, \varepsilon} s_{k,l}, s_{k,m} \rangle - \widehat{h_{A, \varepsilon, 0}}(m-l) | + |\widehat{h_{A, \varepsilon, 0}}(m-l) - \widehat{h_{A, 0}}(m-l)|,\end{multline*}
so using $h_\varepsilon$, we conclude that for every sufficiently large $k$,$$|h_{k,A}(m,l) - \widehat{h_{A, 0}}(m-l)| < \varepsilon,$$
that is, $$\lim_{k \to \infty} h_{k,A}(m,l) = \widehat{h_{A,0}}(m-l)$$ as required. Similarly for $\lim_{k \to \infty} \Vert T_k(h_A) s_{k,m} \Vert$. Thus, we assume from now on that $h \in C^\infty(M)$.

We construct smooth approximations of $h_A$, as follows. For $\delta > 0$ sufficiently small, let $K_\delta \subset A_0$ be a compact set such that $$\mu_{\Lambda_{a_0}}(A_0 \setminus K_\delta) < \delta.$$ We choose the set $K_\delta$ such that for $\delta' < \delta$ it holds that $K_{\delta'} \supset K_\delta$.

Fix an open subset $V_\delta \subset A$ such that $K_\delta \subset V_\delta$. Let $\tau_\delta : M \to [0,1]$ be a smooth function satisfying $$\supp \tau_\delta \subset A,\ \tau_\delta\big{|}_{V_\delta} = 1.$$
Note that since $\supp \tau_\delta \subset A$, it holds that $\overline{V_\delta} \subset A$. Also, $$(\mathbbm 1_A - \tau_\delta )|_{M \setminus \overline{A}} = 0,\ (\mathbbm 1_A - \tau_\delta)|_{V_\delta} = 0,$$
therefore $$\esssupp(\mathbbm 1_A - \tau_\delta) \subset \overline{A} \cap (M \setminus V_\delta) = \overline{A} \setminus V_\delta.$$ Hence, setting $h_\delta= \tau_\delta h$, we see that $$\esssupp(h_A - h_\delta) \cap \Lambda_{a_0} \subset (\overline{A} \setminus V_\delta) \cap \Lambda_{a_0} = (\overline{A} \cap \Lambda_{a_0}) \setminus V_\delta.$$
Thus, since $\mu_{\Lambda_{a_0}}(\partial A \cap \Lambda_{a_0})=0$, we obtain that \begin{multline}\label{esssupp_delta}\mu_{\Lambda_{a_0}}(\esssupp(h_A - h_\delta) \cap \Lambda_{a_0}) \le \mu_{\Lambda_{a_0}}((\overline{A} \cap \Lambda_{a_0}) \setminus V_\delta)\\ = \mu_{\Lambda_{a_0}}(A_0 \setminus V_\delta) \le \mu_{\Lambda_{a_0}}(A_0\setminus K_\delta) < \delta.\end{multline}
Next, identifying $\Lambda_{a_0}$ with $\TT^n$ and fixing $m, l \in \ZZ^n$, we consider the Fourier coefficients $$\widehat{h_{\delta,0}}(m-l) = \int_{\TT^n} h_\delta(a_0, \theta) e_{l-m}(\theta) |d\theta|.$$
Since for $0< \delta' < \delta$ it holds that $K_{\delta'} \supset K_\delta$, we conclude using the dominant convergence theorem that \begin{equation}\label{dom_conv}\lim_{\delta \to 0^+}\widehat{h_{\delta,0}}(m-l) = \widehat{h_{A,0}}(m-l).\end{equation}

Now, we consider
\begin{equation*}  h_{k,A}(m,l) =  \langle h_{A} s_{k,l}, s_{k,m} \rangle,\end{equation*}
where $s_{k,m}, s_{k,l}$ are as in Corollary \ref{approx_cor}. Let $\varepsilon > 0$, and choose $\delta > 0$ sufficiently small such that (using Lemma \ref{L_infty_limsup} together with the bound (\ref{esssupp_delta})) \begin{equation}\label{eps_bd1}|\langle(h_A - h_\delta) \tilde s_{k,l}, \tilde s_{k,m} \rangle| < \frac \varepsilon 4.\end{equation}
Additionally, noting eq. (\ref{dom_conv}), we assume that $\delta > 0$ is sufficiently small so that \begin{equation}\label{eps_bd2} \frac 1 {|\det(\nu(0))|} \lvert \widehat{h_{A,0}}(m-l) - \widehat{h_{\delta, 0}}(m-l) \rvert < \frac \varepsilon 4.\end{equation}
In light of Corollary \ref{approx_cor}, it holds that 
$$\lim_{k \to \infty} h_{k,A}(m,l) = |\det(\nu(0))| \lim_{k \to \infty}\langle h_{A}\tilde s_{k,l}, \tilde s_{k,m} \rangle.$$
Also, in light of Lemma \ref{smooth_appr},\begin{multline*}\langle h_A \tilde s_{k,l}, \tilde s_{k,m} \rangle = \langle (h_A - h_\delta) \tilde s_{k,l}, \tilde s_{k,m} \rangle + \langle h_\delta \tilde s_{k,l}, \tilde s_{k,m} \rangle \\= \langle (h_A - h_\delta) \tilde s_{k,l}, \tilde s_{k,m} \rangle + \frac 1 {|\det(\nu(0))|} \widehat{h_{\delta, 0}}(m-l) + \mco(k^{-1}).\end{multline*}
Thus, combining eq. (\ref{eps_bd1}), (\ref{eps_bd2}), we see that for every sufficiently large $k$, $$\Big{|}\langle h_A \tilde s_{k,l}, \tilde s_{k,m} \rangle - \frac 1 {|\det(\nu(0))|} \widehat{h_{A,0}}(m-l) \Big{|} < \frac \varepsilon 2 + \mco(k^{-1}) < \varepsilon.$$
That is, $$\lim_{k \to \infty} h_{k,A}(m,l) = \widehat{h_{A,0}}(m-l),$$
as needed.

The proof that $\lim_{k \to \infty} \Vert T_k(h_A) s_{k,m} \Vert = \Vert h_{A,0} \Vert_{L^2(\TT^n)}$ is similar. Namely, $$\lim_{k \to \infty} \Vert T_k(h_A) s_{k,m} \Vert^2 = |\det(\nu(0))|\lim_{k \to \infty} \Vert T_k(h_A) \tilde s_{k,m} \Vert^2,$$
and
\begin{multline*} \Vert T_k(h_A) \tilde s_{k,m} \Vert^2 =\\ \langle (h_A - h_\delta) \tilde s_{k,m}, T_k(h_A) \tilde s_{k,m} \rangle + \langle T_k( h_\delta) \tilde s_{k,m}, (h_A - h_\delta) \tilde s_{k,m} \rangle + \Vert T_k(h_\delta) \tilde s_{k,m} \Vert^2.\end{multline*}
Then using Lemma \ref{L_infty_limsup} together with (\ref{esssupp_delta}), and using the dominant convergence theorem together with Corollary \ref{smooth_prod}, we can choose $\delta > 0$ small enough such that for any large enough $k$, $$\Big{|}\Vert T_k(h_A) \tilde s_{k,m} \Vert^2 - \frac 1 {|\det(\nu(0))|} \Vert h_{A, 0}\Vert^2_{L^2(\TT^n)}\Big{|} <\varepsilon.$$
This implies the required. \end{proof}
\subsection{Matrix coefficients of spectral projections}\label{matr_coef_sect2}
In what follows, we prove (again, via smooth approximation) the following proposition.
\begin{proposition}\label{matr_coef_prop2} Let $g(\cdot, k) \in C^\infty(M)$, $b_0 \in g_0(M)$ be as in Theorem \ref{str_conv_thm}. Denote $$\Pi_{k,g,b_0} = \mathbbm 1_{(b_0, \infty)}\big{(}T_k(g(\cdot,k))\big{)},\ A_0 = \{g_0 > b_0\} \cap \Lambda_{a_0}.$$ Identify $\Lambda_{a_0}$ with $\TT^n$, so that $\mathbbm 1_{A_0} \in L^\infty(\TT^n)$. Fix $m, l \in \ZZ^n$ and let $s_{k,m}$, $s_{k,l}$ be the joint eigensections specified in Corollary \ref{orth_set}. Then $$\lim_{k \to \infty} \langle \Pi_{k, g, b_0} s_{k,l}, s_{k,m} \rangle =  \widehat{\mathbbm 1_{A_0}}(m-l).$$\end{proposition}
\begin{proof}[Proof of Proposition \ref{matr_coef_prop2}]
First, we note that in light of Corollary \ref{approx_cor} and Proposition \ref{matr_coef_prop1}, it suffices to prove that $$\lim_{k\to \infty} \langle \Pi_{k,g,b_0} \tilde s_{k,l}, \tilde s_{k,m} \rangle = \lim_{k \to \infty} \langle \mathbbm 1_{\{g_0 > b_0\}} \tilde s_{k,l}, \tilde s_{k,m} \rangle.$$

Taking any sufficiently small $\delta > 0$,
\begin{multline*} \langle \tilde s_{k,l}, \Pi_{k, g, b_0} \tilde s_{k,m} \rangle = \langle \tilde s_{k,l}, T_k(\mathbbm 1_{\{g_0 \ge b_0 + \delta\}}) \Pi_{k,g,b_0} \tilde s_{k,m} \rangle \\ + \langle \tilde s_{k,l}, T_k(\mathbbm 1_{\{|g_0-b_0| < \delta\}}) \Pi_{k,g,b_0} \tilde s_{k,m} \rangle \\+ \langle \tilde s_{k,l},T_k(\mathbbm 1_{\{g_0 \le b_0 - \delta\}}) \Pi_{k,g,b_0} \tilde s_{k,m} \rangle.\end{multline*}
First, in light of Lemma \ref{basic_MS_lemma}, $$|\langle \tilde s_{k,l} ,T_k(\mathbbm 1_{\{g_0 \le b_0 - \delta\}}) \Pi_{k,g,b_0} \tilde s_{k,m} \rangle| = \mco(k^{-\infty})$$
and also \begin{multline*}\langle \tilde s_{k,l}, T_k(\mathbbm 1_{\{g_0 \ge b_0 + \delta\}}) \Pi_{k,g,b_0} \tilde s_{k,m} \rangle= \\ \langle \tilde s_{k,l}, \mathbbm 1_{\{g_0 \ge b_0 + \delta\}} \tilde s_{k,m} \rangle + \mco(k^{-\infty})\\ = \langle \tilde s_{k,l}, \mathbbm 1_{\{g_0 \ge b_0\}} \tilde s_{k,m} \rangle\\ - \langle \tilde s_{k,l}, \mathbbm 1_{\{b_0 \le g_0 \le b_0 + \delta\}}\tilde s_{k,m} \rangle + \mco(k^{-\infty}).\end{multline*}
Thus, \begin{multline*} | \langle \tilde s_{k,l}, \Pi_{k, g,b_0} \tilde s_{k,m} \rangle - \langle \tilde s_{k,l}, \mathbbm 1_{\{g_0 \ge b_0\}} \tilde s_{k,m} \rangle | \le \\ |\langle \tilde s_{k,l}, T_k(\mathbbm 1_{\{|g_0-b_0| < \delta\}}) \Pi_{k,g,b_0} \tilde s_{k,m} \rangle|\\ + | \langle \tilde s_{k,l}, \mathbbm 1_{\{b_0 \le g_0 \le b_0 + \delta\}} \tilde s_{k,m} \rangle| + \mco(k^{-\infty}).\end{multline*}
Additionally,
\begin{multline*}|\langle \tilde s_{k,l}, T_k(\mathbbm 1_{\{|g_0-b_0| < \delta\}}) \Pi_{k,g,b_0} \tilde s_{k,m} \rangle| = |\langle \mathbbm 1_{\{|g_0-b_0| < \delta\}} \tilde s_{k,l}, \Pi_{k,g,b_0}\tilde s_{k,m} \rangle|\\ \le \Vert \tilde s_{k,m} \Vert \Vert \mathbbm 1_{\{|g_0-b_0| < \delta\}}\tilde s_{k,l} \Vert,\end{multline*}
and
\begin{multline*} |\langle \tilde s_{k,l}, \mathbbm 1_{\{b_0 \le g_0 \le b_0 + \delta\}} \tilde s_{k,m} \rangle| = |\langle \mathbbm 1_{\{b_0 < g_0 < b_0 + \delta\}} \tilde s_{k,l}, \tilde s_{k,m} \rangle| \\ \le \Vert \tilde s_{k,m} \Vert \Vert \mathbbm 1_{\{|g_0 - b_0| < \delta\}} \tilde s_{k,l} \Vert.\end{multline*}
Putting the above together, we see that \begin{multline*} | \langle \tilde s_{k,l}, \Pi_{k, g,b_0} \tilde s_{k,m} \rangle - \langle \tilde s_{k,l}, \mathbbm 1_{\{g_0 \ge b_0\}} \tilde s_{k,m} \rangle |\\ \le 2 \Vert \tilde s_{k,m} \Vert \Vert \mathbbm 1_{\{|g_0 - b_0| < \delta\}} \tilde s_{k,l} \Vert + \mco(k^{-\infty}).\end{multline*}
Let $\varepsilon > 0$. In light of Lemma \ref{L_infty_limsup}, if $\delta > 0$ is sufficiently small, then for every sufficiently large $k$ it holds that $$| \langle \tilde s_{k,l}, \Pi_{k, g,b_0} \tilde s_{k,m} \rangle - \langle \tilde s_{k,l}, \mathbbm 1_{\{g_0 \ge b_0\}}\tilde s_{k,m} \rangle | < \frac \varepsilon 2 + \mco(k^{-\infty}) < \varepsilon.$$  
This is as required.
\end{proof}
\section{Proofs of the results in Sect. \ref{pairs_sect}}\label{proj_proof_sect}
\subsection{Proof of Theorem \ref{main_proj_thm} and Corollary \ref{max_cor}}
Recall that for $c= (c^1, .., c^n) \in \RR^n\setminus \{0\}$, we set $$S_k^c = c^1 T_k^1 + ... + c^n T_k^n,\ \Pi_k^c = \mathbbm 1_{(c \cdot a_0,\infty)}(S_k^c).$$
The proof of Theorem \ref{main_proj_thm} uses the embeddings of $\mch_k$ into $L^2(\TT^n)$, as specified in Theorem \ref{str_conv_thm}. First, through a convenient choice of the coordinates (\ref{al_sympl}), we obtain the following.
\begin{lemma}\label{coor_lemma} It suffices to prove Theorem \ref{main_proj_thm} for $c \in \RR^n \setminus\{0\}$ such that in the coordinates (\ref{al_sympl}), $$X_{c \cdot f_0}\big{|}_{\Lambda_{a_0}} = \partial_{\theta^n}.$$ \end{lemma}
\begin{proof}
Let $c \in \RR^n \setminus\{0\}$ be an arbitrary vector such that the flow of $X_{c \cdot f_0} \big{|}_{\Lambda_{a_0}}$ is periodic. Assume that in the coordinates (\ref{al_sympl}), $$X_{c \cdot f_0} \big{|}_{\Lambda_{a_0}} = \eta \cdot \partial_{\theta} = \eta^1 \partial_{\theta^1} + ... + \eta^n \partial_{\theta^n},\ \eta \in \RR^n \setminus \{0\}.$$ Then there exists a minimal period $\tau_0 > 0$ such that $$p := \tau_0 \eta \in \ZZ^n.$$  Let $v_n \in \RR^n$ denote the $n$-th standard basis vector. Since $p$ is a primitive vector\footnote{i.e., if $p = (p^1, ..., p^n)$, then $\gcd(p^1, ..., p^n) = 1$.}, there exists $A \in GL(n, \ZZ)$ such that $$A p^{\top} = v_n.$$
Define the symplectomorphism $$ \Psi_A = \Phi_A \circ \Psi : (U, \omega|_U) \to (D_A \times \TT^n, \omega_{\std}),$$
where $ \Phi_A : D \times \TT^n \to D_A \times \TT^n$ is the symplectomorphism given by $$\Phi_A(y, \theta) = (A^{-\top} y, A \theta).$$
Note that $\Psi_A(\Lambda_{a_0}) = \{0\}\times \TT^n$. In the coordinates $(\tilde y, \tilde \theta)$ given by $\Psi_A$, it holds that $$X_{c \cdot f_0}\big{|}_{\Lambda_{a_0}} = \frac 1 {\tau_0} \partial_{\tilde \theta^n}.$$Hence, in the coordinates given by $\Psi_A$, $$X_{\tau_0 c \cdot f_0} \big{|}_{\Lambda_{a_0}} =\partial_{\tilde \theta^n}.$$
Finally, note that $$\mathbbm 1_{(\tau_0 c \cdot a_0, \infty)}(S_k^{\tau_0 c}) = \mathbbm 1_{(c \cdot a_0, \infty)}(S_k^c),$$ hence $N_k^{\tau_0 c}([t, t']) = N_k^c([t, t'])$  (using the notation of Theorem \ref{main_proj_thm}). 
\end{proof}
In light of Lemma \ref{coor_lemma}, we assume throughout this subsection that in the coordinates (\ref{al_sympl}), $$X_{c \cdot f_0}\big{|}_{\Lambda_{a_0}} = \partial_{\theta^n}.$$
Equivalently, we assume that $c$ is the $n$-th row of $\nu(0)^{-1}$, where $\nu(0)$ is the matrix (\ref{freq_matr}).
Consider $g(\cdot, k)$, $b_0$ and $\Pi_{k,g, b_0}$ as in Theorem \ref{str_conv_thm}, and assume that (in the coordinates (\ref{al_sympl})) \begin{equation}\label{boundary_deriv}\partial_{\theta^n} g_0(x) \ne 0\end{equation}
for some $x \in \{g_0 = b_0\} \cap \Lambda_{a_0}$. 
In what follows, we establish a slightly stronger result than Theorem \ref{main_proj_thm}. Namely, we fix an arbitrary sequence of real numbers $$\alpha_k = c \cdot a_0 + \mco(k^{-1}),$$ and consider the spectral projection \begin{equation}\label{more_general_proj}\Pi_{k, \alpha_k} = \mathbbm 1_{(\alpha_k, \infty)}(S_k^c).\end{equation}
Fix a sequence $p_k$ as in Proposition \ref{spec_properties}, and recall the notation $$\mci_{k} = \{m\in \ZZ^n \ | \ |m^j| \le p_k,\ \text{for all }j=1,...,n\}.$$

We begin with the proof of Theorem \ref{main_proj_thm} in the case $n=1$, which is genuinely simpler. Let $\Pi : L^2(\TT) \to H^2(\TT)$ denote the Cauchy-Szeg\H{o} projection on the classical Hardy space of the circle, and assume that $v \in L^\infty(\TT)$ is real-valued. Then $$T_v := \Pi \mcm_v \Pi$$ is a classical Toeplitz operator, and according to the \textit{Hartman-Wintner Theorem} (\cite{douglas}, Theorem 7.20), \begin{equation}\label{hartman_wintner}\sigma(T_{v}) = [\essinf v, \esssup v].\end{equation}
In particular, if $v = \mathbbm 1_{A}$ for a measurable set $A \subset \TT$ such that $0 < \mu_1(A) < 1$, where $\mu_1$ is the normalized Lebesgue measure of $\TT$, then $$\sigma(T_v) = \sigma(T_{\mathbbm 1_A}) = [0,1].$$
\begin{lemma}\label{n=1conv} Assume that $n=1$. There exists a bounded sequence $N_k \in \ZZ$ such that $$\mcm_{k}^* [\Pi_{k, \alpha_k}]^{\TT} \mcm_k \sto \Pi,$$
where $\mcm_k$ is the operator of multiplication by $e_{N_k}(\theta) = e^{2\pi i N_k \theta}$. Consequently, \begin{equation}\label{n=1conveq}\mcm_k^* [\Pi_{k, \alpha_k}]^{\TT} \Pi_{k, g, b_0}^{\TT} [\Pi_{k, \alpha_k}]^{\TT} \mcm_k \sto T_{\mathbbm 1_{A_0}}\end{equation}
using Theorem \ref{str_conv_thm}, where $A_0 = \{g_0|_{\Lambda_{a_0}} > b_0\}$ satisfies $0 < \mu_1(A_0) < 1$.\end{lemma}
\begin{proof}
In light of Proposition \ref{spec_properties}, noting that $\alpha_k = c a_0 + \mco(k^{-1})$, there exists a bounded sequence $N_k \in \ZZ$ such that for any sufficiently large $k$, for any $m \in \mci_k$, $$[\Pi_{k,\alpha_k}]^{\TT}e_m = \left\{\begin{array}{cc} e_m & \text{if }m \ge N_k,\\ 0 & \text{if } m < N_k.\end{array}\right. $$
Consequently, for any sufficiently large $k$, for any $m \in \mci_k$, $$\mcm_k^* [\Pi_{k,\alpha_k}]^{\TT} \mcm_k e_m = \left\{\begin{array}{cc} e_m & \text{if }m \ge 0,\\ 0 & \text{if } m < 0.\end{array}\right.$$
Hence, $\mcm_k^* [\Pi_{k, \alpha_k}]^{\TT} \mcm_k \sto \Pi$. Also, $\Pi_{k,g, b_0}^{\TT} \sto \mcm_{\mathbbm 1_{A_0}}$ by Theorem \ref{str_conv_thm}, which readily implies that $$\mcm_k^* \Pi_{k, g, b_0}^{\TT} \mcm_k \sto \mcm_{\mathbbm 1_{A_0}}.$$
Thus, we obtain (\ref{n=1conveq}). We note that $0 < \mu_1(A_0) < 1$ due to (\ref{boundary_deriv}). \end{proof}
We obtain Theorem \ref{main_proj_thm} in the case $n=1$.
\begin{corollary} In light of Lemma \ref{n=1conv}, the Hartman-Wintner theorem (\ref{hartman_wintner}) and Corollary \ref{str_conv_spec_imp}, we conclude that $$\lim_{k \to \infty} N_k^c([t, t']) = \infty$$ for all $0 \le t < t' \le 1$. \end{corollary}
From now on, we assume that $n > 1$. The proof of Theorem \ref{main_proj_thm} relies on the study of certain operators on $L^2(\mathbb T^n)$. We begin with a few preliminary results. Recall that $$\Pi^n : L^2(\TT^n) \to H_n^2(\TT^n)$$ is the orthogonal projection on $$H^2_n(\TT^n) = \{g \in L^2(\TT^n) \ | \ \widehat{g}(m) = 0 \text{ for all }m \in \ZZ^n\text{ such that } m^n < 0\}.$$  Let $v \in L^\infty(\TT^n)$ be real-valued. The spectrum of operators of the form $$T^n_{v} = \Pi^n \mcm_v \Pi^n\in \End(H^2_n(\TT^n))$$
may be determined as follows.  Let $\mu_{n-1}$ denote the normalized Lebesgue measure on $\TT^{n-1}$, and write $$\theta' = (\theta^1, ..., \theta^{n-1}) \in \TT^{n-1}.$$
\begin{lemma}\label{direct_int_spec} Let $n > 1$ and assume that $v \in L^\infty(\TT^n)$ is real-valued. Then $\lambda \in \sigma(T^n_{v})$ if and only if for every $\varepsilon >0$, $$\mu_{n-1}(\{ \theta' \in \TT^{n-1}\ | \ [\essinf v_{\theta'}, \esssup v_{\theta'}] \cap (\lambda - \varepsilon, \lambda + \varepsilon) \ne \emptyset\}) > 0.$$
Here, for $\mu_{n-1}$-almost every $\theta'$, $v_{\theta'} \in L^\infty(\TT)$ is determined by $v_{\theta'}(\theta^n) = v(\theta)$. \end{lemma}
\begin{proof}
We note that $$L^2(\TT^n) \equiv \int_{\TT^{n-1}}^{\oplus} L^2(\TT) d\mu_{n-1},\ H^2_n(\TT^n) \equiv \int_{\TT^{n-1}}^{\oplus} H^2(\TT) d\mu_{n-1},$$
via the unitary mapping specified by $$e_m \mapsto e_{m',m^n} : \TT^{n-1} \to C^\infty(\TT),\ m = (m', m^n) \in \ZZ^n,$$
where $e_{m',m^n}(\theta')(\theta^n) = e_m(\theta)$. If we identify $H^2_n(\TT^n)$ with the above direct integral, then $$T^n_{v} = \int_{\TT^{n-1}}^{\oplus} T_{v_{\theta'}} d\mu_{n-1}(\theta').$$
Thus, in light of standard results (\cite{reed_simon}, Theorem XIII.85) about the spectrum of self-adjoint decomposable operators, and in light of the case $n=1$ (\ref{hartman_wintner}), we conclude the required. \end{proof}
We would like to apply Lemma \ref{direct_int_spec} to operators with symbols given by indicator functions, and also allow for finite "shifts" of $H^2_n(\TT^n)$, as follows. For $N \in \ZZ$, define $$H^2_{n, N}(\TT^n) = \{ g \in L^2(\TT^n) \ | \ \widehat{g}(m) = 0 \text{ for all }m \in \ZZ^n \text{ such that } m^n < N\},$$
and let $$\Pi^n_N : L^2(\TT^n) \to H^2_{n,N}(\TT^n)$$
be the orthogonal projection. Denote $$T^n_{v,N} = \Pi^n_N \mcm_v \Pi^n_N \in \End(H^2_{n,N}(\TT^n)).$$
\begin{corollary}\label{indicator_spec} Let $n > 1$ and fix $N \in \ZZ$. Assume that $v \in C^\infty(\TT^n)$ is a smooth real-valued function such that $\partial_{\theta^n} v(\theta_0) \ne 0$ for some $\theta_0 \in v^{-1}(b)$. Then $$\sigma(T^n_{\mathbbm 1_{\{v > b\}}, N}) = [0,1].$$ \end{corollary}
\begin{proof}
First, assume that $N = 0$, so that (using the notations introduced above) $$T^n_{\mathbbm 1_{\{v > b\}}, N} = T^n_{\mathbbm 1_{\{v > b\}}}.$$ Recall that $\mu_1$ is the normalized Lebesgue measure on $\TT$. Write $\theta_0 = (\theta_0', \theta^n_{0})$. Then there exists an open neighborhood $D' \subset \TT^{n-1}$ of $\theta_0'$ such that for every $\theta' \in D'$, \begin{equation}\label{positive_measure_spec}\mu_1(\{\theta^n \in \TT \ | \ v(\theta', \theta^n) > b\}) \in (0,1).\end{equation}
Thus, for every $\theta' \in D'$ it holds that $$\sigma(T_{(\mathbbm 1_{\{v > b\}})_{\theta'}}) = [0,1].$$
In light of Lemma \ref{direct_int_spec}, this readily implies the required. Next, if $N \ne 0$, consider the unitary multiplication operator $$\mcm^n_N : L^2(\TT^n) \to L^2(\TT^n),\ u \mapsto e^{2\pi i N \theta^n} u.$$ Then \begin{equation}\label{N_conj}\Pi^n_N = \mcm^n_{N} \Pi^n (\mcm^n_{N})^*,\ T^n_{\mathbbm 1_{\{v > b\}}, N} = \mcm^n_{N} T^n_{\mathbbm 1_{\{v > b\}}} (\mcm^n_{N})^*,\end{equation}
and so $$\sigma(T^n_{\mathbbm 1_{\{v > b\}}, N}) = \sigma(T^n_{\mathbbm 1_{\{v > b\}}}) = [0,1].$$
\end{proof}
The sequences $[\Pi_k^c]^{\TT^n}$ do not seem to necessarily converge strongly (even after conjugation as in Lemma \ref{n=1conv}), due to potential fluctuations of eigenvalues near $c \cdot a_0$. However, such fluctuations are necessarily constrained (if they exist). The next result will allow us to choose approximate eigenvectors which are not very sensitive to such potential fluctuations.
\begin{lemma}\label{Rlambda} Let $n > 1$. Fix $\lambda \in (0,1)$, and let $v$, $b$ be as in Corollary \ref{indicator_spec}. Assume that $N \in \ZZ$, $N > 0$, and consider the operator $R^{n,\lambda}_{N} \in \End(H^2_{n,N}(\TT^n))$ specified by $$\begin{aligned} &R^{n,\lambda}_{N} = (T^{n,\lambda}_N)^2 + (E^n_N)^* E^n_N,\\ &T^{n, \lambda}_N = T^n_{\mathbbm 1_{\{v > b\}}, N} - \lambda \Id_{H^2_{n,N}(\TT^n)},\ E^n_N = (\Pi^n - \Pi^n_N) \mcm_{\mathbbm 1_{\{v > b\}}} \Pi^n_N.\end{aligned}$$
Then there exists $\varepsilon_\lambda > 0$ such that $[0, \varepsilon_\lambda] \subset \sigma(R^{n,\lambda}_{ N})$.\end{lemma}
\begin{proof}
For brevity, write $$R^\lambda_N = R^{n, \lambda}_N,\ E_N = E^n_N.$$Note that $R^\lambda_{N}$ is a decomposable operator with respect to the direct integral decomposition specified in the proof of Lemma \ref{direct_int_spec} (the decomposition was specified for $L^2(\TT^n)$ and $H^2_n(\TT^n)$, but it is applicable to $H^2_{n,N}(\TT^n)$ as well). Specifically, $$R^\lambda_{N} = \int_{\mathbb T^{n-1}} R^\lambda_{N, \theta'} d\mu_{n-1}(\theta'),$$
where $$\begin{aligned} &R^\lambda_{N, \theta'} = (T^{n, \lambda}_{N, \theta'})^2 + E_{N, \theta'}^* E_{N, \theta'},\\ &T^{n, \lambda}_{N, \theta'} = T^n_{(\mathbbm 1_{\{v >b\}})_{\theta'}, N} - \lambda \Id_{H^2_N(\TT)},\ E_{N, \theta'} = (\Pi - \Pi_N) \mcm_{(\mathbbm 1_{\{v > b\}})_{\theta'}} \Pi_N.\end{aligned}$$
Here, $\Pi_N : L^2(\TT) \to H^2_N(\TT)$ is the orthogonal projection on $$H^2_N(\TT) = \{g \in L^2(\TT) \ | \ \hat g(m) = 0 \text{ for all } m < N\}$$
and $\Pi: L^2(\TT) \to H^2(\TT)$ is the orthogonal projection on the usual Hardy space. Thus, $\Pi - \Pi_N$ is a finite rank operator, therefore the same is true for $E^*_{N, \theta'} E_{N, \theta'}$. Since the essential spectrum is not affected by finite rank perturbations (or more generally, compact perturbations), it holds that $$\sigma_{\ess}(R^\lambda_{N,\theta'}) = \sigma_{\ess}((T^{n, \lambda}_{N, \theta'})^2).$$
Note that, as in the proof of Corollary \ref{indicator_spec}, there exists an open set $D' \subset \mathbb T^{n-1}$ such that for every $\theta' \in D'$, it holds that $\sigma(T^{n, \lambda}_{N, \theta'}) = [-\lambda, 1-\lambda]$. In fact, it is straightforward to see (for $N =0$, this follows using \cite{douglas}, Corollary 7.14 or Proposition 7.24; the case $N > 0$ follows via unitary conjugation, as in eq. (\ref{N_conj})) that for all such $\theta' \in D'$, $$\sigma_{\ess}(T^{n, \lambda}_{N, \theta'}) = [-\lambda, 1-\lambda].$$
Thus, by the spectral mapping theorem for essential spectra, $$\sigma_{\ess}((T^{n, \lambda}_{N, \theta'})^2) = [0, \max\{\lambda^2, (1-\lambda)^2\}],$$
and so in light of standard results about the spectra of decomposable self-adjoint operators (\cite{reed_simon}, Theorem XIII.85), there exists $\varepsilon_{\lambda} > 0$ as needed.
\end{proof}
\begin{proposition}\label{main_proj_prop} Let $n > 1$. Fix $\lambda \in (0,1)$ and let $v$, $b$ be as in Corollary \ref{indicator_spec}. Assume that $B_k \in \End(L^2(\TT^n))$ is a sequence of self-adjoint operators such that $$B_k \sto \mcm_{\mathbbm 1_{\{v > b\}}}.$$ Then for every $\delta >0$, for every sufficiently large $k$, $$\sigma([\Pi_{k, \alpha_k}]^{\TT^n} B_k [\Pi_{k, \alpha_k}]^{\TT^n}) \cap (\lambda - \delta, \lambda + \delta) \ne \emptyset.$$ \end{proposition}
\begin{proof}
Consider the joint eigenvalues $\{a_{k,m}\ | \ m \in \mci_k\}$ of $T_k^1$, ..., $T_k^n$ specified in Proposition \ref{spec_properties}. Then the corresponding eigenvalues of $S_k^c$ are $$c \cdot a_{k,m},\ m \in \mci_k.$$
There exists $C > 0$ such that for every sufficiently large $k$, for every $m \in \mci_k$, $$c \cdot a_{k,m} = c \cdot a_{k,0} + \frac{2\pi} k m^n + q_{k,m},$$
where $$|q_{k,m}| < C \frac{|m|^2}{k^2},\ a_{k,0} = a_0 + \mco(k^{-1}).$$
Since $$\alpha_k = c \cdot a_0 + \mco(k^{-1}),$$ we conclude that there exist $N_0, N_1 \in \ZZ$ with $N_0 \le N_1$ such that for every sufficiently large $k$, for every $m \in \mci_k$, if $m^n \ge N_1$ then $c \cdot a_{k,m} > \alpha_k$, and if $m^n < N_0$ then $c \cdot a_{k,m} \le \alpha_k$. That is, for every sufficiently large $k$, if $m^n \ge N_1$ then $$[\Pi_{k,\alpha_k}]^{\TT^n} e_m = e_m,$$
and if $m^n < N_0$ then $$[\Pi_{k, \alpha_k}]^{\TT^n} e_m = 0.$$
Let $\mcm_{N_0}^n$ denote the unitary multiplication operator specified by $$e_m \mapsto e^{2\pi i N_0 \theta^n} e_m.$$Then for every sufficiently large $k$, for every $$m \in \mcj_k= \big{\{}m \in \ZZ^n \ | \ |m^j| \le \frac {p_k} 2,\ j=1,...,n\big{\}},$$ if $m^n <0$ then $$(\mcm_{N_0}^n)^* [\Pi_{k,\alpha_k}]^{\TT^n} \mcm_{N_0}^n e_m = 0,$$
and if $m^n \ge N_1 - N_0$ then $$(\mcm_{N_0}^n)^* [\Pi_{k,\alpha_k}]^{\TT^n} \mcm_{N_0}^n e_m = e_m.$$
Additionally, it holds that $$(\mcm_{N_0}^n)^* B_k \mcm_{N_0}^n \sto \mcm_{\mathbbm 1_{\{v > b\}}}.$$
Since $$\sigma((\mcm_{N_0}^n)^* [\Pi_{k,\alpha_k}]^{\TT^n}B_k[\Pi_{k,\alpha_k}]^{\TT^n} \mcm_{N_0}^n) = \sigma([\Pi_{k, \alpha_k}]^{\TT^n} B_k [\Pi_{k,\alpha_k}]^{\TT^n}),$$
this means that without loss of generality we may assume that $[\Pi_{k,\alpha_k}]^{\TT^n}$ satisfies that there exists $N > 0$ such that for every sufficiently large $k$, for every $m \in \mcj_k$, if $m^n < 0$ then $$[\Pi_{k,\alpha_k}]^{\TT^n}e_m = 0,$$
and if $m^n \ge N$ then $$[\Pi_{k, \alpha_k}]^{\TT^n} e_m = e_m.$$
In light of this, we see that for every fixed $m \in \ZZ^n$ such that $m^n \ge N$, it holds that for every sufficiently large $k$, $$[\Pi_{k, \alpha_k}]^{\TT^n} e_m = e_m.$$
This readily implies that for every $u \in H^2_{n, N}(\TT^n)$, it holds that \begin{equation}\label{nN_str_conv}\lim_{k \to \infty} [\Pi_{k,\alpha_k}]^{\TT^n}u = u.\end{equation}
In the same way, for every $u \in H^2_n(\TT^n)^{\perp}$, it holds that \begin{equation}\label{nperp_str_conv}\lim_{k \to \infty} [\Pi_{k, \alpha_k}]^{\TT^n} u = 0.\end{equation}

Now, fix $\varepsilon > 0$ sufficiently small, and note that by Lemma \ref{Rlambda}, there exists a unit vector $u_\varepsilon \in H^2_{n,N}(\TT^n)$ such that $$\langle R_{N}^{n, \lambda} u_\varepsilon, u_\varepsilon \rangle < \varepsilon^2.$$
That is, \begin{equation}\label{Rlambda_est}\Vert \Pi_N^n \mcm_{\mathbbm 1_{\{v > b\}}} u_{\varepsilon} - \lambda u_{\varepsilon} \Vert < \varepsilon,\ \Vert(\Pi^n - \Pi^n_N) \mcm_{\mathbbm 1_{\{v > b\}}} u_\varepsilon \Vert < \varepsilon.\end{equation}
In what follows, we estimate $$\Vert [\Pi_{k,\alpha_k}]^{\TT^n} B_k [\Pi_{k, \alpha_k}]^{\TT^n} u_\varepsilon - \lambda u_\varepsilon \Vert.$$
Note that since $B_k \sto \mcm_{\mathbbm 1_{\{v > b\}}}$, there exists some $B > 0$ such that $\Vert B_k \Vert_{\op} < B$ for all $k \in \NN$. First,
\begin{multline*}\Vert [\Pi_{k,\alpha_k}]^{\TT^n} B_k [\Pi_{k, \alpha_k}]^{\TT^n} u_\varepsilon - \lambda u_\varepsilon \Vert \\ \le \Vert[\Pi_{k, \alpha_k}]^{\TT^n} B_k([\Pi_{k, \alpha_k}]^{\TT^n} u_\varepsilon - u_\varepsilon) \Vert + \Vert[\Pi_{k,\alpha_k}]^{\TT^n} B_k u_\varepsilon - \lambda u_\varepsilon \Vert \\ \le B \Vert[\Pi_{k, \alpha_k}]^{\TT^n} u_\varepsilon - u_\varepsilon \Vert + \Vert[\Pi_{k,\alpha_k}]^{\TT^n} B_k u_\varepsilon - \lambda u_{\varepsilon} \Vert.\end{multline*}
Next, \begin{multline*} \Vert[\Pi_{k,\alpha_k}]^{\TT^n} B_k u_\varepsilon - \lambda u_{\varepsilon} \Vert \\ \le \Vert[\Pi_{k,\alpha_k}]^{\TT^n}(B_k - \mcm_{\mathbbm 1_{\{v > b\}}})u_\varepsilon \Vert + \Vert[\Pi_{k,\alpha_k}]^{\TT^n} \mcm_{\mathbbm 1_{\{v > b\}}} u_\varepsilon - \lambda u_\varepsilon \Vert \le\\ \Vert(B_k - \mcm_{\mathbbm 1_{\{v >b \}}}) u_\varepsilon \Vert + \Vert \Pi^n_N \mcm_{\mathbbm 1_{\{v > b\}}} u_\varepsilon - \lambda u_\varepsilon \Vert + \Vert([\Pi_{k,\alpha_k}]^{\TT^n} - \Pi^n_N) \mcm_{\mathbbm 1_{\{v > b\}}} u_\varepsilon \Vert.\end{multline*}
Also, using the decomposition $$\mcm_{\mathbbm 1_{\{v > b\}}} u_\varepsilon = (\Id_{L^2(\TT^n)} - \Pi^n + (\Pi^n - \Pi^n_N) + \Pi^n_N) \mcm_{\mathbbm 1_{\{v > b\}}} u_\varepsilon,$$
and noting that $\Pi^n_N(\Id_{L^2(\TT^n)} - \Pi^n) = 0 = \Pi^n_N(\Pi^n - \Pi^n_N)$, we see that 
\begin{multline*} \Vert([\Pi_{k,\alpha_k}]^{\TT^n} - \Pi^n_N) \mcm_{\mathbbm 1_{\{v > b\}}} u_\varepsilon \Vert \le \\ \Vert [\Pi_{k,\alpha_k}]^{\TT^n}(\Id_{L^2(\TT^n)} - \Pi^n) \mcm_{\mathbbm 1_{\{v >b\}}} u_\varepsilon \Vert +\\ \Vert [\Pi_{k,\alpha_k}]^{\TT^n}(\Pi^n - \Pi^n_N) \mcm_{\mathbbm 1_{\{v > b\}}} u_\varepsilon\Vert +\\ \Vert([\Pi_{k,\alpha_k}]^{\TT^n} - \Pi^n_N) \Pi^n_N \mcm_{\mathbbm 1_{\{v >b\}}} u_\varepsilon \Vert\\ \le \Vert [\Pi_{k,\alpha_k}]^{\TT^n}(\Id_{L^2(\TT^n)} - \Pi^n) \mcm_{\mathbbm 1_{\{v >b\}}} u_\varepsilon \Vert +\\ \Vert (\Pi^n - \Pi^n_N) \mcm_{\mathbbm 1_{\{v > b\}}} u_\varepsilon\Vert +\\ \Vert([\Pi_{k,\alpha_k}]^{\TT^n} - \Pi^n_N) \Pi^n_N \mcm_{\mathbbm 1_{\{v >b\}}} u_\varepsilon \Vert.\end{multline*}
Putting everything together,
\begin{multline*} \Vert [\Pi_{k,\alpha_k}]^{\TT^n} B_k [\Pi_{k, \alpha_k}]^{\TT^n} u_\varepsilon - \lambda u_\varepsilon \Vert \\ \le B \Vert[\Pi_{k, \alpha_k}]^{\TT^n} u_\varepsilon - u_\varepsilon \Vert +\\ \Vert(B_k - \mcm_{\mathbbm 1_{\{v >b \}}}) u_\varepsilon \Vert + \Vert \Pi^n_N \mcm_{\mathbbm 1_{\{v > b\}}} u_\varepsilon - \lambda u_\varepsilon \Vert +\\\Vert [\Pi_{k,\alpha_k}]^{\TT^n}(\Id_{L^2(\TT^n)} - \Pi^n) \mcm_{\mathbbm 1_{\{v >b\}}} u_\varepsilon \Vert +\\ \Vert (\Pi^n - \Pi^n_N) \mcm_{\mathbbm 1_{\{v > b\}}} u_\varepsilon\Vert +\\ \Vert([\Pi_{k,\alpha_k}]^{\TT^n} - \Pi^n_N) \Pi^n_N \mcm_{\mathbbm 1_{\{v >b\}}} u_\varepsilon \Vert.\end{multline*}
Now, using (\ref{Rlambda_est}),
\begin{multline*} \Vert [\Pi_{k,\alpha_k}]^{\TT^n} B_k [\Pi_{k, \alpha_k}]^{\TT^n} u_\varepsilon - \lambda u_\varepsilon \Vert \\ < 2\varepsilon +  B \Vert[\Pi_{k, \alpha_k}]^{\TT^n} u_\varepsilon - u_\varepsilon \Vert +\\ \Vert(B_k - \mcm_{\mathbbm 1_{\{v >b \}}}) u_\varepsilon \Vert  +\\ \Vert [\Pi_{k,\alpha_k}]^{\TT^n}(\Id_{L^2(\TT^n)} - \Pi^n) \mcm_{\mathbbm 1_{\{v >b\}}} u_\varepsilon \Vert +\\ \Vert([\Pi_{k,\alpha_k}]^{\TT^n} - \Pi^n_N) \Pi^n_N \mcm_{\mathbbm 1_{\{v >b\}}} u_\varepsilon \Vert.\end{multline*}
Finally, $$\lim_{k\to \infty}\Vert(B_k - \mcm_{\mathbbm 1_{\{v >b \}}}) u_\varepsilon \Vert = 0$$
because $B_k \sto \mcm_{\mathbbm 1_{\{v > b\}}}$, and $$\lim_{k\to \infty} \Vert[\Pi_{k, \alpha_k}]^{\TT^n} u_\varepsilon - u_\varepsilon \Vert = \lim_{k\to \infty} \Vert([\Pi_{k,\alpha_k}]^{\TT^n} - \Pi^n_N) \Pi^n_N \mcm_{\mathbbm 1_{\{v >b\}}} u_\varepsilon \Vert=0$$
in light of (\ref{nN_str_conv}), and $$\lim_{k \to \infty}  \Vert [\Pi_{k,\alpha_k}]^{\TT^n}(\Id_{L^2(\TT^n)} - \Pi^n) \mcm_{\mathbbm 1_{\{v >b\}}} u_\varepsilon \Vert =0$$
in light of (\ref{nperp_str_conv}), so we conclude that for every sufficiently large $k$, $$ \Vert [\Pi_{k,\alpha_k}]^{\TT^n} B_k [\Pi_{k, \alpha_k}]^{\TT^n} u_\varepsilon - \lambda u_\varepsilon \Vert < 3\varepsilon.$$
This means that for every sufficiently large $k$, $$\dist(\lambda, \sigma([\Pi_{k,\alpha_k}]^{\TT^n} B_k [\Pi_{k,\alpha_k}]^{\TT^n})) < 3\varepsilon.$$
Thus, there exists $\lambda_k \in \sigma([\Pi_{k,\alpha_k}]^{\TT^n} B_k [\Pi_{k,\alpha_k}]^{\TT^n})$ such that $|\lambda_k - \lambda| < 3\varepsilon$, and assuming that $3\varepsilon < \delta$, we obtain the required.

\end{proof}
\begin{corollary} For $n >1$, Theorem \ref{main_proj_thm} is a quick consequence\footnote{For any $N \in \ZZ$, $N > 0$, we can choose $N$ disjoint open sub-intervals of $[t, t']$, and for any sufficiently large $k$, all the sub-intervals will contain elements of $\sigma(\Pi_{k,\alpha_k} \Pi_{k, g, b_0} \Pi_{k,\alpha_k})$.} of Proposition \ref{main_proj_prop}, with $v = g_0\big{|}_{\Lambda_{a_0}}$, $b = b_0$ (and $\alpha_k = c \cdot a_0$). \end{corollary}
Corollary \ref{max_cor} follows from Theorem \ref{main_proj_thm} using general results from the theory of two projections (\cite{bottspitk, shabtai2}).
\begin{proof}[Proof of Corollary \ref{max_cor}] Fix $0 \ne \phi \in \mca$. Then there exists a constant $B_\phi>0$ such that if $\mch$ is any complex separable Hilbert space and $P_1, P_2 : \mch \to \mch$ are orthogonal projections, then $$\Vert \phi(P_1, P_2) \Vert_{\op} \le B_\phi.$$
Moreover, $\phi \in \mca$ determines a continuous function $\beta_\phi :[0,1] \to [0,\infty)$ such that $$B_\phi = \max_{[0,1]} \beta_\phi,$$
and $$\Vert \phi(P_1, P_2) \Vert_{\op} \ge \max_{\sigma(P_1 P_2 P_1)} \beta_\phi.$$
Finally, by the hypothesis of Corollary \ref{max_cor} (or otherwise, in the setting of Proposition \ref{main_proj_prop}), it holds that the spectrum of $\Pi^c_{k} \Pi_{k,g,b_0} \Pi^c_{k}$ becomes dense in the interval $[0,1]$, hence $$\lim_{k \to\infty} \max_{\sigma(\Pi^c_{k} \Pi_{k,g,b_0} \Pi^c_{k})}\beta_\phi = B_\phi.$$
Noting that $$B_\phi \ge \Vert \phi(\Pi^c_{k}, \Pi_{k,g,b_0})\Vert_{\op} \ge \max_{\sigma(\Pi^c_{k} \Pi_{k,g,b_0} \Pi^c_{k})}\beta_\phi,$$
we obtain the required. \end{proof}
\subsection{Proof of Proposition \ref{two_proj_prop}}
Berezin-Toeplitz quantization satisfies the following trace correspondence property (see, e.g., \cite{charpol1}). If $h(\cdot, k) \in C^\infty(M)$ admits an asymptotic expansion $h(\cdot, k) \sim \sum_{l \ge 0} h_l k^{-l}$ in the $C^\infty$ topology, then $$\tr\big{(}T_k(h(\cdot,k))\big{)}= \Big{(}\frac k {2\pi} \Big{)}^n \int_M h_0 \mu + \mco(k^{n-1}).$$
In particular (taking $h= 1$), $$d_k = \dim(\mch_k) = \Big{(} \frac k {2\pi} \Big{)}^n\mu(M) + \mco(k^{n-1}).$$
The trace correspondence is also valid for a fixed, real-valued $h \in L^\infty(M)$. Namely, $$\tr(T_k(h)) = \Big{(}\frac k {2\pi} \Big{)}^n \int_M h \mu + \mco(k^{n-1}).$$
Recall also that if $h(\cdot, k) \in C^\infty(M)$ has the expansion $h(\cdot, k) \sim \sum_{l \ge 0} h_l k^{-l}$ and $\phi \in C^\infty(\RR)$, then (\cite{charles}, Proposition 12) $$\Vert \phi(T_k(h(\cdot, k))) - T_k(\phi \circ h_0) \Vert_{\op} = \mco(k^{-1}).$$

First, we show that for any $0 < t < \frac 1 2$, $$N_k\big{(}(t, 1-t)\big{)} = o(k^n).$$
Denote $$R_k = \Pi_{k,1} \Pi_{k,2} \Pi_{k,1},$$
and for an interval $I \subset [0,1]$ denote $\sigma_I(R_k) = \sigma(R_k) \cap I$.
We will see that $\tr(R_k - R_k^2) = o(k^n),$ which implies that $N_k\big{(}(t, 1-t)\big{)} = o(k^n)$. Indeed, if $m_k(\lambda)$ is the multiplicity of $\lambda \in \sigma(R_k)$, then \begin{multline*}\tr(R_k - R_k^2) = \sum_{\lambda \in \sigma(R_k)} m_k(\lambda) \lambda(1-\lambda) \ge\\ \sum_{\lambda \in  \sigma_{(t,1-t)}(R_k)} m_k(\lambda) t(1-t) = t(1-t)N_k\big{(}(t, 1-t)\big{)}.\end{multline*}
In what follows, we use the Schatten norms $$\Vert \cdot \Vert_q : \End(\mch_k) \to [0,\infty),\ \Vert Q \Vert_q = (\tr(|Q|^q))^{\frac 1 q},\ q = 1,2.$$
We note the inequalities \begin{equation}\label{schatten_ineq}\begin{aligned} & \Vert Q Q' \Vert_1 \le \Vert Q \Vert_{\op} \Vert Q' \Vert_1,\ \Vert Q Q' \Vert_1 \le \Vert Q \Vert_{1} \Vert Q' \Vert_{\op}\\ &\Vert Q\Vert_2^2 = \Vert |Q|^2 \Vert_1 \le \Vert Q \Vert_{\op} \Vert Q \Vert_1,\\ &\Vert Q \Vert_q \le (d_k)^{\frac 1 q} \Vert Q \Vert_{\op}.\end{aligned}\end{equation}
Note that for two orthogonal projections $P$, $Q$, it holds that $$|[P,Q]|^2 = ([P,Q])^*[P,Q] = PQP + QPQ - PQPQ-QPQP,$$
and so $$\Vert[P,Q]\Vert_2^2 = 2\tr(PQP-PQPQP).$$
Thus, in our present context, it suffices to prove that $$\Vert[\Pi_{k,1}, \Pi_{k,2}]\Vert_2^2 = o(k^n).$$

Fix $\delta > 0$, and choose smooth functions $\tau_{\alpha_j, \delta} : \RR \to [0,1]$, $j = 1,2$, such that $$\tau_{\alpha_j, \delta}\big{|}_{(\alpha_j, \infty)} = 1,\ \supp \tau_{\alpha_j, \delta} \subset (\alpha_j - \delta, \infty).$$
Write $$\Pi_{k,j} = A_{k,j}+ B_{k,j},\ A_{k,j} = \tau_{\alpha_j, \delta}(T_k(h_j(\cdot, k))).$$
Then $$\Vert[\Pi_{k,1}, \Pi_{k,2}]\Vert_2 \le \Vert[\Pi_{k,1},\Pi_{k,2}]- [A_{k,1}, A_{k,2}]\Vert_2 + \Vert[A_{k,1}, A_{k,2}]\Vert_2,$$
where (by (\ref{schatten_ineq})) $$\Vert[A_{k,1}, A_{k,2}]\Vert_2 \le d_k^{\frac 1 2} \Vert[A_{k,1}, A_{k,2}]\Vert_{\op} = \mco(k^{\frac n 2 - 1}).$$
Also, $$[\Pi_{k,1},\Pi_{k,2}]- [A_{k,1}, A_{k,2}]= [B_{k,1}, \Pi_{k,2}] + [A_{k,1}, B_{k,2}],$$
therefore $$\Vert[\Pi_{k,1}, \Pi_{k,2}] - [A_{k,1}, A_{k,2}]\Vert_2 \le \Vert [B_{k,1}, \Pi_{k,2}]\Vert_2 + \Vert[A_{k,1}, B_{k,2}]\Vert_2.$$
Generally, in light of (\ref{schatten_ineq}), we see that $$\Vert QR \Vert_2^2 = \Vert R^* Q^* Q R \Vert_1 \le \Vert R \Vert_{\op}^2 \Vert Q \Vert^2_2,\ \Vert QR \Vert_2^2 \le \Vert R \Vert^2_2 \Vert Q \Vert_{\op}^2,$$
hence in our case, $$\Vert [B_{k,1}, \Pi_{k,2}]\Vert_2 \le 2\Vert B_{k,1} \Vert_2,\ \Vert[A_{k,1}, B_{k,2}]\Vert_2 \le 2 \Vert B_{k,2} \Vert_2.$$
Also, noting that $\Vert B_{k,1} \Vert_{\op} \le 1$, $\Vert B_{k,2} \Vert_{\op} \le 1$, $$\Vert B_{k,1} \Vert_2 = \sqrt{\Vert |B_{k,1}|^2 \Vert_1} \le \sqrt{\Vert B_{k,1}\Vert_1},\ \Vert B_{k,2} \Vert_2 \le \sqrt{\Vert B_{k,2} \Vert_1}.$$
Thus, we obtain the bound \begin{equation}\label{comm_2_bd}\Vert[\Pi_{k,1}, \Pi_{k,2}]\Vert_2 \le  d_k^{\frac 1 2} \Vert[A_{k,1}, A_{k,2}]\Vert_{\op} + 2\sqrt{\Vert B_{k,1} \Vert_1} + 2 \sqrt{\Vert B_{k,2} \Vert_1}.\end{equation}

Fix $\chi_{\alpha_j,\delta} : \RR \to [0,1]$, $j =1,2$, such that $$\chi_{\alpha_j, \delta}\big{|}_{[\alpha_j -\delta, \alpha_j + \delta]} = 1,\ \supp \chi_{\alpha_j,\delta} \subset (\alpha_j - 2\delta, \alpha_j + 2 \delta).$$
Then $$0 \le \tau_{\alpha_j,\delta} - \mathbbm 1_{(\alpha_j, \infty)} \le \chi_{\alpha_j,\delta},$$
therefore \begin{multline*}\Vert B_{k,1} \Vert_1 \le \Vert \chi_{\alpha_1, \delta}(T_k(h_1(\cdot,k))) \Vert_1 = \Vert T_k(\chi_{\alpha_1, \delta}\circ h_{1,0}) \Vert_1 + \mco(k^{n-1})\\ \le \mu(\{|h_{1,0} - \alpha_1| < 2\delta\})\Big{(}\frac k {2\pi} \Big{)}^n + \mco(k^{n-1}).\end{multline*}
Similarly, $$\Vert B_{k,2} \Vert_1 \le \mu(\{|h_{2,0} - \alpha_2| < 2\delta\})\Big{(} \frac k {2\pi} \Big{)}^n + \mco(k^{n-1}).$$
Finally, recall that for $r_1, r_2, r_3 \in \RR$ it holds that $$(r_1 + r_2 + r_3)^2 \le 3(r_1^2 + r_2^2 + r_3^2).$$
Hence, applying the above to (\ref{comm_2_bd}), we obtain\begin{multline*}\Vert [\Pi_{k,1}, \Pi_{k,2}]\Vert_2^2 \le 3 d_k \Vert[A_{k,1}, A_{k,2}]\Vert_{\op}^2 +12 \Vert B_{k,1} \Vert_1 + 12 \Vert B_{k,2} \Vert_1\\ \le 12 \Big{(} \frac k {2\pi} \Big{)}^n \Big{(} \mu\big{(}\{|h_{1,0} - \alpha_1| < 2\delta\}\big{)} + \mu\big{(}\{|h_{2,0} - \alpha_2| < 2\delta\} \big{)} \Big{)} + \mco(k^{n-1}).\end{multline*}
Thus, continuity from above of $\mu$, for any $\varepsilon > 0$, there exists $\delta > 0$ such that for every sufficiently large $k$, $$\frac{\Vert[\Pi_{k,1}, \Pi_{k,2}]\Vert_2^2}{k^n} < \frac \varepsilon 2 + \mco(k^{-1}) < \varepsilon,$$
as required.

The estimate for $N_k([1-t, 1])$, $0 < t < \frac 1 2$, follows similarly. First, note that $$\Tr(R_k) = \Tr(\Pi_{k,1} \Pi_{k,2}),$$
and then note that \begin{multline*} \Big{|}\Tr(\Pi_{k,1} \Pi_{k,2}) -\Tr(A_{k,1} A_{k,2})\Big{|} = \Big{|}\Tr\Big{(}(\Pi_{k,1} - A_{k,1})\Pi_{k,2} + A_{k,1}(\Pi_{k,2} - A_{k,2}) \Big{)} \Big{|}\\ \le \Vert \Pi_{k,1} - A_{k,1} \Vert_1 + \Vert \Pi_{k,2} - A_{k,2} \Vert_1,\end{multline*}
and so for any $\varepsilon > 0$ we may choose $\delta > 0$ sufficiently small such that $$\frac 1 {k^n}\Big{|} \Tr(\Pi_{k,1} \Pi_{k,2}) - \Tr(A_{k,1} A_{k,2}) \Big{|} < \frac \varepsilon 4 + \mco(k^{-1}).$$
Also, $$\Tr(A_{k,1} A_{k,2}) = \Big{(} \frac k {2\pi} \Big{)}^n \int_M (\tau_{\alpha_1, \delta} \circ h_{1,0}) (\tau_{\alpha_2, \delta} \circ h_{2,0}) d\mu + \mco(k^{n-1}).$$
Using the regularity of $\mu$ and the fact that $\mu(\{h_{1,0} = \alpha_1\}) = 0 = \mu(\{h_{2,0} = \alpha_2\})$, we see that if $\delta > 0$ is small enough then $$\frac 1 {k^n} \Big{|}\Tr(A_{k,1} A_{k,2})- \Big{(}\frac k {2\pi} \Big{)}^n\mu(\{h_{1,0} > \alpha_1\} \cap \{h_{2,0} > \alpha_2\})\Big{|} < \frac \varepsilon 4 + \mco(k^{-1}).$$
Thus, we conclude that $$\Tr(R_k) = \Big{(} \frac k {2\pi} \Big{)}^n \mu(\{h_{1,0} > \alpha_1\}\cap \{h_{2,0} > \alpha_2\}) + o(k^n).$$
Also, $$\Tr(R_k) = \sum_{\lambda \in \sigma_{[0, 1-t)}(R_k)} m_k(\lambda) \lambda + \sum_{\lambda \in \sigma_{[1-t,1]}(R_k) } m_k(\lambda) \lambda,$$
where $$\sum_{\lambda \in \sigma_{[0,1-t)}(R_k) } m_k(\lambda) \lambda \le \frac 1 t \sum_{\lambda \in \sigma(R_k)} m_k(\lambda) \lambda(1-\lambda) = \frac 1 t \Tr(R_k - R_k^2) = o(k^n).$$
Thus, $$\sum_{\lambda \in \sigma_{[1-t,1]}(R_k)} m_k(\lambda) \lambda = \Tr(R_k) + o(k^n).$$
Also, \begin{multline*}0 \le N_k([1-t, 1]) - \sum_{\lambda \in \sigma_{[1-t,1]}(R_k)} m_k(\lambda) \lambda\\ = \sum_{\lambda \in \sigma_{[1-t,1]}(R_k)} m_k(\lambda)(1-\lambda)\le \frac 1 {1-t} \Tr(R_k - R_k^2) = o(k^n).\end{multline*}
Finally, we conclude that $$N_k([1-t, 1]) = \Big{(} \frac k {2\pi} \Big{)}^n \mu(\{h_{1,0} > \alpha_1\} \cap \{h_{2,0} > \alpha_2\}) + o(k^n),$$
and the estimate for $N_k([0,t])$ now follows almost immediately. Namely, by comparing $\Pi_{k,1}$ and $A_{k,1}$, we see that $$\dim(\Pi_{k,1} \mch_k) = \Tr(\Pi_{k,1}) = \Big{(} \frac k {2\pi} \Big{)}^n \mu(\{h_{1,0} > \alpha_1\}) + o(k^n).$$
Since \begin{multline*}\dim(\Pi_{k,1} \mch_k) = N_k([0,t]) + N_k\big{(}(t, 1-t)\big{)}+ N_k([1-t, 1]) \\= N_k([0,t]) + \Big{(}\frac k {2\pi} \Big{)}^n \mu(\{h_{1,0} > \alpha_1\} \cap \{h_{2,0} > \alpha_2\}) + o(k^n),\end{multline*}
we obtain the required. \hfill $\qedsymbol$
\section{Proof of Theorem \ref{contr_thm}}
Recall that we denote $$C^\infty_{\inv}(M) = \{g \in C^\infty(M) \ | \ \{g, f^j_{0}\} = 0,\ j = 1,...,n\},$$
where $$\{\cdot, \cdot\} : C^\infty(M) \times C^\infty(M) \to C^\infty(M)$$
is the Poisson bracket. Consider a Lie subalgebra $$\mfg = (V_{0} \oplus V_{\inv}, \{\cdot, \cdot\}) \subset (C^\infty(M), \{\cdot,\cdot\}),$$ where $V_{0} \oplus V_{\inv}$ is a direct sum of vector spaces, $$(V_{\inv}, \{\cdot, \cdot\}) \subset (C^\infty_{\inv}(M), \{\cdot, \cdot\})$$ is a Lie subalgebra, and for all $g \in V_{0}$, $h \in V_{\inv}$, it holds that $\{g, h\} \in V_{0}$.

Fix $0 \ne R \in \RR$, and define the linear maps $$t_{k,R} : \mfg \to \mfg,\ t_{k,R} \big{|}_{V_0} = \frac R k \Id,\ t_{k,R} \big{|}_{V_{\inv}} = \Id.$$
A straightforward computation shows that $\mfg$ contracts via $t_{k,R}$ to $$\mfg_\infty = (V_{\inv}, \{\cdot, \cdot\}) \ltimes (V_0, 0),$$
where $(V_0, 0)$ is the commutative Lie algebra with underlying vector space $V_0$.

Let $\xi_k,\xi_\infty^R : \mfg \to \End(C^\infty(\TT^n))$ be as specified in Theorem \ref{contr_thm}. It is straightforward to verify that $\xi_\infty^R$ is a genuine Lie algebra representation of $\mfg_\infty$. The fact that the sequence $\xi_k$ contracts to $\xi_\infty^R$ via the maps $t_{k,R}$ is a quick corollary of Theorem \ref{str_conv_thm}, since $t_{k,R}$ "improves" the asymptotic behaviour of $\xi_k$. The more demanding task in the proof of Theorem \ref{contr_thm} is to establish that $\xi_k$ is a semiclassical projective quasi-representation as in Definition \ref{spqr}. We begin with some preliminary results.
\begin{lemma}\label{inv_bdd} Let $u \in C^\infty(\TT^n)$. Let $X_g$ denote the Hamiltonian vector field of $g\in C^\infty_{\inv}(M)$, and write $$X_g\big{|}_{\Lambda_{a_0}} = \sum_{j=1}^n \zeta_j(g) X_{f_0^j}\big{|}_{\Lambda_{a_0}}.$$
Recall that $$\xi_k(g) = ik \sum_{j=1}^n \zeta_j(g)\big{(}[T_k^j]^{\TT^n} - a_{k,0}^j \Id_{L^2(\TT^n)}\big{)}.$$
Then  $$\lim_{k \to \infty} \xi_k(g)u =  X_g\big{|}_{\Lambda_{a_0}}u.$$\end{lemma}
\begin{proof}
We assume without loss of generality that in the coordinates (\ref{al_sympl}), the restriction of the Hamiltonian vector field of $f_0^j$ to $\Lambda_{a_0}$ equals $\partial_{\theta^j}$, $j = 1,...,n$ (otherwise, we replace $f_0^j$ and $[T_k^j]^{\TT^n} - a_{k,0}^j\Id_{L^2(\TT^n)}$ by matching linear combinations determined by the rows of $\nu(0)^{-1}$), and prove that $$\lim_{k \to \infty} ik([T_k^j]^{\TT^n} - a_{k,0}^j \Id_{L^2(\TT^n)})u = \partial_{\theta^j} u.$$
The general case follows by linearity.

In what follows, we use Proposition \ref{spec_properties} with the specific choice $$\mci_k = \{m \in \ZZ^n \ | \ |m^j| \le \lfloor k^{\frac 1 4} \rfloor,\ j=1,...,n\},$$
and denote $$V_k = \Sp\{e_m \ | \ m \in \mci_k\} \subset L^2(\mathbb T^n).$$
Write $$u = u_k + u_k^\perp,\ u_k \in V_k,\ u_k^\perp \in V_k^\perp.$$
Then $\Vert u_k^\perp\Vert = \mco(k^{-\infty})$, since $u \in C^\infty(\TT^n)$. Denote $$A_k^j = [T_k^j]^{\TT^n} - a_{k,0}^j \Id_{L^2(\TT^n)}.$$
Then $\Vert A^j_{k} \Vert_{\op} = \mco(1)$,
hence it suffices to show that $$\lim_{k \to \infty} ik A_k^j u = \partial_{\theta^j}u.$$
In light of Proposition \ref{spec_properties}, and the definition of the assignment $$[\ \cdot \ ]^{\TT^n} : \End(\mch_k) \to \End(L^2(\TT^n)),$$ there exist $q_{k,m}^j \in \RR$, $m \in \mci_k$, such that $$A_k^j e_m = \Big{(}\frac{2\pi m^j} k + q_{k,m}^j \Big{)} e_m,$$
and moreover there exists $C> 0$ such that $$|q_{k,m}^j| \le C \frac{|m|^2}{k^2}.$$
Thus,
\begin{multline*} ik A_k^j u_k = ik \sum_{m \in \mci_k} \widehat u(m) \Big{(}\frac{2\pi m^j} k + q_{k,m}^j\Big{)} e_m \\ = \sum_{m \in \mci_k} 2\pi i m^j \widehat{u}(m) e_m + ik \sum_{m \in \mci_k} \widehat{u}(m) q_{k,m}^j e_m,\end{multline*}
where $$\lim_{k \to \infty} \sum_{m \in \mci_k} 2\pi i m^j \widehat{u}(m) e_m = \partial_{\theta^j} u,$$
and $$\Vert ik \sum_{m \in \mci_k} \widehat{u}(m) q_{k,m}^j e_m \Vert \le \frac 1 k \sum_{m \in \mci_k} |\widehat{u}(m)| |m|^2 = O(k^{-1}),$$
as required.
\end{proof}
We also require the following.
\begin{lemma}\label{inv_zero_decay} Let $u \in C^\infty(\TT^n)$. Assume that $g \in C^\infty_{\inv}(M)$ satisfies $g\big{|}_{\Lambda_{a_0}} = 0$. Then $$[T_k(g)]^{\TT^n} u = O(k^{-1}).$$\end{lemma}
\begin{proof}
Let $\mci_k = \{m \in \ZZ^n \ | \ |m^j| \le \lfloor k^{\frac 1 4} \rfloor,\ j = 1,...,n\}$. Then there exists $C>0$ such that for every sufficiently large $k$, for every $m \in \mci_k$, it holds that \begin{equation}\label{uniform_m_est}\Vert [T_k(g)]^{\TT^n}e_m \Vert < \frac C k (1+|m|).\end{equation}
Indeed, let $s_{k,m}$ be the joint eigensection of $T_k^1, ..., T_k^n$ corresponding to $a_{k,m}$, as in Theorem \ref{str_conv_thm}. Then by \cite{charles_bohr}, Proposition 3.8, there exist $c_{k,m} \in \CC$ and a Lagrangian section $(\hat s_{k,m}, a_{k,m}, k)$, where $\hat s_{k,m}$ is associated with $\Lambda_{a_{k,m}}$, such that $$s_{k,m} = c_{k,m} \hat s_{k,m} + \mco(k^{-\infty}).$$
Here we group the pairs $k, m \in \mci_k$ together to form a single sequence $p\in \NN$, and view $(\hat s_{k,m}, a_{k,m}, k)$ as a single Lagrangian section with symbol as specified in \cite{charles_bohr}, Proposition 3.8. Then $c_{k,m}$ is uniformly bounded with respect to $k$, $m\in \mci_k$ (this may be shown as in Lemma \ref{c_k_abs}, noting that $a_{k,m}, m \in \mci_k$ belong to shrinking neighborhoods of $a_0$), and the estimate $\mco(k^{-\infty})$ is similarly uniform over $m \in \mci_k$. Now (\cite{charles_bohr}, Proposition 2.7), $$(T_k(g) \hat s_{k,m}, a_{k,m}, k)$$ is a Lagrangian section with symbol (see Sect. \ref{lagr_sect}) $$g v + \mco(\hbar),$$ where $v$ is the symbol of $(\hat s_{k,m}, a_{k,m}, k)$. Note that $g(a) = g|_{\Lambda_a}$ is constant for all $a$ in a small open neighborhood of $a_0$ (because $g \in C^\infty_{\inv}(M)$).  This implies that $$ T_k(g) \hat s_{k,m} = g( a_{k,m}) \hat s_{k,m} + \mco(k^{-1}),$$
and the $\mco(k^{-1})$ estimate is uniform over $m \in \mci_k$. Now, by Proposition \ref{spec_properties}, there exist $r_k, q_{k,m} \in \RR^n$ such that $$a_{k,m} = a_0 + r_k +\frac{2\pi} k\nu(0) m +  q_{k,m},$$
where $r_k = \mco(k^{-1})$ and there exists $\tilde C > 0$ such that for all sufficiently large $k$, for every $m \in \mci_k$, $$|q_{k,m}| \le \tilde C \frac{|m|^2}{k^2}.$$
Thus, there exists $\tilde C_1 > 0$ such that for every sufficiently large $k$, for every $m \in \mci_k$, $$|g(a_{k,m})| \le \tilde C_1 \left(\frac{1+|m|}k\right).$$
This establishes (\ref{uniform_m_est}).
Finally, set $$V_k = \Sp\{e_m \ | \ m \in \mci_k\},$$
and write $$u = u_k + u_k^\perp,\ u_k \in V_k,\ u_k^\perp \in V_k^\perp.$$
Then $\Vert u_k^\perp \Vert = \mco(k^{-\infty})$, and $$ \Vert [T_k(g)]^{\TT^n} u_k\Vert \le  \sum_{m \in \mci_k} |\widehat{u}(m)| \Vert [T_k(g)]^{\TT^n} e_m \Vert  \le \frac C k \sum_{m \in \mci_k} |\widehat{u}(m)|(1+|m|),$$
as required.
\end{proof}
The quasi-multiplicativity of the maps $T_k : C^\infty(M) \to \End(\mch_k)$ readily implies the following.
\begin{corollary}\label{prod_inv_zero_decay} Let $g \in C^\infty(M)$ and $u \in C^\infty(\TT^n)$. Assume that there exists an open neighborhood $U_0$ of $\Lambda_{a_0}$ such that $$g|_{U_0} = (g_1 g_2)\big{|}_{U_0},\ g_1 \in C^\infty_{\inv}(M),\ g_2 \in C^\infty(M),$$
and $g_1\big{|}_{\Lambda_{a_0}} = 0$. Then $[T_k(g)]^{\TT^n} u = \mco(k^{-1})$. \end{corollary}
\begin{proof}
Let $$\tau\in C_{\inv}^\infty(M),\ \tau : M \to [0,1]$$ be compactly supported in $U_0$, and assume that $\tau = 1$ in an open neighborhood of $\Lambda_{a_0}$. Then $$[T_k(g)]^{\TT^n} u = [T_k((1-\tau) g)]^{\TT^n} u + [T_k(\tau g)]^{\TT^n} u.$$
Let $\mci_k = \{m \in \ZZ^n \ | \ |m^j| \le \lfloor k^{\frac 1 4} \rfloor, j =1,...,n\}$, and write $$V_k = \Sp\{e_m \ | \ m \in \mci_k\}.$$ Then $$u = u_k + u_k^\perp,\ u_k \in V_k,\ u_k^\perp \in V_k^{\perp},$$
and note that $\Vert u_k^\perp \Vert = \mco(k^{-\infty})$ since $u \in C^\infty(\TT^n)$, and that $$\Vert [T_k((1-\tau) g)]^{\TT^n} u_k \Vert \le \sum_{m \in \mci_k} |\widehat{u}(m)| \Vert T_k((1-\tau) g) s_{k,m} \Vert.$$
However, since $s_{k,m}$ is a joint eigensection of $T_k^1,...,T_k^n$ associated with joint eigenvalue $a_{k,m} = a_0 + o(1)$, it holds that (see Lemma \ref{lagr_ms})\footnote{The microsupport of a sequence $s_p \in \mch_{k_p}$ is defined in the same way as in Sect. \ref{MS_sect}, replacing $s_k, k$ with $s_p, k_p$.} $$\MS(s_{k,m}) \subset \Lambda_{a_0},$$
and so $$\MS(T_k((1-\tau) g) s_{k,m}) = \emptyset,$$
that is, $$\Vert T_k((1-\tau) g) s_{k,m} \Vert = \mco(k^{-\infty})$$
uniformly for $m \in \mci_k$ (this can also be verified by direct computation). Thus, $$[T_k(g)]^{\TT^n} u = [T_k(\tau g)]^{\TT^n} u + O(k^{-\infty}).$$
Finally, noting that $\tau g = \tau g_1 g_2$ and using the quasi-multiplicativity property of $T_k : C^\infty(M) \to \End(\mch_k)$, Lemma \ref{inv_zero_decay}, and the fact that $\Vert[T_k(\tau g_2)]^{\TT^n} \Vert_{\op} = \mco(1)$ (by norm correspondence (\ref{norm_corr})), we conclude that  $$[T_k(\tau g)]^{\TT^n} u = \big{(}[T_k(\tau g_2)]^{\TT^n} [T_k(g_1)]^{\TT^n} + \mco(k^{-1})\big{)} u = \mco(k^{-1}),$$
as needed.
\end{proof}
We can now verify (separately) that $\xi_k$ is a semiclassical projective quasi-representation which contracts (via $t_{k,R}$) to $\xi_\infty^R$.
\begin{lemma} Let $P_{\inv} : V_0 \oplus V_{\inv} \to V_{\inv}$ denote the projection, and define the functional $$l : V_0 \oplus V_{\inv} \to \RR,\ l(g) =(P_{\inv} g)|_{\Lambda_{a_0}},$$
where the latter is viewed as a constant. Set $c(g,h) = l(\{g, h\})$. Then for every $g, h \in V_0 \oplus V_{\inv}$ and $u \in C^\infty(\TT^n)$, $$\Vert ([\xi_k(g), \xi_k(h)] - \xi_k(\{g, h\}) - ik c(g,h) \Id_{L^2(\TT^n)}) u \Vert = \mco(1).$$ \end{lemma}
\begin{proof}
We may prove the cases $g, h \in V_0$, $g \in V_0, h \in V_{\inv}$, $g, h \in V_{\inv}$ separately (by linearity).

If $g, h \in V_0$ then $$\xi_k(g) = ik [T_k(g)]^{\TT^n},\ \xi_k(h) = ik [T_k(h)]^{\TT^n}.$$Then by the correspondence principle (\ref{corr_princ}), and using the multiplicativity of the map $$[\ \cdot \ ]^{\TT^n} : \End(\mch_k) \to \End(L^2(\TT^n)),$$
we see that $$[\xi_k(g), \xi_k(h)] = -k^2 \big{[}[T_k(g)]^{\TT^n}, [T_k(h)]^{\TT^n}\big{]} = ik[T_k(\{g, h\})]^{\TT^n} +\mco(1).$$
Writing $$\{g,h\}_{\inv} = P_{\inv} \{g, h\},\ \{g, h\}_0 = \{g, h\} - \{g, h\}_{\inv},$$
it holds that $$\xi_k(\{g, h\}) = ik [T_k(\{g, h\}_0)]^{\TT^n} + \xi_k( \{g,h\}_{\inv}),$$
therefore \begin{multline*}[\xi_k(g), \xi_k(h)] - \xi_k(\{g, h\}) - i k (\{g, h\}_{\inv})|_{\Lambda_{a_0}} \Id_{L^2(\TT^n)}\\ = ik [T_k( \{g, h\}_0)]^{\TT^n} + ik [T_k(\{g, h\}_{\inv})]^{\TT^n} -ik [T_k( \{g, h\}_0)]^{\TT^n}\\ - \xi_k(\{g, h\}_{\inv}) - ik(\{g, h\}_{\inv})|_{\Lambda_{a_0}} \Id_{L^2(\TT^n)} + \mco(1)\\ = ik [T_k(\{g, h\}_{\inv})]^{\TT^n} -\xi_k(\{g, h\}_{\inv}) -ik (\{g, h\}_{\inv})|_{\Lambda_{a_0}} \Id_{L^2(\TT^n)} + \mco(1).\end{multline*}
In light of Lemma \ref{inv_bdd}, $\xi_k(\{g, h\}_{\inv}) u$ converges, and in particular $$\Vert \xi_k(\{g, h\}_{\inv}) u \Vert = \mco(1).$$
Now, writing $$u = u_k + u_k^{\perp},\ u_k \in \mch_k^{\TT^n},\ u_k^\perp \in (\mch_k^{\TT^n})^{\perp},$$
we note that $\Vert u_k^\perp \Vert = \mco(k^{-\infty})$. Thus, it only remains to show that $$\Vert ik\big{(} [T_k(\{g, h\}_{\inv})]^{\TT^n} - (\{g, h\}_{\inv})|_{\Lambda_{a_0}} \Id_{L^2(\TT^n)} \big{)} u_k \Vert = \mco(1).$$
Indeed, \begin{multline*}  ik\big{(} [T_k(\{g, h\}_{\inv})]^{\TT^n} - (\{g, h\}_{\inv})|_{\Lambda_{a_0}} \Id_{L^2(\TT^n)} \big{)} u_k\\ = ik \big{[}T_k\big{(}\{g, h\}_{\inv} - (\{g, h\}_{\inv})\big{|}_{\Lambda_{a_0}} \big{)}\big{]}^{\TT^n} u_k \\= ik \big{[}T_k\big{(}\{g, h\}_{\inv} - (\{g, h\}_{\inv})\big{|}_{\Lambda_{a_0}} \big{)}\big{]}^{\TT^n} u = \mco(1),\end{multline*}
where the last estimate follows from Lemma \ref{inv_zero_decay}. 

Next, assume $g \in V_0$, $h \in V_{\inv}$. Then $\{g, h\} \in V_0$, and so $\{g,h\}_{\inv} = 0$. Also, \begin{multline*}[\xi_k(g), \xi_k(h)] = \big{[}ik [T_k(g)]^{\TT^n}, ik\sum_{j=1}^n \zeta_j(h)([T_k^j]^{\TT^n} - a_{k,0}^j \Id_{L^2(\TT^n)})\big{]} \\ = -k^2 \sum_{j=1}^n \zeta_j(h) \big{[}[T_k(g)]^{\TT^n}, [T_k^j]^{\TT^n} \big{]} = ik\Big{[} T_k\Big{(} \{g, \sum_{j=1}^n \zeta_j(h) f_0^j\} \Big{)}\Big{]}^{\TT^n} + \mco(1).\end{multline*}
Thus, it suffices to show that $$ik \Big{[} T_k\Big{(}\{g, \sum_{j=1}^n \zeta_j(h) f_0^j - h\}\Big{)} \Big{]}^{\TT^n} u = \mco(1).$$
This holds by Corollary \ref{prod_inv_zero_decay}. Indeed, in the coordinates of (\ref{al_sympl}),\ $$X_h\big{|}_U = \sum_{j=1}^n \zeta_j(h)(y) X_j\big{|}_U,$$
where $X_j$ is the Hamiltonian vector field of $f_0^j$. Thus, $$\{g, \sum_{j=1}^n \zeta_j(h)f_0^j -h\}\big{|}_U(y,\theta) = \sum_{j=1}^n(\zeta_j(h)(y) - \zeta_j(h)(0)) X_j(g)\big{|}_U(y,\theta),$$
and we see that there exists an open neighborhood $U_0$ of $\Lambda_{a_0}$ such that the function $$H = \{g, \sum_{j=1}^n \zeta_j(h)f_0^j - h\}$$
satisfies $$H\big{|}_{U_0} =\sum_{j=1}^n (H_{1,j} H_{2,j})\big{|}_{U_0},\ H_{1,j} \in C_{\inv}^\infty(M),\ H_{1,j}\big{|}_{\Lambda_{a_0}} = 0,\ H_{2,j} \in C^\infty(M).$$
Finally, if $g, h \in V_{\inv}$, then $$[\xi_k(g),\xi_k(h)] = 0,$$
and there exists an open neighborhood $U_0$ of $\Lambda_{a_0}$ such that $$\{g, h\}\big{|}_{U_0} = 0,$$
which means that $\xi_k(\{g, h\}) = 0$. 
\end{proof}
Finally, we show that $\xi_k$ contracts to $\xi_\infty^R$ via $t_{k,R}$.
\begin{lemma} Let $g,h \in \mfg$ and $u \in C^\infty(\TT^n)$. Then $$\begin{aligned} &\lim_{k\to \infty} \xi_k\circ t_{k,R} (g) u = \xi_\infty^R(g) u,\\ &\lim_{k \to \infty} [\xi_k\circ t_{k,R}(g), \xi_k\circ t_{k,R}(h)] u = \xi_\infty^R(\{g, h\}_\infty)u.\end{aligned}$$ \end{lemma}
\begin{proof}
Write $\xi_k^R = \xi_k \circ t_{k,R}$, and identify $\Lambda_{a_0} \simeq \TT^n$. If $g \in V_0$ then $$\xi_k^R(g) = iR [T_k(g)]^{\TT^n},$$
hence $$\lim_{k \to \infty} \xi_k^R(g) u = iR \mcm_{g|_{\Lambda_{a_0}}} u =\xi_\infty^R(g) u$$
by Theorem \ref{str_conv_thm}. If $g \in V_{\inv}$ then $t_{k,R}(g) = g$. Hence, by Lemma \ref{inv_bdd}, $$\lim_{k \to \infty} \xi_k^R(g) u = X_g\big{|}_{\Lambda_{a_0}} u.$$
Next, if $g, h \in V_0$ then $$[\xi_k^R(g), \xi_k^R(h)] = -R^2\big{[}[T_k(g)]^{\TT^n}, [T_k(h)]^{\TT^n} \big{]} = \frac{iR^2} k [T_k(\{g, h\})]^{\TT^n} + \mco(k^{-2}),$$
therefore, noting that $\{g, h\}_\infty = 0$, $$\lim_{k \to \infty} [\xi_k^R(g), \xi_k^R(h)] u = 0 =\xi_\infty^R(\{g, h\}_\infty)u.$$
If $g \in V_0$, $h \in V_{\inv}$ then \begin{multline*} [\xi_k^R(g), \xi_k^R(h)] = \frac R k [\xi_k(g), \xi_k(h)]\\ = -kR \big{[}[T_k(g)]^{\TT^n}, \sum_{j=1}^n \zeta_j(h)\big{(}[T_k^j]^{\TT^n} - a_{k,0}^j \Id_{L^2(\TT^n)}\big{)}\big{]}\\ = -kR \sum_{j=1}^n \zeta_j(h)\big{[}[T_k(g)]^{\TT^n}, [T_k^j]^{\TT^n}\big{]}\\ = -kR \sum_{j=1}^n \zeta_j(h) \big{(}-\frac i k [T_k(\{g, f_0^j\})]^{\TT^n} + \mco(k^{-2})\big{)}\\ = i R\Big{[} T_k\Big{(}\{g, \sum_{j=1}^n \zeta_j(h) f_0^j\}\Big{)}\Big{]}^{\TT^n} + \mco(k^{-1}).\end{multline*}
Thus, noting that $$\{g, \sum_{j=1}^n \zeta_j(h) f_0^j\}\big{|}_{\Lambda_{a_0}} = \{g, h\}\big{|}_{\Lambda_{a_0}},$$
we see (using Theorem \ref{str_conv_thm}) that $$\lim_{k\to \infty} [\xi_k^R(g), \xi_k^R(h)] u = i R \mcm_{\{g, h\}\big{|}_{\Lambda_{a_0}}} u = \xi_\infty^R(\{g, h\}_\infty) u.$$
Finally, if $g, h \in V_{\inv}$ then $$[\xi_k^R(g), \xi_k^R(h)] = 0 = \xi_\infty^R(\{g, h\}_\infty).$$
The general case $g, h \in \mfg$ follows by linearity.
\end{proof} 

\textbf{Acknowledgements.} 
The work presented here traces its origins to my graduate research on pairs of spectral projections of spin operators, begun as an M.Sc. project under the supervision of Leonid Polterovich and Yohann Le Floch and continued during my Ph.D. studies under the supervision of Leonid Polterovich and Lev Buhovsky. I am especially grateful to Leonid Polterovich for suggesting this direction of research, and to all three of my supervisors for their guidance during my graduate studies. I also thank David Kazhdan and Laurent Charles for important suggestions, and Alejandro Uribe and Yael Karshon for valuable discussions.

\end{document}